\documentclass[11pt]{article}

\usepackage[T1]{fontenc}
\usepackage[utf8]{inputenc}
\usepackage{geometry}
\usepackage{amsmath}
\usepackage{amssymb}
\usepackage{amsthm}

\usepackage{graphicx}
\usepackage{float}
\usepackage{booktabs}
\usepackage{longtable}
\usepackage[font=small]{caption}
\usepackage{subcaption}
\usepackage{xcolor}

\usepackage[authoryear]{natbib}
\PassOptionsToPackage{hyphens}{url}
\usepackage[colorlinks=true,linkcolor=red,citecolor=blue,urlcolor=magenta,breaklinks=true]{hyperref}
\usepackage{enumitem}

\newcommand{\argmin}{\mathop{\mathrm{argmin}}}
\newcommand{\argmax}{\mathop{\mathrm{argmax}}}
\newcommand{\norm}[1]{||#1||}

\makeatletter
\newcommand\blfootnote[1]{%
  \begingroup
  \renewcommand\thefootnote{}\renewcommand\@makefnmark{}%
  \footnotetext{#1}%
  \endgroup
}
\makeatother

\theoremstyle{plain}
\newtheorem{theorem}{\textbf{Theorem}}
\newtheorem{proposition}[theorem]{\textbf{Proposition}}
\newtheorem{lemma}[theorem]{\textbf{Lemma}}

\theoremstyle{definition}
\newtheorem{definition}{\textbf{Definition}}

\begin{document}

\title{Oblivious Learning and Collusive Pricing}

\author{Yuhang Wu\textsuperscript{$\dagger$} \qquad Assaf Zeevi\textsuperscript{$\ddagger$}}

\date{This version: October 1, 2026}

\maketitle
\blfootnote{$\dagger$~Graduate School of Business, Columbia University.}
\blfootnote{$\ddagger$~Graduate School of Business and Data Science Institute, Columbia University.}
\vspace{-2em}

\begin{abstract}
On a platform with many sellers, should a pricing algorithm explicitly model competitors' prices when learning demand? Classical arguments suggest that ignoring competitors induces model misspecification and inefficiency, yet findings from {\it algorithmic collusion} suggest that ignoring competitor prices may, surprisingly,  facilitate collusive outcomes and improve profits. We study this problem in a competitive market with unknown noisy demand, in which sellers repeatedly set prices, either incorporating competitor prices in learning  their demand models (\emph{informed}),  or ignoring them (\emph{oblivious}). We show that, relative to a monopolist, an oblivious seller in a competitive market must conduct more aggressive price exploration to compensate for the loss of dynamic competitor information. When all sellers are oblivious, prices converge to the competitive outcome under persistent exploration, while a continuum of pseudo-equilibria arises when exploration is ``insufficient.''  In markets with a mix of oblivious and informed sellers, the informed strictly out-earn the oblivious. In game-theoretic terms,  the unique Nash equilibrium is the all-informed market, in which prices converge to the competitive outcome efficiently, and  oblivious modeling does not robustly lead to collusive patterns.  
\end{abstract}


\section{Introduction}\label{sec:intro}
\subsection{Background}
The rising popularity of learning algorithms and the advent of AI have significantly expanded the scope of algorithmic decision-making, allowing algorithms to make (semi-)autonomous decisions that were once under the full purview of humans. The resulting \textit{algorithmic economy} \citep{lobelRevenueManagementRise2021} is hence one in which interactions among human decision makers are increasingly replaced by interactions among algorithms. This is particularly true in the context of pricing analytics, where decades of research have given rise to an abundance of dynamic pricing algorithms \citep{bitranOverviewPricingModels2003, denboerDynamicPricingLearning2015}.

Consequently, there is a growing need to better understand how pricing algorithms interact and influence market outcomes, both for regulatory oversight and revenue management practices. One of the most recent concerns is \textit{algorithmic collusion}, in which independent pricing algorithms operating in competitive markets may learn to exhibit or even sustain \textit{supracompetitive prices} (i.e., prices that are higher than the competitive level or even near-collusive). The concept of algorithmic collusion emerged around 2015 \citep{cooperLearningPricingModels2015, salcedoPricingAlgorithmsTacit2015} and became a topic of heated discussion with the seminal work of \citet{calvanoArtificialIntelligenceAlgorithmic2020}. Since then, it has received a fair amount of attention both in academic work and in the popular press and regulatory circles \citep{AlgorithmsCollusionCompetition2017, calzolariMisleadingConsequencesComparing2021, klobuchar2024antitrust, doj2024realpage}. Despite extensive numerical studies and theoretical exposition, many problems remain open in this area \citep{denboerAlgorithmicCollusionMathematical2024, abadaAlgorithmicCollusionWhere2025}.

\subsection{Motivation}\label{sec:motivation}
A natural and  fundamental question pertinent to pricing practices in competitive markets, is whether sellers should explicitly incorporate competitors' prices into their demand models when learning to price. In particular, when a seller's demand depends on other sellers' prices, as is natural in many real world scenarios, explicitly modeling competitors' prices appears both economically sound and methodologically prudent. The algorithmic collusion literature has challenged this view. 

\paragraph{Oblivious learning and collusive outcomes.}  While explicit collusion is prohibited by antitrust laws, \emph{tacit} collusion through a seemingly benign pricing algorithm is attractive to sellers, as the opacity of algorithmic decision-making affords plausible deniability, and there is little consensus on what constitutes a legally ``correct'' learning algorithm in competitive settings. Several works suggest that sellers may benefit from ignoring competitors' prices when modeling demand, as this may facilitate collusive outcomes \citep{cooperLearningPricingModels2015, hansenFrontiersAlgorithmicCollusion2021, douglasNaiveAlgorithmicCollusion2024, baek2026misspecified}. Moreover, there are several practical considerations for oblivious modeling; collecting competitors' prices may be costly or infeasible, and learning from one's own price-demand data is simpler. This means that, if ignoring competitors' prices can lead to collusive outcomes, sellers may have both economic and operational incentives to do so.

\paragraph{Impact of demand noise and price exploration.} 
Although oblivious sellers effectively behave as monopolists in their own demand models, the interaction of ``reasonable'' monopolistic learning algorithms may give rise to  puzzling market behavior. A key determining factor in this setting is the amount of ``noise'' that governs demand related observations,  and the extent of exploration employed by the learning algorithms to ensure proper identification of the demand model. In a strict monopolistic setting  there is a fairly clear understanding of these determinants, and specifically the extent of exploration that is ``necessary''  for near optimal performance. (See, e.g., \citet{denBoerSimultaneouslyLearningOptimizing2014} and \citet{keskinDynamicPricingUnknown2014}, which establish optimality of a cumulative price exploration of order square root of the horizon.)  In contrast, when all sellers in a competitive market are oblivious and adopt similar exploration rates the picture that emerges is more ambiguous.  The following numerical simulations illustrate  a wide range of possible outcomes, spanning near-competitive prices, near-collusive prices, and indeterminate regimes (Figure~\ref{fig:samplepaths}).\footnote{The panels use a symmetric duopoly with $\alpha = 2.5$, $\beta = -1$, $\gamma = 0.4$, $[l, u] = [0.5, 2.5]$, $\nu_n^2 = 0.025/\sqrt{n+1}$, and horizon $n = 10^4$; see Section~\ref{sec:model} for definitions.} This phenomenon echoes earlier observations in settings \emph{without demand noise or price exploration}: \citet[\S4.3]{cooperLearningPricingModels2015} documents a continuum of possible price limit points under oblivious learning. That similar patterns persist even in the presence of judicious exploration and noisy observations  raises several fundamental questions, first and foremost,  whether and how these elements affect the nature of limit points.  

\begin{figure}[!htbp]
\centering
\begin{subfigure}{.33\textwidth}
\centering
\includegraphics[width=\linewidth]{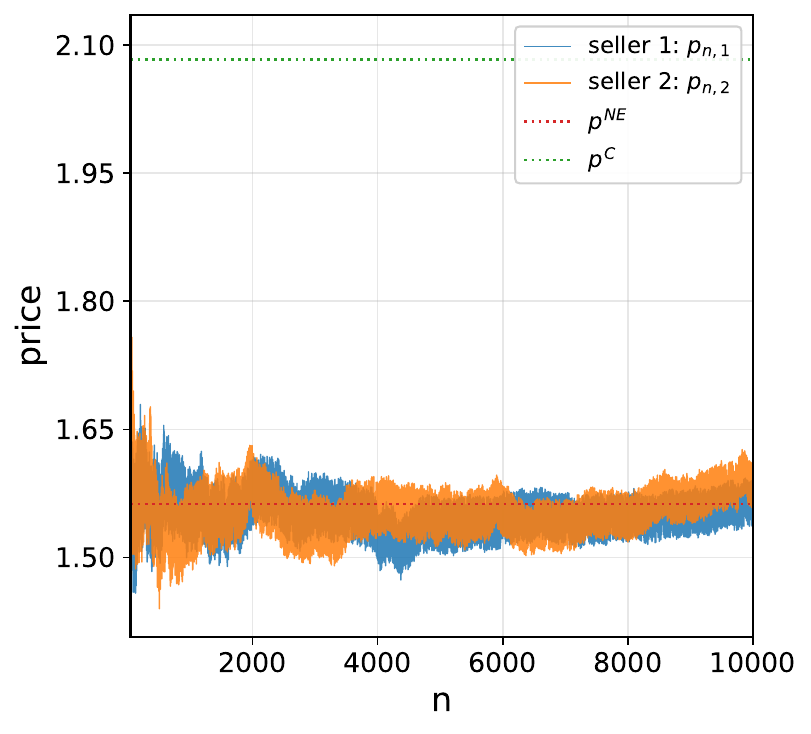}
\caption{A near-competitive outcome.}
\label{fig:samplepaths1}
\end{subfigure}%
\begin{subfigure}{.33\textwidth}
\centering
\includegraphics[width=\linewidth]{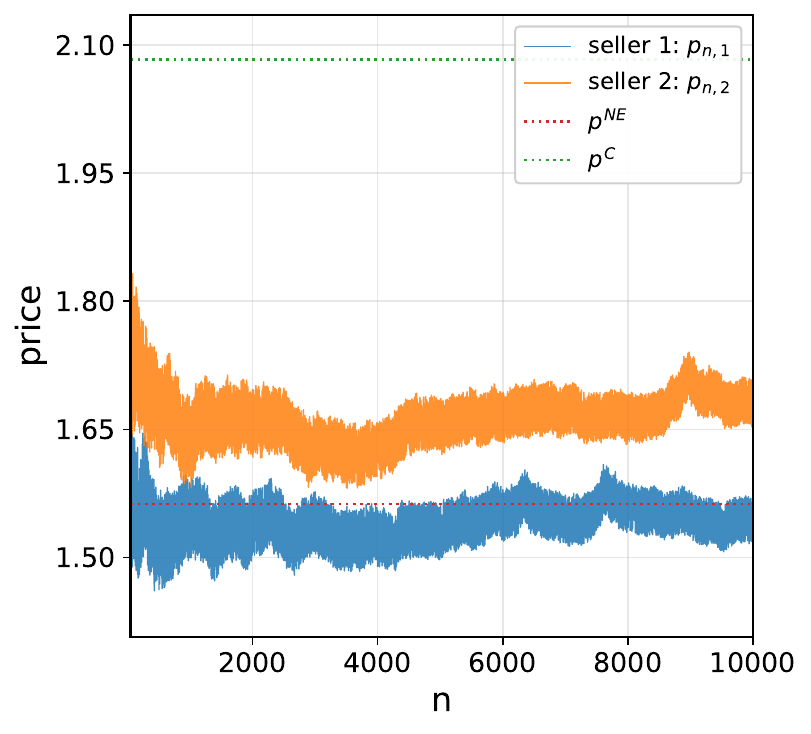}
\caption{An indeterminate outcome.}
\label{fig:samplepaths2}
\end{subfigure}%
\begin{subfigure}{.33\textwidth}
\centering
\includegraphics[width=\linewidth]{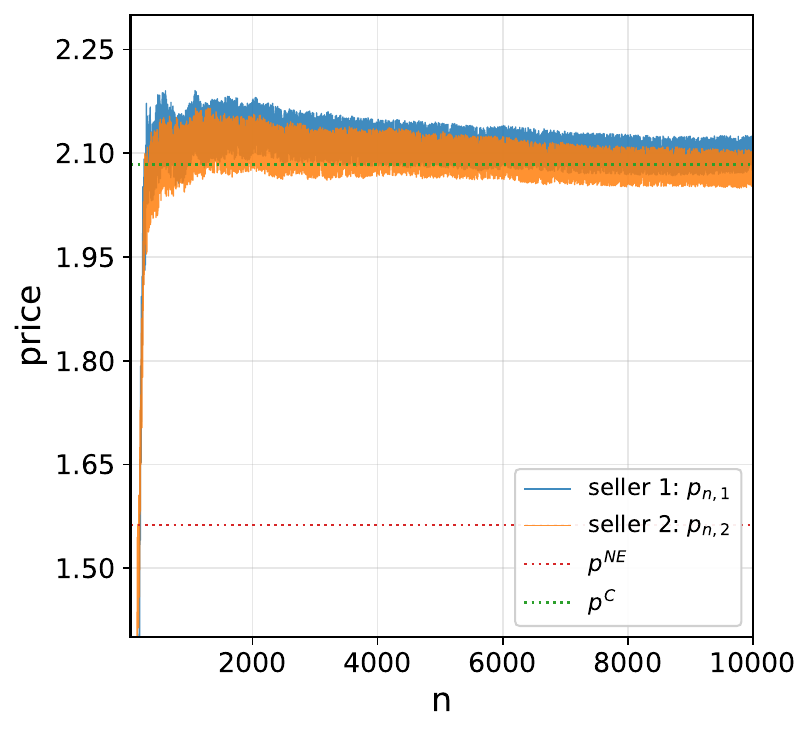}
\caption{A near-collusive outcome.}
\label{fig:samplepaths3}
\end{subfigure}%
\caption{Three sample paths of the per-period prices of two oblivious sellers under a ``standard'' (square-root) exploration schedule. In each panel, the two solid lines trace the two sellers' prices over the pricing horizon and the two dotted lines mark the competitive and collusive price benchmarks. The three panels illustrate, respectively, a near-competitive outcome, an indeterminate outcome, and a near-collusive outcome.}
\label{fig:samplepaths}
\end{figure}

\subsection{Contributions}
 We consider a stylized competitive market in which multiple sellers repeatedly set prices and estimate demand from noisy observations via iterated least squares, and sellers may either incorporate competitors' prices into their demand models (\textit{informed}) or ignore them (\textit{oblivious}). Our contributions include the following: 
\begin{enumerate}
\item When all sellers are oblivious, sublinear exploration leads to convergence to sample-path-dependent limits. When one seller's price variance dominates the rest, the dominant seller earns sublinear regret while the others incur linear regret in general (Sections~\ref{sec:divergent-exploration}--\ref{sec:variance-dominance}).
\item This regret dichotomy induces a ``spiral-up'' logic that pushes sellers toward \emph{linear} exploration rates, under which prices converge to the competitive (Nash) outcome, but sellers nevertheless incur a persistent ``exploration tax'' on their revenue (Sections~\ref{sec:variance-dominance}--\ref{sec:global-convergence}).
\item When informed sellers are present, they consistently learn the true demand function, and the market converges to the competitive outcome. In particular, informed sellers do not bear the persistent ``exploration tax'' on their revenue (Sections~\ref{sec:all-informed}--\ref{sec:mixedmarkets}).
\item Consolidating these three market compositions, the choice between oblivious and informed modeling can be formulated as a strategy game (Table~\ref{tab:meta-game-intro}), in which \emph{informed} is the dominant strategy for each seller, and the unique strict pure-strategy Nash equilibrium is the \emph{all-informed market}. The same conclusion carries over to a general multi-seller setting (Section~\ref{sec:meta-strategy}).
\end{enumerate}

\begin{table}[H]
\centering
\renewcommand{\arraystretch}{1.8}
\setlength{\tabcolsep}{12pt}
\begin{tabular}{c|c|c}
\hline
 & \textbf{\textsf{oblivious}} & \textbf{\textsf{informed}} \\
\hline
\textbf{\textsf{oblivious}} &
\textcolor{red}{$<$\,Nash}, \textcolor{red}{$<$\,Nash} &
\textcolor{red}{$<$\,Nash}, \textcolor{green!50!black}{Nash} \\
\hline
\textbf{\textsf{informed}} &
\textcolor{green!50!black}{Nash}, \textcolor{red}{$<$\,Nash} &
\textcolor{green!50!black}{\textbf{Nash}}, \textcolor{green!50!black}{\textbf{Nash}} \\
\hline
\end{tabular}
\caption{The strategy game. Rows index seller~1, columns index seller~2. Each seller chooses whether to model competitor prices (\emph{informed}) or not (\emph{oblivious}). Cell entries summarize the asymptotic revenue earned by each seller relative to the Nash benchmark (the surplus capture ratio to be defined in Section~\ref{sec:performance-metric}). The unique strict pure-strategy Nash equilibrium of this game is \emph{(informed, informed)}.}
\label{tab:meta-game-intro}
\end{table}

The rest of the paper is organized as follows. Section~\ref{sec:literature} discusses relevant literature; Section~\ref{sec:model} presents the model and preliminaries; Section~\ref{sec:oblivious-sellers} studies oblivious learning dynamics and the resulting spiral-up phenomenon; Section~\ref{sec:convergence} establishes convergence to the competitive outcome across all three market compositions and synthesizes the strategy game; Section~\ref{sec:discussion} concludes the paper with a discussion of implications and future directions. Code for all numerical experiments is available at \url{https://github.com/yw3453/ob-learn}.

\section{Literature Review} \label{sec:literature}
\paragraph{Algorithmic collusion.} The seminal work of \citet{calvanoArtificialIntelligenceAlgorithmic2020} drew significant attention by providing numerical evidence that Q-learning agents can learn to charge \textit{supracompetitive prices} and exhibit reward-punishment schemes. Numerous works followed up, adding further numerical evidence; see, e.g., \citet{fishAlgorithmicCollusionLarge2026} and references therein. Most theoretical work on algorithmic collusion proceeds through the lens of a repeated Prisoner's Dilemma with a \emph{deterministic} payoff matrix, primarily for reinforcement learning (RL) algorithms \citep{hansenFrontiersAlgorithmicCollusion2021, banchioArtificialIntelligenceSpontaneous2022, abadaArtificialIntelligenceCan2023, abadaCollusionMistakeDoes2024, douglasNaiveAlgorithmicCollusion2024, lambinLessMeetsEye2024, bertrandQlearnersCanProvably2025}. Despite numerous insights, the above works are limited in that they \emph{do not} consider demand noise. Another line of work on algorithmic collusion stems from an algorithmic game-theoretic view and largely studies repeated (Bertrand) pricing games \citep{brownCompetitionPricingAlgorithms2023, arunachaleswaranAlgorithmicCollusionThreats2024, bichlerOnlineOptimizationAlgorithms2025, cartea2026algorithmic} but under deterministic demands as well. \citet{cao2026collusive} show that a shared AI model can produce supracompetitive prices; there, price correlation stems from a common algorithm, whereas our sellers learn independently. Some works also design algorithms intended to collude \citep{aouadAlgorithmicCollusionAssortment2021, meylahnLearningColludePricing2022, lootsDataDrivenCollusion2023}. Empirical evidence has been sparse \citep{musolffAlgorithmicPricingPrice2022, assadAlgorithmicPricingCompetition2024}. For more comprehensive reviews, we refer readers to \citet{denboerAlgorithmicCollusionMathematical2024} and \citet{abadaAlgorithmicCollusionWhere2025}.

\paragraph{Pricing and learning in competition.} A common result in multi-seller revenue management is that the system will achieve optimal regret with respect to a Nash equilibrium if every seller deploys a specific (often informed) algorithm (see, e.g., \citet{yangCompetitiveDemandLearning2024, bracale2026online, bracale2026revenue, li2026lego}). An earlier line of work studies multi-seller pricing with a fixed inventory and a finite time horizon (see, e.g., \citet{gallego2014dynamic, chenRecentDevelopmentsDynamic2015}). Learning in competition is also related to learning in non-stationary environments, but the commonly assumed sublinear variation budget that constrains how much the environment can change over time (e.g., \citet{besbesNonStationaryStochasticOptimization2015}) does not apply to a competitive setting with possibly persistent strategic interactions. In game theoretic contexts, when payoffs are \emph{known} and players adaptively best respond to historical averages of opponents' actions, classical results suggest convergence to Nash (see, e.g., \citet{milgrom1990rationalizability,milgrom1991adaptive}). Our setting departs from this benchmark in that the demand parameters are unknown, so the question is no longer one of purely adaptive decision-making but also of joint estimation under (possibly misspecified) demand models.

\paragraph{Oblivious learning.} Most closely related to our work is \citet{cooperLearningPricingModels2015}, who consider two sellers with linear demand and study the limit prices when sellers do not model competitor prices. Crucially, however, they do not consider price exploration, without which demand noise is known to cause suboptimal performance \citep{besbes2009dynamic, denBoerSimultaneouslyLearningOptimizing2014, keskinDynamicPricingUnknown2014}. Multiple works related to algorithmic collusion also consider oblivious settings \citep{hansenFrontiersAlgorithmicCollusion2021, douglasNaiveAlgorithmicCollusion2024} without formally discussing this modeling decision. There is also a fruitful line of work whose focus is to \textit{design} a distributed algorithm such that, if it is deployed by all sellers, the market converges to a Nash outcome efficiently; e.g., \citet{golrezaei2020no}, \citet{goyalLearningPriceCompetition2023}, \citet{liAdaptiveLearningUncertain2026}, and \citet{yang2026learning} consider variants of online gradient descent. Two concurrent works also study pricing algorithms with omitted competitive information. \citet{lin2025competition} study a duopoly in which both firms run misspecified monopoly-style OLS. They derive a stability condition under which the market converges to the Nash equilibrium and discuss factors that lead to supracompetitive prices. \citet{baek2026misspecified} analyze an estimate-then-optimize pipeline in which firms first randomize prices over an initial window and then exploit a misspecified monopoly estimate, showing that limiting prices could reach supracompetitive levels. Their results complement our analysis of regimes in which persistent excitation is absent and a continuum of pseudo-equilibria may arise (Section~\ref{sec:oblivious-sellers}).

\paragraph{Learning under misspecified models.} Classical results on learning under a misspecified model show that maximum-likelihood estimators converge to a \emph{pseudo-true} parameter \citep{berkLimitingBehaviorPosterior1966, whiteMaximumLikelihoodEstimation1982}. \citet{espondaBerkNashEquilibrium2016} extend this idea to a game-theoretic setting and propose a relevant solution concept, whereas we show explicit convergence from primitives. Closer to our pricing context, \citet{besbesSurprisingSufficiencyLinear2015} and \citet{nambiarDynamicLearningPricing2019} study single-seller misspecification, whereas we explicitly consider a multi-seller setting where the misspecification comes from ignoring competitors. \citet{yang2026driven} show that oblivious sellers with partial price feedback converge to the Nash equilibrium, while \citet{light2026experimentation} show that correlated experimentation leads to a conjectural-variations equilibrium. Both analyses assume that the relevant cross-seller price correlations converge, while we impose no such condition and instead derive convergence directly from the learning dynamics.

\section{Model}\label{sec:model}

\subsection{Demand Model}\label{sec:demand-model}
We consider a market with $N \ge 2$ sellers indexed by $[N] = \{1,\ldots,N\}$, where each seller offers a single product. Time is discrete and indexed by $n \in \mathbb{Z}_+$. Let $\mathbf{p}_n = (p_{n,1},\ldots,p_{n,N}) \in \mathbb{R}^N$ denote the vector of prices at time $n$. The demand faced by seller $i$ at time $n$ follows a linear model:
\begin{equation}\label{eq:demand}
d_{n,i} = \alpha_i + \beta_i p_{n,i} + \sum_{j \neq i} \gamma_{i,j} p_{n,j} + \varepsilon_{n,i},
\end{equation}
where the following assumptions hold:
\begin{itemize}
\item Prices are bounded: $p_{n,i} \in [l,u]$ for all $n$ and $i$, with $u>l>0$.
\item The noise terms $\varepsilon_{n,i}$ are i.i.d.\ across $n$ and $i$, with zero mean and bounded support.
\item The demand parameters satisfy $\alpha_i>0$, $\beta_i<0$, $\gamma_{i,j}>0$ for all $i\neq j$, and $-\beta_i>\gamma_i$ for all $i$, where $\gamma_i\triangleq \sum_{j\neq i}\gamma_{i,j}$.
\end{itemize}
These assumptions are standard in the dynamic pricing literature \citep{keskinDynamicPricingUnknown2014, cooperLearningPricingModels2015}. They ensure that expected demand is non-negative, products are imperfect substitutes, and own-price effects dominate cross-price effects. The price bounds reflect practical considerations, such as marginal costs and reasonable upper limits on prices. We next distinguish sellers by whether they explicitly model competitive effects.

\begin{definition}[\textbf{Oblivious and informed}]
Seller $i$ is said to be \emph{oblivious} if it models demand using a misspecified monopolistic form
\[
d_{n,i} = a_i + b_i p_{n,i} + \varepsilon_{n,i}.
\]
In contrast, seller $i$ is \emph{informed} if it uses a correct demand model~\eqref{eq:demand}.
\end{definition}

\subsection{Learning and Pricing Dynamics}\label{sec:learning-dynamics}
All sellers, oblivious or informed, follow the same learning-and-pricing protocol, differing only in the demand model they estimate. At time $n$, an oblivious seller $i$ only observes its own price-demand history $\mathcal{H}_{n,i}^{ob} = \{(p_{m,i}, d_{m,i})\}_{m=1}^n$, whereas an informed seller $i$ observes the full price vector, $\mathcal{H}_{n,i}^{in} = \{(\mathbf{p}_m, d_{m,i})\}_{m=1}^n$. No seller observes other sellers' demands. Define the regressors
\[
x_{n,i}^{ob} = (1, p_{n,i})^\top \in \mathbb{R}^2,
\qquad
x_{n,i}^{in} = (1, p_{n,1},\ldots,p_{n,N})^\top \in \mathbb{R}^{N+1}.
\]
An oblivious seller thus estimates two parameters, while an informed seller estimates $N+1$ parameters. For each seller $i$, let $\Theta_i^{ob}\subset\mathbb{R}^2$ and $\Theta_i^{in}\subset\mathbb{R}^{N+1}$ denote known compact and convex parameter sets. We assume throughout that the own-price slope coordinate is bounded above by a negative constant on both $\Theta_i^{ob}$ and $\Theta_i^{in}$ so that the induced greedy price is bounded. At time $n$, seller $i$ computes the least-squares estimator and projects it onto the feasible set:
\[
\tilde{\theta}_{n,i}^{ob} = \argmin_{\theta\in\mathbb{R}^2}\sum_{m=1}^n [d_{m,i} - (x_{m,i}^{ob})^\top\theta]^2,
\qquad
\hat{\theta}_{n,i}^{ob} = \pi_{\Theta_i^{ob}}(\tilde{\theta}_{n,i}^{ob}),
\]
\[
\tilde{\theta}_{n,i}^{in} = \argmin_{\theta\in\mathbb{R}^{N+1}}\sum_{m=1}^n [d_{m,i} - (x_{m,i}^{in})^\top\theta]^2,
\qquad
\hat{\theta}_{n,i}^{in} = \pi_{\Theta_i^{in}}(\tilde{\theta}_{n,i}^{in}),
\]
where $\pi_{\mathcal C}$ denotes Euclidean projection onto the set $\mathcal C$. The oblivious seller's expected single-period revenue under price $p$ is then estimated as $\hat{r}^{ob}(p, \hat{\theta}_{n,i}^{ob}) = p \cdot (\hat{a}_{n,i} + \hat{b}_{n,i} p)$. For the informed seller, the expected single-period revenue under price $p$ given a vector of \emph{predicted competitors' prices} $\hat{\mathbf{p}}_{n+1,-i} = (\hat p_{n+1,j})_{j \neq i}$ is estimated as $\hat{r}^{in}(p, \hat{\theta}_{n,i}^{in}, \hat{\mathbf{p}}_{n+1,-i}) = p \cdot (\hat{\alpha}_{n,i} + \hat{\beta}_{n,i} p + \sum_{j \neq i} \hat{\gamma}_{n,i,j}\, \hat p_{n+1,j})$. Define the myopic revenue-maximizing prices as
\begin{equation}\label{eq:greedy-prices}
\phi^{ob}(\theta^{ob}) = \argmax_{p \in \mathbb{R}} \hat{r}^{ob}(p, \theta^{ob}), \quad \phi_i^{in}(\theta^{in}_i; \hat{\mathbf{p}}_{-i}) = \argmax_{p \in \mathbb{R}} \hat{r}^{in}(p, \theta^{in}_i, \hat{\mathbf{p}}_{-i}).
\end{equation}
Since both revenue functions are strictly concave quadratics in own price, the greedy prices are
\begin{equation}\label{eq:obliviousgreedynextprice}
\tilde p_{n+1,i}^{ob} = \phi^{ob}(\hat{\theta}_{n,i}^{ob}) = -\frac{\hat{a}_{n,i}}{2 \hat{b}_{n,i}}, \quad \tilde p_{n+1,i}^{in} = \phi_i^{in}(\hat{\theta}_{n,i}^{in}; \hat{\mathbf{p}}_{n+1,-i}) = -\frac{\hat{\alpha}_{n,i} + \sum_{j \neq i} \hat{\gamma}_{n,i,j}\, \hat p_{n+1,j}}{2 \hat{\beta}_{n,i}}.
\end{equation}
To explore, all sellers then add a random perturbation $z_{n+1,i}$ to the greedy price, where the $z_{n,i}$ are independent across $n$ and $i$ with mean zero, variance $\nu_{n,i}^2$, and bounded support $|z_{n,i}| \le \bar z_i$; the final price is
\begin{equation}\label{eq:realized-price}
p_{n+1,i} \;=\; \tilde p_{n+1,i} + z_{n+1,i}.
\end{equation}
We assume the feasible range $[l,u]$ is wide enough that the perturbed greedy price never leaves it:
\begin{itemize}
\item \emph{Oblivious seller $i$:} $\phi^{ob}(\Theta_i^{ob}) + [-\bar z_i, \bar z_i] \subseteq [l,u]$.
\item \emph{Informed seller $i$:} $\phi_i^{in}\left(\Theta_i^{in};\, [l,u]^{N-1}\right) + [-\bar z_i, \bar z_i] \subseteq [l,u]$.
\end{itemize}
This is a standard interior assumption in dynamic pricing \citep{denBoerSimultaneouslyLearningOptimizing2014, keskinDynamicPricingUnknown2014, besbesSurprisingSufficiencyLinear2015}, always attainable by widening $[l,u]$ or shrinking the exploration support.

\paragraph{Alternative algorithms.}
While there are many reasonable dynamic pricing algorithms \citep{denboerDynamicPricingLearning2015}, virtually any ``good'' one has three essential components: \emph{estimate}, \emph{exploit}, and \emph{explore}. Our model is thus intentionally stylized to reflect this paradigm and maintain tractability for the analysis of oblivious vs. informed modeling. As a robustness check, Appendix~\ref{sec:robustness-learning-rule} provides numerical evidence that this paper's qualitative conclusions carry over to ridge-regularized estimation and $\varepsilon$-greedy exploration, suggesting that our results generalize beyond the specific learning rule we study.

\paragraph{Alternative models.}
Two alternative modeling choices are also common in the algorithmic-collusion literature. Finite price grids paired with bandit-style or RL algorithms (Section~\ref{sec:literature}) admit no-regret guarantees but discard the parametric demand structure, and existing theoretical analyses are often confined to two-price grids. The (multinomial) logit demand model \citep{aouadAlgorithmicCollusionAssortment2021, goyalLearningPriceCompetition2023, lootsDataDrivenCollusion2023} differs informationally: every successful sale by one seller is a sale a competitor does not make, so sellers observe a partial signal about rival demand from their own data, whereas our linear-demand model with unknown intercept offers no such handle and is therefore strictly harder. 

\subsection{Solution Concepts}\label{sec:solution-concept}
We benchmark learning dynamics against two full-information outcomes: the \emph{competitive outcome} (Nash equilibrium) and the \emph{collusive outcome}. In a pure-strategy Nash equilibrium (NE), each seller best responds to competitors' prices. The unconstrained equilibrium price vector $\mathbf{p}^{NE}$ satisfies
\[
p_i^{NE} = \argmax_{p_i\in\mathbb R} p_i\!\left(\alpha_i + \beta_i p_i + \sum_{j\neq i}\gamma_{i,j}p_j^{NE}\right),
\qquad i\in[N].
\]
Let $\Gamma$ denote the $N\times N$ matrix with $\Gamma_{ii}=2\beta_i$ and $\Gamma_{ij}=\gamma_{i,j}$ for $i\neq j$, and let $\boldsymbol{\alpha}=(\alpha_1,\ldots,\alpha_N)^\top$. Since $-\beta_i>\gamma_i$ for all $i$, $\Gamma$ is strictly diagonally dominant hence invertible, yielding the unique equilibrium $\mathbf{p}^{NE} = -\Gamma^{-1}\boldsymbol{\alpha}$. On the other hand, the unconstrained \emph{collusive outcome} maximizes total expected revenue:
\[
\mathbf{p}^{C} = \argmax_{\mathbf{p}\in\mathbb R^N}
\sum_{i=1}^N p_i\!\left(\alpha_i + \beta_i p_i + \sum_{j\neq i}\gamma_{i,j}p_j\right).
\]
The first-order conditions yield the $N\times N$ symmetric matrix $H$ with $H_{ii}=2\beta_i$ and $H_{ij}=\gamma_{i,j}+\gamma_{j,i}$ for $i\neq j$. A unique collusive outcome $\mathbf{p}^{C} = -H^{-1}\boldsymbol{\alpha}$ exists whenever $H$ is negative definite, a standard requirement in differentiated-demand oligopoly models \citep{singhPriceQuantityCompetition1984,vivesOligopolyPricingOld1999,choneLinearDemandSystems2020} that we impose throughout.\footnote{When $-2\beta_i>\gamma_i+\gamma_i^{\mathrm{col}}$ for all $i$, where $\gamma_i^{\mathrm{col}}\triangleq \sum_{j\neq i}\gamma_{j,i}$, we obtain an economically informative ordering of $\mathbf{p}^{C}\ge\mathbf{p}^{NE}$, i.e., collusion raises every seller's price above its competitive level and negatively impacts consumers. This condition aids interpretation only and is not required by any of our results. While richer notions of collusion exist to ensure fairer division of gains \citep{osborneProfitSharingCollusive1983,fischerCollusionBargainingAsymmetric2019}, they are bounded above by $\mathbf{p}^{C}$ and thus unused for the ease of exposition.} To allow meaningful discussions of competitive and collusive outcomes, we also assume that $\mathbf{p}^{NE},\mathbf{p}^{C}\in(l,u)^N$.

\subsection{Performance Metric}\label{sec:performance-metric}
\paragraph{Regret and Nash equilibrium.}
A commonly used metric in competitive dynamic pricing is \emph{regret}, the cumulative revenue gap relative to an oracle that best-responds to the realized competitor profile at every step \citep{sternDynamicLearningStrategic2020, yangCompetitiveDemandLearning2024, liAdaptiveLearningUncertain2026, li2026lego}. Formally, seller~$i$'s expected per-period revenue and the full-information best response are
\[
R_i(p_i, \mathbf{p}_{-i}) \triangleq p_i \left(\alpha_i + \beta_i p_i + \sum_{j\neq i}\gamma_{i,j} p_j\right),\qquad \mathrm{BR}_i(\mathbf{p}_{-i}) \triangleq \frac{\alpha_i + \sum_{j\neq i}\gamma_{i,j} p_j}{-2\beta_i},
\]
respectively. The regret of seller~$i$ over horizon $T$ is defined as
\begin{equation*}
\Delta_i(T) \;\triangleq\; \mathbb{E}\!\left[\sum_{n=1}^{T} \left(R_i\!\left(\mathrm{BR}_i(\mathbf{p}_{n,-i}), \mathbf{p}_{n,-i}\right) - R_i\!\left(p_{n,i}, \mathbf{p}_{n,-i}\right)\right)\right],
\end{equation*}
where the expectation is over the randomness of pricing policies. To set the stage, the following proposition situates the Nash equilibrium as the reference of regret; the proof is in Appendix~\ref{proof:prop:dynamicbenchmark}.\footnote{In classic dynamic pricing, regret can often be reduced to the cumulative squared distance from the per-period price to the monopoly optimum \citep{keskinDynamicPricingUnknown2014}. Proposition~\ref{prop:dynamicbenchmark} follows similar steps and instead uncovers the Nash equilibrium as the reference.}

\begin{proposition}[Regret]\label{prop:dynamicbenchmark}
For every $T\ge 1$ and every seller $i\in[N]$,
\begin{equation}\label{eq:Delta-vi}
\Delta_i(T) \;=\; -\beta_i\,\mathbb{E}\!\left\{\sum_{n=1}^{T}\left[\sum_{j\neq i}\frac{\gamma_{i,j}}{-2\beta_i}\!\left(p_{n,j} - p_j^{NE}\right) - \!\left(p_{n,i} - p_i^{NE}\right)\right]^2\right\} \lesssim \mathbb{E} \left[\sum_{n=1}^{T}\|\mathbf{p}_n - \mathbf{p}^{NE}\|_2^2\right].
\end{equation}
In particular, 
$$
\sum_{i=1}^{N}\Delta_i(T) \asymp \mathbb{E} \left[\sum_{n=1}^{T}\|\mathbf{p}_n - \mathbf{p}^{NE}\|_2^2 \right].
$$
\end{proposition}
Here $\lesssim $ is used to denote an upper bound that is non-asymptotic but omits constants; and $\asymp$ is used to denote quantities that are equivalent up to constants (i.e., whose ratio is bounded above and below).   

\paragraph{Regret decomposition.}
For clean attribution of revenue loss, it is also useful to decompose regret into a learning component and an (active) exploration tax:
\begin{equation}\label{eq:regret-decomp}
\Delta_i(T) \;=\; \Delta_i^{\mathrm{le}}(T) \;+\; \Delta_i^{\mathrm{ex}}(T),
\end{equation}
where the learning term
\[
\Delta_i^{\mathrm{le}}(T) \;=\; -\beta_i\,\mathbb{E}\left\{\sum_{n=1}^{T}\left[\sum_{j\neq i}\frac{\gamma_{i,j}}{-2\beta_i}\left(\tilde p_{n,j} - p_j^{NE}\right)- \left(\tilde p_{n,i} - p_i^{NE}\right)\right]^2\right\} \lesssim \mathbb{E} \left[\sum_{n=1}^{T}\|\tilde{\mathbf{p}}_n - \mathbf{p}^{NE}\|_2^2\right]
\]
similarly satisfies $\sum_{i=1}^{N}\Delta_i^{\mathrm{le}}(T) \asymp \mathbb{E} \left[\sum_{n=1}^{T}\|\tilde{\mathbf{p}}_n - \mathbf{p}^{NE}\|_2^2\right]$, and the exploration tax
\[
\Delta_i^{\mathrm{ex}}(T) \;=\; -\beta_i\sum_{n=1}^{T}\left(\nu_{n,i}^2 + \sum_{j\neq i}\left(\frac{\gamma_{i,j}}{2\beta_i}\right)^2\nu_{n,j}^2\right).
\]
We will henceforth denote $\mathrm{\mathbf{RegLe}}(T) \triangleq \sum_{i=1}^{N}\Delta_i^{\mathrm{le}}(T)$ and $\mathrm{\mathbf{RegEx}}(T) \triangleq \sum_{i=1}^{N}\Delta_i^{\mathrm{ex}}(T)$ for the total learning and exploration components of regret, respectively.

\paragraph{Surplus-capture ratio.}
Another commonly used metric is the \emph{surplus-capture ratio} \citep{calvanoArtificialIntelligenceAlgorithmic2020}, which for seller~$i$ along a price path $\{\mathbf{p}_n\}_{n\ge 1}$ is defined as
\begin{equation*}
S_i \;\triangleq\; \liminf_{n\to\infty}\;
\frac{\frac{1}{n}\sum_{m=1}^{n} p_{m,i}\, d_{m,i} \;-\; \Pi_i^{NE}}{\Pi_i^{C} \;-\; \Pi_i^{NE}},
\end{equation*}
where the seller-specific per-period revenues at the competitive and collusive outcomes are $\Pi_i^{NE} \triangleq p_i^{NE}(\alpha_i + \beta_i p_i^{NE} + \sum_{j\neq i}\gamma_{i,j} p_j^{NE})$ and $\Pi_i^{C} \triangleq p_i^{C}(\alpha_i + \beta_i p_i^{C} + \sum_{j\neq i}\gamma_{i,j} p_j^{C})$ respectively; $\Pi_i^{C} - \Pi_i^{NE}$ is assumed to be strictly positive for every $i$. By construction, $S_i = 0$ means that seller~$i$ matches its Nash benchmark on average, $S_i = 1$ means it matches its collusive benchmark, and both the super-collusive case $S_i > 1$ and the sub-Nash case $S_i < 0$ are possible. 

\paragraph{On the choice of performance metric.}
The two metrics are informative in complementary ways. \textbf{Regret} is a \emph{relative} measure of revenue: it compares a seller's revenue to the \emph{best achievable} given the realized competitor prices. However, since minimizing regret is closely connected to driving prices toward the Nash equilibrium (Proposition~\ref{prop:dynamicbenchmark}), it is meaningful only when (1) prices are known to settle in the long run and (2) they settle near the Nash equilibrium. For example, if competitor prices do not settle and/or are primarily supracompetitive, a seller can earn a decent surplus capture yet incur linear regret. \textbf{Surplus capture}, by contrast, is an \emph{absolute} measure: it speaks directly to the economic outcome a seller cares about. However, since absolute revenue depends heavily on competitors, surplus capture no longer reflects the seller's decision quality cleanly when long-run competitor prices are random or suboptimal. For example, if competitors price too low, the best thing a seller can do is best-respond, which incurs low regret yet may still yield low surplus capture. We therefore adopt both metrics in our analysis, as together they provide a more complete picture of the economic outcome of interest to sellers.

\section{Oblivious Dynamics and the Spiral-Up} \label{sec:oblivious-sellers}
In monopolistic dynamic pricing with unknown demand, the role of price exploration is relatively well understood \citep{besbes2009dynamic, denBoerSimultaneouslyLearningOptimizing2014, den2014dynamicmultiple, keskinDynamicPricingUnknown2014}. In contrast, under \emph{oblivious} demand modeling, estimation is misspecified, affecting both what can be learned and what exploration is needed for stable performance. Rather than imposing an ad hoc exploration rule, in this section we consider markets in which all sellers are oblivious and ask \emph{how oblivious sellers should explore} and \emph{where prices settle} as a result. 

\paragraph{Exploration rate.}
At the center of our analysis is the sellers' \emph{cumulative price variance} up to time $n$,
\[
J_{n,i}
\;\triangleq\;
\sum_{m=1}^n \left(p_{m,i} - \bar p_{n,i}\right)^2
\;=\;
\sum_{m=2}^{n} \left(1 - \frac{1}{m}\right)\left(p_{m,i} - \bar p_{m-1,i}\right)^2, \quad i \in [N],
\]
where $\bar p_{n,i}\triangleq \tfrac{1}{n}\sum_{m=1}^n p_{m,i}$. We assume that the first two prices are different so that $J_{n,i}>0$ for all $n\ge 2$. Sellers have sizable control over $J_{n,i}$ through their choices of perturbations $z_{n,i}$, and we use \emph{exploration rate} throughout the paper to refer to the asymptotic growth of $J_{n,i}$ as a function of $n$ (e.g., linear vs.\ sublinear). 

\paragraph{Oblivious greedy prices.}
To understand how exploration rate affects market dynamics, it is imperative to study how greedy prices evolve. A direct least-squares calculation gives
\begin{equation}\label{eq:oblivious-ls-coeffs}
\tilde a_{n,i} = \alpha_i + \sum_{j\neq i}\gamma_{i,j}\left(\bar p_{n,j} - r_{n,i\leftarrow j}\bar p_{n,i}\right) + \bar\varepsilon_{n,i} - w_{n,i}\bar p_{n,i}, \qquad 
\tilde b_{n,i} = \beta_i + \sum_{j\neq i}\gamma_{i,j}\,r_{n,i\leftarrow j} + w_{n,i},
\end{equation}
where $w_{n,i} \triangleq \tfrac{1}{J_{n,i}} \sum_{m=1}^n (p_{m,i} - \bar p_{n,i})\varepsilon_{m,i}$ and $r_{n,i\leftarrow j} \triangleq \tfrac{1}{J_{n,i}} \sum_{m=1}^{n} (p_{m,i} - \bar p_{n,i})(p_{m,j} - \bar p_{n,j})$, and the greedy (conditional-mean) price at time $n+1$ is $\tilde p_{n+1,i} = -\tilde a_{n,i} / (2\,\tilde b_{n,i})$ absent projection. Since $\bar\varepsilon_{n,i}\triangleq \tfrac{1}{n}\sum_{m=1}^n \varepsilon_{m,i} \to 0$ a.s., two sources of oblivious dynamics surface: a \emph{self-noise} term $w_{n,i}$ capturing the correlation between own price deviations and idiosyncratic demand shocks, and an \emph{omitted-variable channel} $r_{n,i\leftarrow j}$ that measures strategic interactions through empirical cross-covariance between two sellers' prices. Both play important roles in determining the long-run market outcome.

\subsection{Baseline: Divergent Exploration Rate}\label{sec:divergent-exploration}

Intuitively, $w_{n,i}$ should vanish under effective learning, otherwise idiosyncratic shocks persistently distort demand estimates and pricing decisions. The next theorem formalizes this link between exploration and noise attenuation; the proof is in Appendix~\ref{proof:thm:impactdemandnoise}. We note that statements like ``on event $A$, $B$ happens almost surely'' mean $\mathbb{P}(A\setminus B)=0$.

\begin{theorem}[\textbf{Noise attenuation}]\label{thm:impactdemandnoise}
As $n \to \infty$, the following holds:
\begin{enumerate}
\item On the event that $\tfrac{J_{n,i}}{\log n \log \log n} \rightarrow \infty$, $w_{n,i} \rightarrow 0$ a.s.
\item On the event that $\sum_{m=2}^{\infty} \left(p_{m,i} - \bar p_{m-1,i}\right)^2 (\log m)^{1+\delta} < \infty$ for some $\delta>0$, $w_{n,i} \rightarrow w_{\infty,i}$ for some finite random variable $w_{\infty,i}$ a.s.
\end{enumerate}
\end{theorem}

\paragraph{Discussion.}
Since $J_{n,i}$ is nondecreasing, each sample path satisfies either $J_{n,i}\to\infty$ or $J_{n,i}\to J_{\infty,i}<\infty$, and the two cases of Theorem~\ref{thm:impactdemandnoise} correspond to these scenarios up to log factors: idiosyncratic demand noise becomes asymptotically negligible only if exploration is unbounded, otherwise the long-run price (and revenue) depends on the sample-path dependent limit $w_{\infty,i}$. This mirrors the well-known failure of certainty-equivalent control in monopolistic dynamic pricing \citep{denBoerSimultaneouslyLearningOptimizing2014, keskinDynamicPricingUnknown2014, keskinIncompleteLearningCertaintyEquivalence2018}: competition does \emph{not} automatically save an oblivious seller, as $w_{n,i}$ vanishes only if the seller's \emph{own} prices exhibit sufficient dispersion. We henceforth restrict our attention to the event
\begin{equation}\tag{Min-Var}\label{eq:explore-logn-cond}
\frac{J_{n,i}}{\log n \log \log n} \longrightarrow \infty, \qquad i \in [N],
\end{equation}
so that $w_{n,i}$ is asymptotically negligible. By Lemma~\ref{lem:Jn-lower}, the floor can be cleared by adding perturbations with cumulative variance diverging faster than $\log n \log\log n$ (e.g., $\mathrm{Var}(z_{n,i})\asymp n^{-c}$ for any $c\in[0,1)$).

\subsection{Market Outcome under Sublinear Exploration}\label{sec:sublinear-exploration}

While (Min-Var) provides a baseline for what the exploration rate should be, it does not specify the exact order of $J_{n,i}$ an oblivious seller should aim for. Typically, sublinear rates are desirable to balance exploration and exploitation. To study the market outcome under sublinear exploration, we consider the following event:
\begin{equation}\tag{Max-Var}\label{eq:Jn-polylog}
J_{n,i} / n \le(\log n)^{-(2+\epsilon_i)}\quad\text{eventually for some } \epsilon_i > 0, \qquad i \in [N].
\end{equation}
Event~(Max-Var) is only logarithmically stronger than plain sublinearity $J_{n,i}=o(n)$, and it sits comfortably above the noise floor~(Min-Var): any schedule with $\log n\log\log n \lesssim J_{n,i} \lesssim n/(\log n)^{2+\epsilon_i}$ qualifies (e.g., $J_{n,i}\asymp n^{c}$ for any $c\in(0,1)$). The following result establishes the convergence of average price under these conditions; the proof is in Appendix~\ref{proof:thm:cesaro-main}.

\begin{theorem}[Ces\`aro price convergence]\label{thm:cesaro-main}
On the event that (Min-Var) and (Max-Var) hold, there exist random $\bar q_i\in[l,u]$ such that $\bar{p}_{n,i} \to \bar q_i$ a.s. as $n\to\infty$ for all $i\in[N]$.
\end{theorem}

\paragraph{A continuum of pseudo-equilibria and non-identifiability.}
Theorem~\ref{thm:cesaro-main} guarantees a notion of convergence but leaves open the question of which specific limit is reached. Consider a two-seller market where events (Min-Var) and (Max-Var) happen, and for ease of exposition suppose $r_{n,i\leftarrow j}$ stabilizes to $r_i^\star$ and $r_{n,j\leftarrow i}$ to $r_j^\star$. Any long-run Ces\`aro price pair $(\bar q_1,\bar q_2)$ must satisfy each seller's oblivious best-response operator in the limit, i.e., 
\[
(2\beta_i+\gamma_{i,j}r_i^\star)\bar q_i+\gamma_{i,j}\bar q_j=-\alpha_i, \qquad i \ne j,
\]
which is exactly the system from the no-noise/no-exploration analysis in \citet[\S4.3]{cooperLearningPricingModels2015}. Because the system has two equations for four unknowns $\bar q_1,\bar q_2,r_1^\star,r_2^\star$, it is underdetermined and thus yields a \emph{continuum} of candidate pseudo-equilibria indexed by $(r_1^\star,r_2^\star)$. Moreover, since $r_{n,i\leftarrow j}$ is the cross-covariance between two sellers' realized prices, the limiting prices cannot be determined from primitives without additional assumptions. Figure~\ref{fig:pe-continuum-region-revenue} confirms this non-identifiability by simulating a duopoly and plotting the scattered long-run average prices and surplus-capture ratios; this explains Figure~\ref{fig:samplepaths}. 

\begin{figure}[ht]
\centering
\begin{minipage}[t]{.49\textwidth}
\centering
{\small\textbf{Symmetric exploration}}\\[2pt]
\begin{subfigure}[t]{.49\linewidth}
\centering
\includegraphics[width=\linewidth]{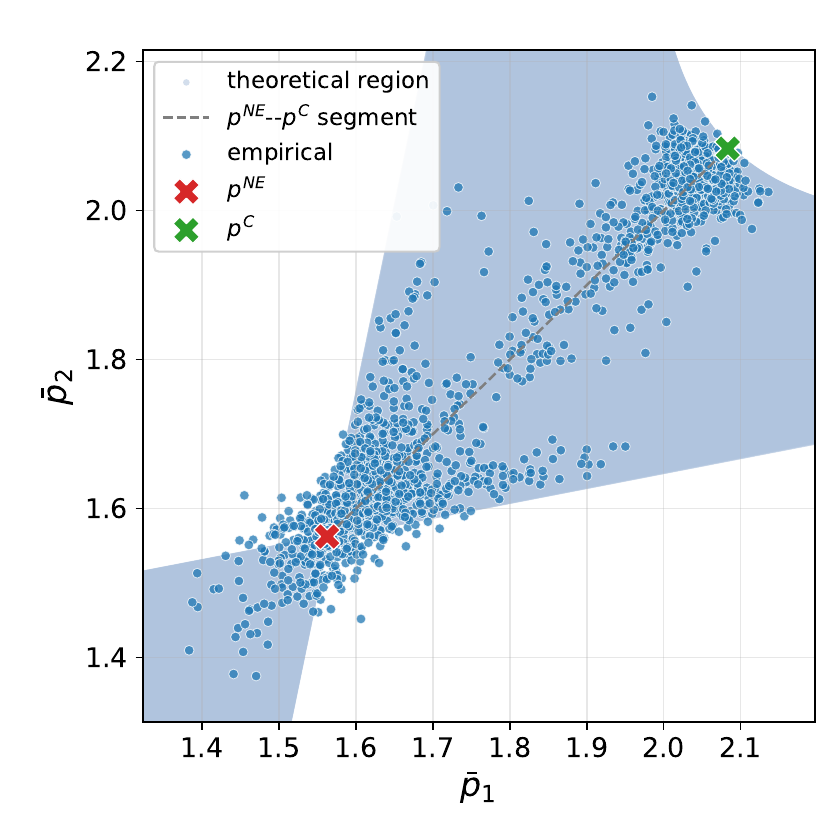}
\caption{Price.}
\label{fig:pe-continuum-sym-price}
\end{subfigure}\hfill
\begin{subfigure}[t]{.49\linewidth}
\centering
\includegraphics[width=\linewidth]{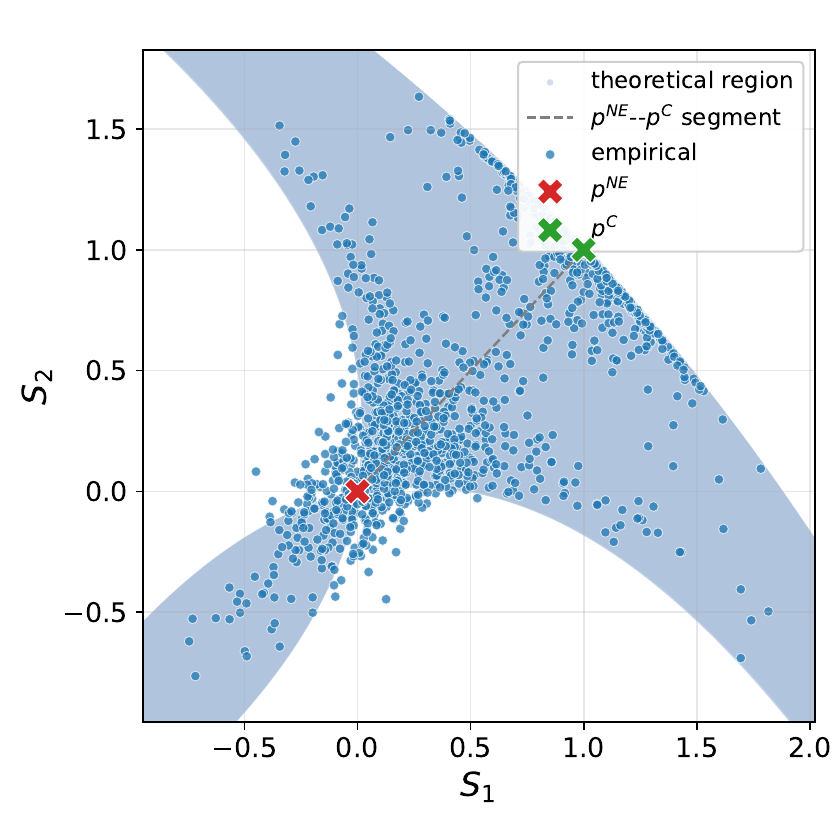}
\caption{Surplus capture.}
\label{fig:pe-continuum-sym-revenue}
\end{subfigure}
\end{minipage}\hfill
\begin{minipage}[t]{.49\textwidth}
\centering
{\small\textbf{Asymmetric exploration}}\\[2pt]
\begin{subfigure}[t]{.49\linewidth}
\centering
\includegraphics[width=\linewidth]{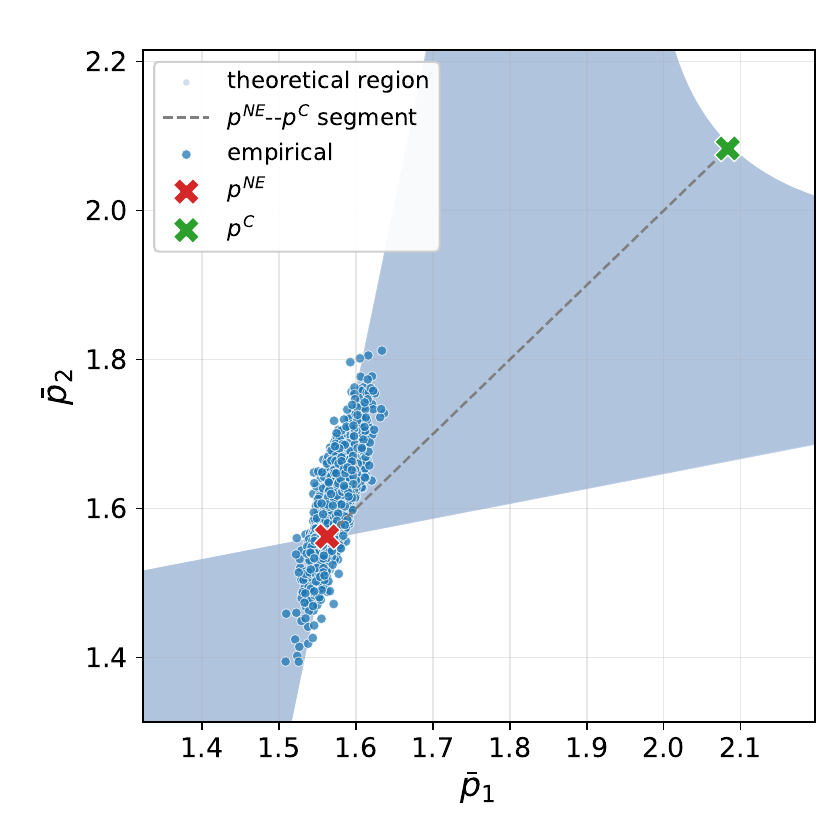}
\caption{Price.}
\label{fig:pe-continuum-asym-price}
\end{subfigure}\hfill
\begin{subfigure}[t]{.49\linewidth}
\centering
\includegraphics[width=\linewidth]{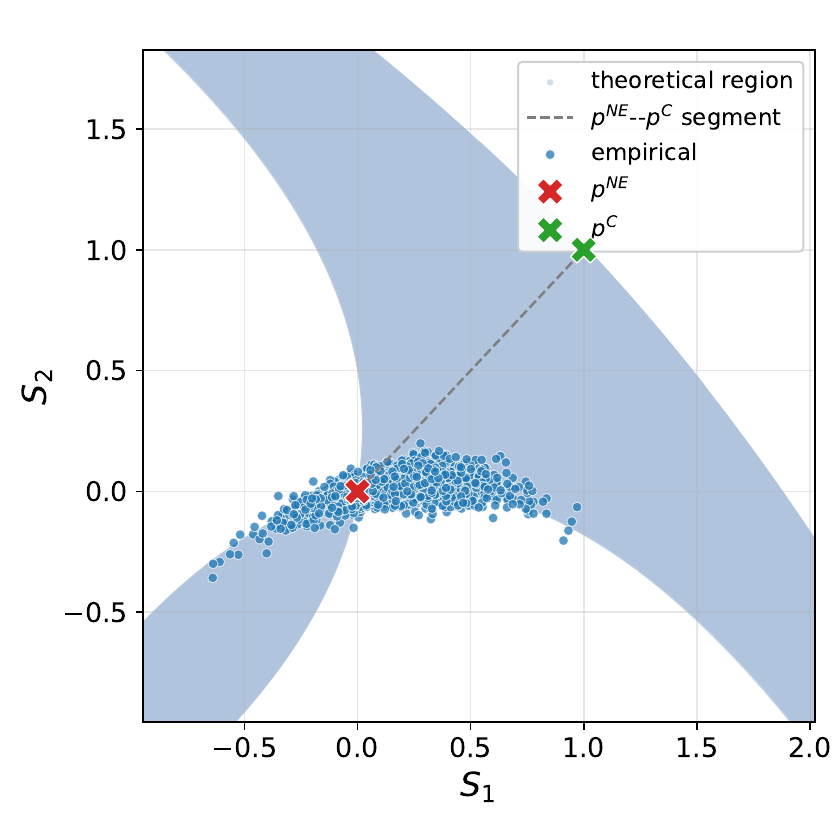}
\caption{Surplus capture.}
\label{fig:pe-continuum-asym-revenue}
\end{subfigure}
\end{minipage}
\caption{Long-run average prices (panels (a), (c)) and surplus-capture ratios $(S_1,S_2)$ (panels (b), (d)) in a symmetric duopoly with $\alpha = 2.5$, $\beta = -1$, $\gamma = 0.4$, $[l, u] = [0.5, 2.5]$, and horizon $T = 8\times 10^4$, started from a diverse spread of warm-up prices. Shaded is the theoretically possible region of pseudo-equilibria; red and green markers are the competitive and collusive benchmarks, respectively. The exploration schedules are $\nu_{n,i}^2 = 0.3(n+1)^{-\eta_i}$, where $\eta_1 = \eta_2 = 0.85$ in panels (a)--(b) and $\eta_1 = 0.2$, $\eta_2 = 0.9$ in panels (c)--(d). Following \citet[\S4.3]{cooperLearningPricingModels2015}, the shaded region corresponds to the set of admissible regression-ratio pairs $(r_1, r_2)$---those satisfying $r_i < -\beta_i / \gamma_{i,j}$ and either $r_1 = r_2 = 0$ or $0 < r_1 r_2 \le 1$---mapped to per-period revenue and surplus-capture.}
\label{fig:pe-continuum-region-revenue}
\end{figure}

\paragraph{Symmetric exploration and collusive outcomes.}
A notable observation arises from the comparison of Figure~\ref{fig:pe-continuum-region-revenue}'s symmetric panels (a)--(b) and asymmetric panels (c)--(d). When sellers have symmetric exploration schedules, the cloud concentrates on the $\bar p_1=\bar p_2$ diagonal, with the majority of seeds settling above $\mathbf p^{NE}$ toward $\mathbf p^C$, suggesting a tendency toward collusive outcomes. However, once this symmetry is broken, the cloud becomes tightly clustered around Nash, and near-collusive limits do not appear. This suggests that collusive outcomes are indeed possible yet brittle because they are artifacts of exploration symmetry, which is unlikely in practice without explicit coordination. Notably, the continuum and its collapse under variance dominance are not tied to OLS with additive dithering: Appendix~\ref{sec:robustness-learning-rule} (Figure~\ref{fig:robustness-continuum}) reproduces both mechanisms under ridge-regularized least squares and $\varepsilon$-greedy exploration. Together, these observations motivate a closer examination of the market outcome under asymmetric exploration.

\paragraph{Regret as the appropriate metric.}
We note that, because long-run competitor prices remain random and a seller's absolute revenue is tied to the realized competitor profile, surplus capture obfuscates the seller's decision quality in the regime of Theorem~\ref{thm:cesaro-main}. Indeed, panels (c)--(d) of Figure~\ref{fig:pe-continuum-region-revenue} exhibit a sizable dispersion of seller-2 prices and thus surplus-capture ratios. Intuitively, when competitors settle at low prices, a seller cannot hope to elevate its revenue unilaterally, and the best thing it can do is to best-respond and capture what is achievable. As a result, \textbf{regret} is a more appropriate metric in this regime, especially since prices are observed to settle near Nash under asymmetric exploration (see Section~\ref{sec:performance-metric} for more discussion). We therefore analyze the market outcome in regret terms in the rest of the section.

\subsection{Variance Dominance and the Spiral-up Phenomenon}\label{sec:variance-dominance}

\begin{figure}[!htbp]
\centering
\begin{subfigure}{.48\textwidth}
\centering
\includegraphics[width=\linewidth]{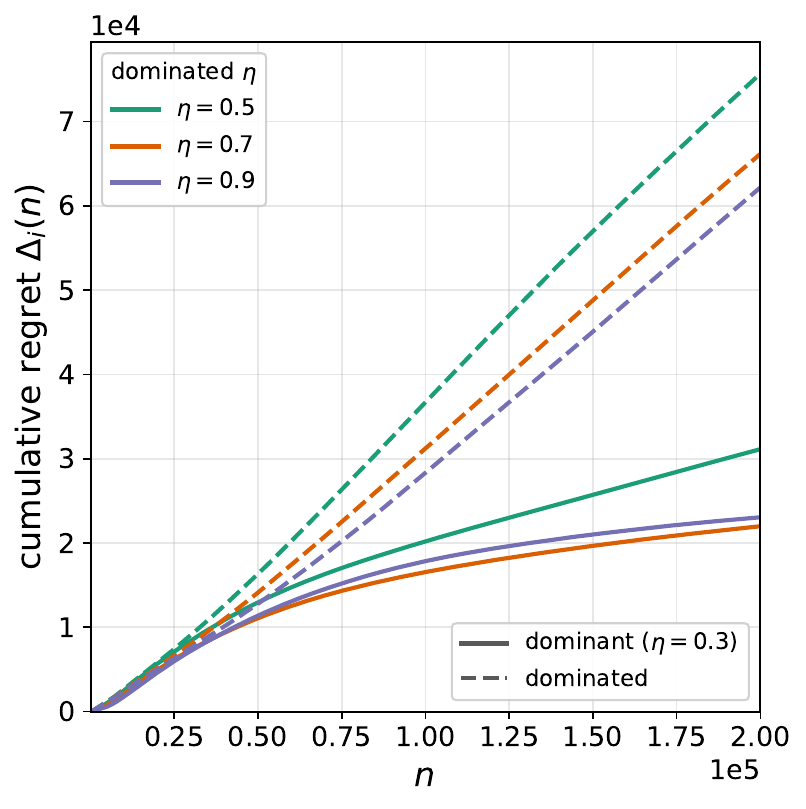}
\caption{$N = 2$.}
\label{fig:vd-N2}
\end{subfigure}\hfill
\begin{subfigure}{.48\textwidth}
\centering
\includegraphics[width=\linewidth]{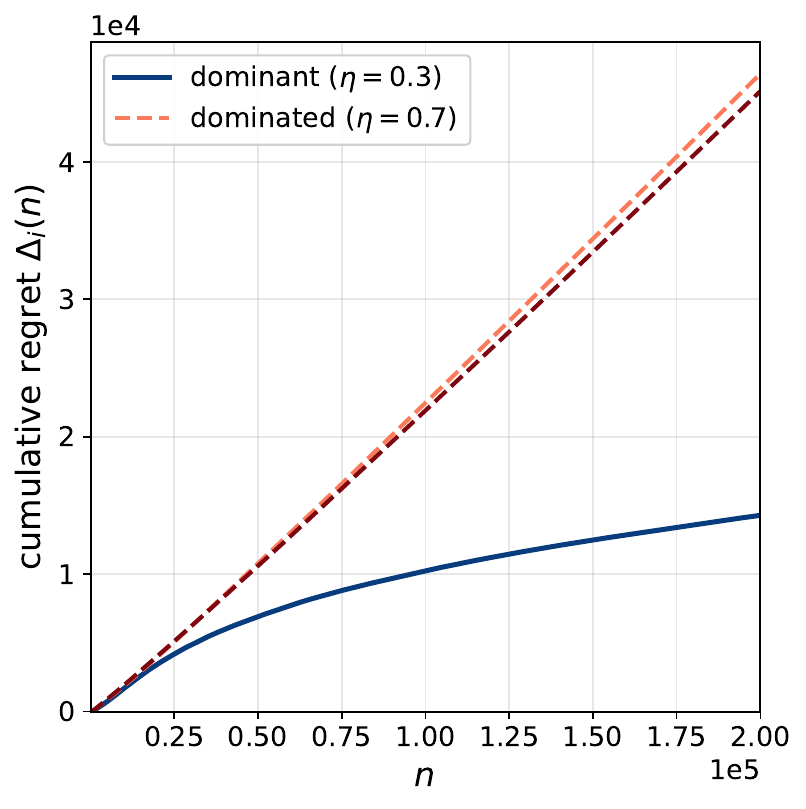}
\caption{$N = 3$.}
\label{fig:vd-N3}
\end{subfigure}
\caption{The regret dichotomy under variance dominance (Theorem~\ref{thm:multi-seller-VD}). All sellers use $\nu_{n,i}^2 = 0.05\,(n+1)^{-\eta_i}$. \emph{Left:} symmetric duopoly with $\alpha_i = 2.5$, $\beta_i = -1$, $\gamma_{i,j} = 0.6$, $[l, u] = [0.5, 3.5]$; the dominant seller's exponent is fixed at $\eta = 0.3$ and the dominated exponent varies over $\eta \in \{0.5, 0.7, 0.9\}$; each color is one market, pairing a dominated seller with its dominant counterpart, and the dominated curves are labeled by their exploration exponent $\eta$. \emph{Right:} we use $\gamma_{i,j} = 0.3$ and simulate one dominant seller and two dominated sellers. For both panels horizon $T = 2\times 10^5$ per configuration. The dominated sellers' regret (dashed) grows linearly, while the dominant seller's regret (solid) grows sublinearly.}
\label{fig:variance-dominance-regret}
\end{figure}

To formalize the notion of asymmetric exploration, we consider the following event: there exists $j^\star \in [N]$ such that 
\begin{equation}\tag{Var-Dom}\label{eq:variance-dominance-cond}
J_{n,i}/J_{n,j^\star}\;\le\;(\log n)^{-(2+\epsilon_i)}\text{ eventually for some } \epsilon_i>0 \text{ for all }i\ne j^\star.
\end{equation}
On this event, we refer to seller $j^\star$ as the \emph{variance-dominant} seller and sellers $i \ne j^\star$ as \emph{variance-dominated}. Event~(Var-Dom) is only logarithmically stronger than $J_{n,i}=o(J_{n,j^\star})$, and it characterizes asymmetric exploration by letting seller $j^\star$'s cumulative price variance outgrow every competitor's. Intuitively, seller $j^\star$ can ``get ahead'' by collecting more information than its competitors. Figure~\ref{fig:variance-dominance-regret} verifies this intuition: the variance-dominant seller's regret grows \emph{sublinearly}, while every dominated seller's regret grows \emph{linearly}.

\paragraph{Regret dichotomy.}
The following theorem makes this dichotomy precise; proof can be found in Appendix~\ref{proof:thm:multi-seller-VD}. In words, the dominant seller's misspecified estimates become asymptotically correct modulo projections, allowing best-response to competitors that leads to \emph{sublinear} regret. In contrast, dominated sellers incur sublinear regret only if the cross-regression $\bar s_i$ converges to 0. Since Theorem~\ref{thm:cesaro-main} and Figure~\ref{fig:pe-continuum-region-revenue} suggest that the law of $\{r_{n,i\leftarrow j}\}_{j \ne i}$ occupies an $(N-1)$-dimensional continuum in the long run---in which $\bar s_i=0$ corresponds to a low-probability hyperplane---sublinear regret accords with a knife-edge coincidence. Therefore, a dominated seller is generally expected to incur \emph{linear} regret.

\begin{theorem}[Variance dominance]\label{thm:multi-seller-VD}
On the event that (Min-Var), (Max-Var), and (Var-Dom) hold, $\bar{p}_{n,i}\to\bar q_i$ for some $\bar q_i\in[l,u]$ almost surely for all $i$, and:
\begin{enumerate}[label=(\roman*)]
\item \emph{(Dominant seller)} If $(\alpha_{j^\star}+\sum_{i\ne j^\star}\gamma_{j^\star,i} \bar q_i,\ \beta_{j^\star})\in\mathrm{int}(\Theta_{j^\star}^{ob})$, then $\tilde p_{n,j^\star}\to\mathrm{BR}_{j^\star}(\bar q_{-j^\star})$ almost surely; in particular, $\Delta_{j^\star}(T)=o(T)$.
\item \emph{(Dominated sellers)} For each $i\ne j^\star$ whose parameter projection is eventually slack, the long-run aggregate cross-regression
\[
\bar s_i\;=\;\lim_{T\to\infty}\frac1T\sum_{n<T}\sum_{j\ne i}\gamma_{i,j}\,r_{n,i\leftarrow j}
\]
exists, and $\Delta_i(T)=o(T)$ if and only if $\bar s_i=0$.
\end{enumerate}
\end{theorem}

\paragraph{A spiral-up phenomenon: linear exploration rate.}
The variance dominance mechanism in Theorem~\ref{thm:multi-seller-VD} and Figure~\ref{fig:variance-dominance-regret} introduces \emph{instability} to the pseudo-equilibria in Theorem~\ref{thm:cesaro-main} and Figure~\ref{fig:pe-continuum-region-revenue}, because a strategic seller would explore more aggressively in the hope of gaining a dominant position. If all sellers behave as such, a \emph{spiral-up} logic emerges: the only way to reliably avoid being dominated is to explore at the maximal sustainable order, which is linear,
\[
J_{n,i} = \Theta(n),
\]
since any $\omega(n)$ growth would require unbounded fluctuations and persistent clipping. By Lemma~\ref{lem:Jn-lower}, a linear rate is implementable with small, non-diminishing perturbations (e.g., $z_{n,i}\sim\mathrm{Unif}(-c,c)$ for $c>0$). Crucially, this linear rate is not without its cost: sellers would pay a persistent exploration tax of $|\beta_i|\nu_{n,i}^2$ per period (Lemma~\ref{lem:asymp-revenue}). The rate is thus \emph{insurance} against being variance-dominated and the cost of a \emph{vicious competitive spiral}.

\section{Convergence to the Competitive Outcome}\label{sec:convergence}

To understand how oblivious learning compares to informed, we need to study the market outcome when (i) all sellers are oblivious, (ii) all sellers are informed, and (iii) some sellers are oblivious while others are informed (i.e., a mixed market). In this section, we present a convergence result for each of these three market structures through the lens of price convergence, regret, and surplus capture. Numerical evidence for price convergence and regret is provided along with each theorem, while a consolidated numerical study of surplus capture across the three market compositions can be found in Appendix~\ref{sec:meta-revenue}. 

\paragraph{Heterogeneous exploration rates.}
To allow for heterogeneous exploration behavior, throughout this section suppose that
\[
\nu_{n,i}^2 = c_i\, n^{-\eta_i}, \quad 0 \le \eta_i < 1, \quad \eta_{\min} = \min_{i\in[N]}\eta_i, \quad \eta_{\max} = \max_{i\in[N]}\eta_i,
\]
with $c_i>0$. This assumption is not necessary for our results, but it allows a clean ordering of exploration rates among sellers and is thus adopted for expositional clarity. While oblivious sellers explore persistently ($\eta_i = 0$) per the spiral-up prediction of Section~\ref{sec:variance-dominance}, we assume $\eta_j > 0$ for every informed seller $j \in \mathcal{I}^{in}$ since informed sellers do not suffer from model misspecification. Without loss of generality, assume that the first $N+1$ prices produce a full-rank empirical Fisher information matrix for each informed seller, allowing least squares thereafter. 

\paragraph{Forecast rules.}
As explained in Section~\ref{sec:learning-dynamics}, an informed seller needs a forecast rule to map competitors' past prices into a prediction of next-period prices to compute the greedy response \eqref{eq:greedy-prices}. A comprehensive survey of possible forecast rules is beyond the scope of this paper, and so we adopt the running-mean forecast rule $\hat p_{n+1, j} = m_{n, j}$, which is operationally simple and robust to competitors' exploration perturbations; this is also the canonical choice for the classical adaptive-best-response results in game theory (Section~\ref{sec:literature}). To complement the running-mean forecast, in Appendix~\ref{sec:stackelberg-benchmark} we also consider a clairvoyant forecast rule as an informational upper bound, which further strengthens the conclusion of the strategy game (Section~\ref{sec:meta-strategy}). 

\smallskip\noindent\textbf{Interior limit.}
For ease of exposition, we assume that parameter projections are slack at the convergence limits considered in this section. In particular, we assume that the true parameters for informed sellers are in the interior of the feasible set, i.e., $(\alpha_i, \gamma_{i,1}, \ldots, \gamma_{i,i-1}, \beta_i, \gamma_{i,i+1}, \ldots, \gamma_{i,N})^\top \in \mathrm{int}(\Theta_i^{in})$ for all $i\in \mathcal{I}^{in}$. Similarly, even though oblivious sellers estimate misspecified two-parameter models, at the competitive benchmark they admit natural ``pseudo-true'' targets 
$$
\theta_i^{*, \,ob} \triangleq (a_i^*,b_i^*)= \left(\alpha_i + \sum_{j\neq i}\gamma_{i,j}p_j^{NE},\ \beta_i\right);
$$
we assume $\theta_i^{*, \,ob} \in \mathrm{int}(\Theta_i^{ob})$ for all $i \in \mathcal{I}^{ob}$.

\subsection{All-Oblivious Markets}\label{sec:global-convergence}

Let $L_\phi^{ob}$ denote the Lipschitz constant of the oblivious pricing map $\phi^{ob}(\cdot)$ over $\bigcup_i\Theta_i^{ob}$, $C_x$ be a uniform bound on $\|x_{n,i}^{ob}\|_2$, and $\mathcal{F}_n \triangleq \sigma\{(p_{m,i},d_{m,i}) : m\le n,\ i\in[N]\}$ be the filtration generated by the entire price-demand history up to time $n$. Let $\Gamma^{\mathrm{off}}$ denote the off-diagonal part of the Nash matrix $\Gamma$ (Section~\ref{sec:solution-concept}), i.e., $\Gamma^{\mathrm{off}}_{ii} = 0$ and $\Gamma^{\mathrm{off}}_{ij} = \gamma_{i,j}$ for $i\ne j$, and define $\bar\gamma \triangleq \tfrac{1}{2}(\lVert\Gamma^{\mathrm{off}}\rVert_\infty + \lVert\Gamma^{\mathrm{off}}\rVert_1)$. Because oblivious sellers' exploration does not decay, the regressors are \emph{persistently excited} \citep{goldenshluger2013linear, bastani2021mostly}: by Lemma~\ref{lem:yy-lower}, the constant exploration variance guarantees a constant $C_M>0$ with
\begin{equation}\label{eq:globalconvergencespectrallowerboundcondition}
\mathbb{E}\!\left[x_{n,i}^{ob}(x_{n,i}^{ob})^\top\mid\mathcal{F}_{n-1}\right] \succeq C_M I_{2} \qquad \text{for all } n \text{ and } i.
\end{equation}
The next theorem establishes global convergence to the competitive outcome when exploration is sufficiently strong relative to model misspecification; its proof is in Appendix~\ref{proof:thm:globalconvergencetocompetitiveoutcome}.

\begin{theorem}[\textbf{Global convergence to the competitive outcome}]\label{thm:globalconvergencetocompetitiveoutcome}
Suppose that 
\begin{equation}\label{eq:globalconvergencecondition}
\rho \triangleq \bar\gamma L_\phi^{ob} C_x / C_M\;<\; 1.
\end{equation}
Then:
\begin{enumerate}[label=(\roman*)]
\item $\hat\theta_{n,i}^{ob}\to \theta_i^{*,ob}$ and $\tilde p_{n,i}\to p_i^{NE}$ a.s.\ for every $i\in[N]$;
\item the aggregate learning regret and exploration tax satisfy
$$
\mathrm{\mathbf{RegLe}}(T)
=
\begin{cases}
O\!\left(\log T\right), & \rho < 1/2,\\[3pt]
O\!\left((\log T)^2\right), & \rho = 1/2,\\[3pt]
O\!\left(T^{2\rho-1}\right), & 1/2 < \rho < 1,
\end{cases}
$$
and $\mathrm{\mathbf{RegEx}}(T)=\Theta(T)$, so the total regret $\Delta_i(T)=\Theta(T)$ for every $i\in[N]$;
\item the surplus capture $S_i = -\,|\beta_i|\nu_i^2 / (\Pi_i^{C}-\Pi_i^{NE}) < 0$ for every $i\in[N]$ a.s.
\end{enumerate}
\end{theorem}

\paragraph{Discussion.}
The quantity $\bar\gamma L_\phi^{ob} C_x$ in~\eqref{eq:globalconvergencecondition} upper-bounds the \emph{strategic feedback loop} created by omitted competitor prices, and $C_M$ measures exploration strength via the size of the regressor covariance. When the strength of exploration dominates the effect of model misspecification, $\rho < 1$ and the induced bias behaves like a \emph{controlled perturbation}, contracting the dynamics toward $\mathbf p^{NE}$. The learning-induced portion of regret is thus sublinear, but the persistent-exploration term contributes a $\Theta(T)$ tax (Proposition~\ref{prop:dynamicbenchmark}) and a negative surplus capture. Though undesirable, this $\Theta(T)$ rate is seller-controlled and of different nature from that of dominated sellers in Theorem~\ref{thm:multi-seller-VD}. 

\begin{figure}[!htbp]
\centering
\begin{subfigure}{.48\textwidth}
\centering
\includegraphics[width=\linewidth]{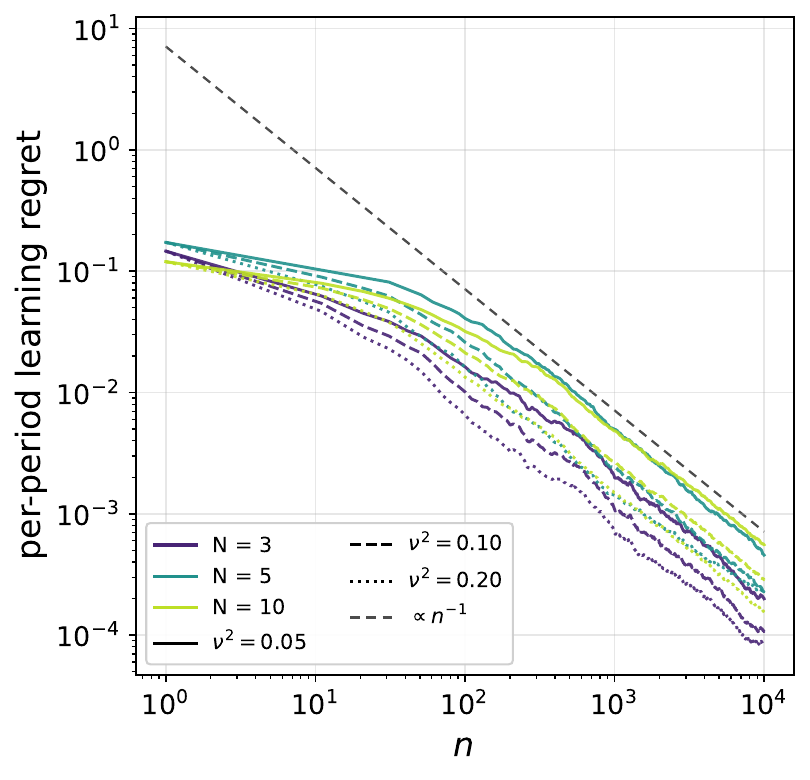}
\caption{Asymmetric markets, $N\in\{3,5,10\}$.}
\label{fig:global-convergence-stress-test-asym}
\end{subfigure}\hfill
\begin{subfigure}{.48\textwidth}
\centering
\includegraphics[width=\linewidth]{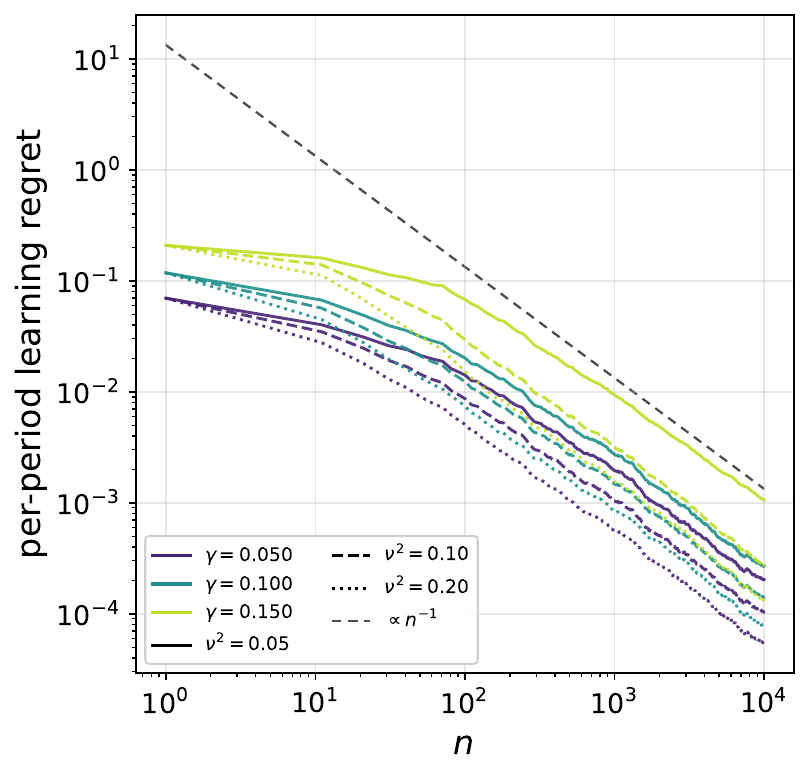}
\caption{Symmetric $N=5$, varying $\gamma$.}
\label{fig:global-convergence-stress-test-dom}
\end{subfigure}
\caption{Empirical stress test of Theorem~\ref{thm:globalconvergencetocompetitiveoutcome}: aggregate per-period learning regret (the summand of $\mathrm{\mathbf{RegLe}}(T)$) on log--log axes, overlaid with the predicted $n^{-1}$ reference. \emph{Left:} asymmetric markets at $N \in \{3, 5, 10\}$ and $\nu^2 \in \{0.05, 0.10, 0.20\}$, with heterogeneous demand primitives drawn once per $N$ from $\alpha_i \sim \mathcal N(2.5, 0.4^2)$, $\beta_i \sim \mathcal N(-1, 0.2^2)$, and $\gamma_{i,j} = (0.4/(N-1)) \cdot U[0.7, 1.3]$ (the Gaussian draws are clipped to bounded intervals). \emph{Right:} symmetric $N = 5$ markets with $\alpha_i = 2.5$, $\beta_i = -1$, $\gamma \in \{0.05, 0.10, 0.15\}$. Both panels use $[l, u] = [0.5, 2.5]$, and $T = 10^4$ per configuration. Across all 18 configurations the sufficient condition~\eqref{eq:globalconvergencecondition} is violated by roughly two orders of magnitude, yet every trajectory decays at the asymptotic $n^{-1}$ rate, so the cumulative learning regret is $O(\log T)$.}
\label{fig:global-convergence-stress-test}
\end{figure}

\paragraph{Incomplete learning.}
That persistent exploration collapses a continuum of pseudo-equilibria to a single competitive outcome presents a competitive analog of the \emph{incomplete learning} phenomenon documented by \citet{denBoerSimultaneouslyLearningOptimizing2014} and \citet{keskinIncompleteLearningCertaintyEquivalence2018} in monopolistic dynamic pricing: when exploration is insufficient, the parameters are never learned and long-run outcomes settle at self-fulfilling, sample-path-dependent prices (cf. Figure~\ref{fig:pe-continuum-region-revenue}) instead of the true optimum. Read through this lens, possible supracompetitive prices of an all-oblivious market are better understood as \emph{failures to learn} than as emergent collusion. 

\paragraph{Beyond the small-gain condition.}
Figure~\ref{fig:global-convergence-stress-test} stress-tests condition~\eqref{eq:globalconvergencecondition} over an asymmetric $N \in \{3, 5, 10\}$ sweep (left panel) and a symmetric $N=5$ sweep over the cross-price coefficient $\gamma$ (right panel), each crossed with three exploration variances $\nu^2$. Across all $18$ configurations, $\bar\gamma\, L_\phi^{ob}\, C_x$ is of order $1$--$10$ while $C_M$ stays below $10^{-1}$---so $\rho$ is of order $10$--$100$ and condition~\eqref{eq:globalconvergencecondition} is violated by one to two orders of magnitude---yet every seller's per-period learning regret decays at the asymptotic $n^{-1}$ rate (so the cumulative learning regret is the $O(\log T)$ predicted in the strongest regime). The contraction condition is therefore conservative and not necessary in practice. A complementary \emph{mean-dynamics} analysis confirms the pattern: tracking the running price moments yields a deterministic skeleton ODE whose \emph{unique} equilibrium corresponds to the Nash equilibrium and is locally asymptotically stable, irrespective of whether~\eqref{eq:globalconvergencecondition} holds (Appendix~\ref{sec:appendix-ode-derivation}). Numerically, this convergence of the all-oblivious market also persists under ridge and $\varepsilon$-greedy learning (Appendix~\ref{sec:robustness-learning-rule}, Figure~\ref{fig:robustness-mse}), providing further evidence that the competitive outcome is a robust attractor of the market dynamics under persistent exploration.

\subsection{All-Informed Markets}\label{sec:all-informed}

We now consider the all-informed market, $\mathcal{I}^{in} = [N]$. An informed seller $i$'s correctly specified demand model has parameter $\theta_i^{*,in} \triangleq (\alpha_i,\gamma_{i,1},\ldots,\gamma_{i,i-1},\beta_i,\gamma_{i,i+1},\ldots,\gamma_{i,N})^\top \in \mathbb R^{N+1}$, so that $(x_{n,i}^{in})^\top \theta_i^{*,in}$ reproduces the expected demand. Throughout this subsection, let $\Gamma$ be the Nash matrix from Section~\ref{sec:solution-concept} and define the linearized best-response map around $\mathbf p^{NE}$,
\[
B \;\triangleq\; I - \frac{1}{2}\,\mathrm{diag}(1/\beta_i)\,\Gamma.
\]
Under the standing assumption $\gamma_i < -\beta_i$, $I - B$ is positive stable, so the Lyapunov equation $P(I-B) + (I-B)^\top P = I$ admits a unique symmetric positive-definite solution $P$; denote $\mu \triangleq 1/\sigma_{\max}(P) > 0$. The following result characterizes the convergence of realized prices to the competitive equilibrium; the proof is in Appendix~\ref{proof:thm:allinformedmeanforecast}.

\begin{theorem}[All-informed market]\label{thm:allinformedmeanforecast}
Suppose
\begin{equation}\label{eq:conditioninformedsellersmixed}
\eta_{\min} + 1 \;>\; 2\, \eta_{\max}.
\end{equation}
Then:
\begin{enumerate}[label=(\roman*)]
\item $\hat\theta_{n,k}^{in}\to\theta_k^{*,\,in}$ and $\tilde p_{n,k}\to p_k^{NE}$ a.s.\ for every $k\in[N]$;
\item the aggregate learning regret and exploration tax satisfy
\begin{equation}\label{eq:cor:allinformed-rate}
\mathrm{\mathbf{RegLe}}(T) = O\!\left(T^{1-\mu}(\log T)^3 + T^{\eta_{\max}}\log T\right), \qquad \mathrm{\mathbf{RegEx}}(T) = O\!\left(T^{1-\eta_{\min}}\right),
\end{equation}
where $T^{1-\mu}$ is read as $\mathrm{polylog}\,T$ when $\mu \ge 1$, so the total regret $\Delta_i(T)$ is sublinear for every $i\in[N]$;
\item the surplus capture $S_i = 0$ for every $i\in[N]$ a.s.
\end{enumerate}
Moreover, in the special case where $\eta_i = 1/2$ for every $i\in[N]$ and, additionally,
\begin{equation}\label{eq:cor:symmpart-regularity}
\lambda_{\max}(B + B^\top) \;<\; 3/2,
\end{equation}
every seller attains regret $\Delta_i(T) = O(\sqrt T\,\log T)$.
\end{theorem}

\paragraph{Discussion.}
The learning regret $\mathrm{\mathbf{RegLe}}(T)$ in~\eqref{eq:cor:allinformed-rate} aggregates a market-side rate (independent of the exploration schedule) and the parameter-estimation error (Theorem~\ref{thm:informedolsrate}), while the residual exploration variance is collected in the tax $\mathrm{\mathbf{RegEx}}(T)$. Because every seller explores at a \emph{decaying} rate ($\eta_i > 0$), both the learning regret and the tax are sublinear, and all sellers attain a surplus capture of 0. In the special case $\eta_i = 1/2$, the tax is $O(\sqrt T)$ and the learning regret $O(\sqrt T\,\log T)$, so every seller achieves the near-optimal $\tilde O(\sqrt T)$ regret, matching that for monopolistic dynamic pricing~\citep{denBoerSimultaneouslyLearningOptimizing2014, keskinDynamicPricingUnknown2014}. Thus, the all-informed market can be interpreted as a well-behaved market where sellers reach a stable outcome efficiently without paying a persistent, $\Theta(T)$ exploration tax on their revenue.

\paragraph{Interpreting the conditions.}
Condition~\eqref{eq:conditioninformedsellersmixed} is a regularity requirement on relative exploration rates to ensure that the learning signal (mean growth) dominates the noise (variance). Condition~\eqref{eq:cor:symmpart-regularity} is a requirement on the strength of cross-price interactions: it holds automatically in markets with symmetric demand primitives and whenever $N \le 3$, and it is implied by the collusive concavity condition $H \prec 0$ of Section~\ref{sec:solution-concept} when sellers share a common own-price slope $\beta$ (Appendix~\ref{proof:thm:allinformedmeanforecast}). 

\paragraph{Numerical verification.}
Figure~\ref{fig:allinformed-numerics} corroborates Theorem~\ref{thm:allinformedmeanforecast} in the special case $\eta_i = 1/2$; panel~(\ref{fig:EA-symm}) is a symmetric market at $N \in \{2, 5\}$, and panel~(\ref{fig:EA-asym}) is an asymmetric market at $N \in \{3, 5\}$ with heterogeneous demand primitives. In both panels the regularity condition~\eqref{eq:cor:symmpart-regularity} holds and the per-period learning regret decays at slope $\approx -1/2$ as predicted by Theorem~\ref{thm:allinformedmeanforecast}. Numerically, similar all-informed convergence holds under ridge and $\varepsilon$-greedy learning (Appendix~\ref{sec:robustness-learning-rule}, Figure~\ref{fig:robustness-mse}).

\begin{figure}[!htbp]
\centering
\begin{subfigure}{.48\textwidth}
\centering
\includegraphics[width=\linewidth]{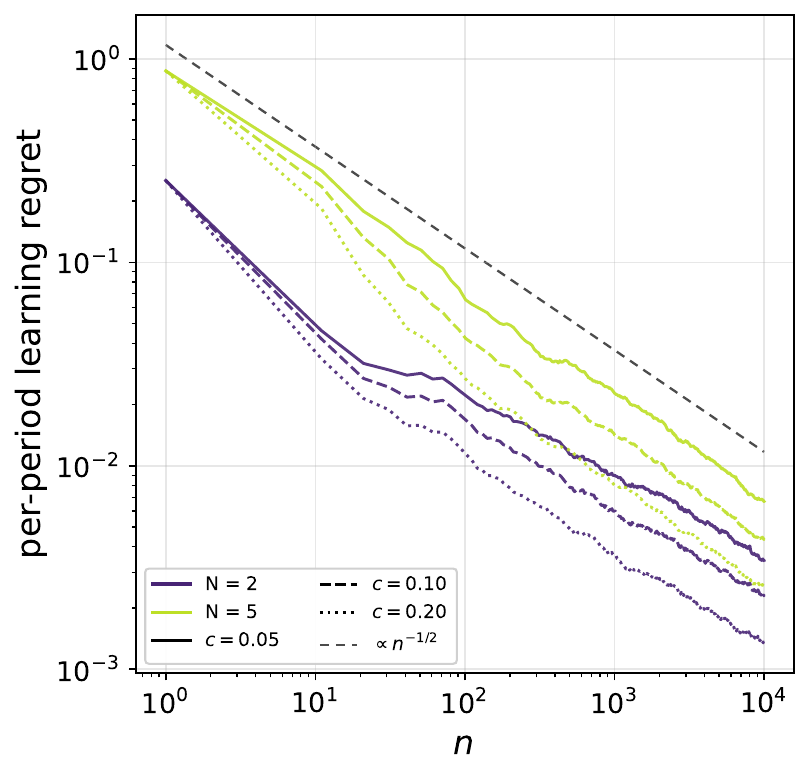}
\caption{Symmetric all-informed, $N \in \{2, 5\}$.}
\label{fig:EA-symm}
\end{subfigure}\hfill
\begin{subfigure}{.48\textwidth}
\centering
\includegraphics[width=\linewidth]{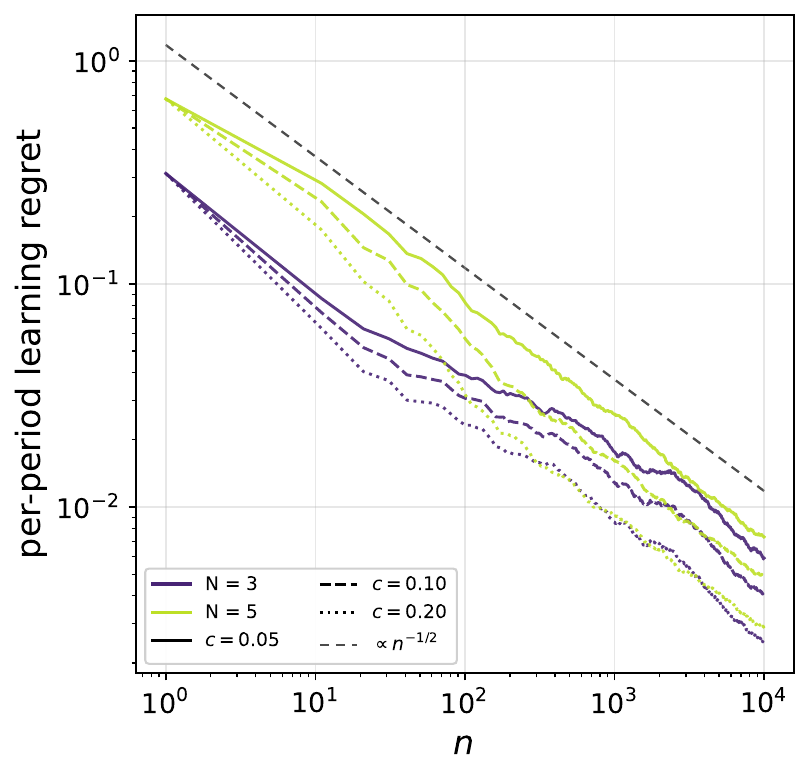}
\caption{Asymmetric all-informed, $N \in \{3, 5\}$.}
\label{fig:EA-asym}
\end{subfigure}
\caption{Empirical corroboration of Theorem~\ref{thm:allinformedmeanforecast} in the special case $\eta_i = 1/2$. Both panels overlay the aggregate per-period learning regret on log--log axes, with the predicted $n^{-1/2}$ reference, and use horizon $T = 10^4$. \emph{Left:} symmetric markets at $N \in \{2, 5\}$ with $\alpha_i = 2.5$, $\beta_i = -1$, $\gamma_{i,j} = 0.4/(N-1)$ and sweeping $\nu^2 \in \{0.05, 0.10, 0.20\}$. \emph{Right:} asymmetric markets at $N \in \{3, 5\}$ with heterogeneous primitives drawn as in Figure~\ref{fig:global-convergence-stress-test} (one draw per $N$), same $\nu^2$ grid. In both, regularity~\eqref{eq:cor:symmpart-regularity} holds and the per-period learning regret decays at slope $\approx -1/2$ (so the cumulative regret is $O(\sqrt T\,\log T)$).}
\label{fig:allinformed-numerics}
\end{figure}

\subsection{Mixed Markets}\label{sec:mixedmarkets}

We now consider mixed markets where $\mathcal{I}^{ob} \cup \mathcal{I}^{in} = [N]$ and both subsets are non-empty. Before stating the result, we introduce two aggregate constants beyond the primitives $C_M$, $C_x$, and $\bar\gamma$ inherited from Theorem~\ref{thm:globalconvergencetocompetitiveoutcome}. The \emph{greedy-map Lipschitz envelope} $L \triangleq \max\{L_\phi^{ob},L_\phi^{in,\theta}\}$ combines the oblivious greedy-map Lipschitz constant $L_\phi^{ob}$ with $L_\phi^{in,\theta}$, the Lipschitz constant of $\phi_j^{in}(\theta;\,\mathbf m)$ in $\theta$ over the informed projection boxes and competitor prices $\mathbf m \in [l,u]^{N-1}$. The \emph{informed-coupling magnitude} $\bar\mu \triangleq \tfrac12(\lVert B\rVert_\infty + \lVert B\rVert_1)$, where $B$ is from Section~\ref{sec:all-informed}, captures the size of $\gamma_{i,j}/(2|\beta_i|)$, the sensitivity of the mean-price forecast rule. The following theorem establishes convergence to the competitive outcome in mixed markets; the proof is in Appendix~\ref{proof:thm:mixedmarketconvergence}.

\begin{theorem}[\textbf{Mixed-market convergence}]\label{thm:mixedmarketconvergence}
Suppose:
\begin{itemize}
\item[(i)] (Exploration regularity, informed side) $\eta_{\max} < 1/2$;
\item[(ii)] (Small gain, informed side) $L + 2\bar\mu \;<\; 2$;
\item[(iii)] (Small gain, oblivious side) 
\[
C_M \;>\; 4\,C_x\,\bar\gamma\,\left[L \;+\; \frac{2- 2\bar\mu}{2 - 2\bar\mu - L} \cdot \bar\mu\right].
\]
\end{itemize}
Then:
\begin{enumerate}[label=(\roman*)]
\item $\hat\theta_{n,i}^{ob} \to \theta_i^{*,ob}$, $\hat\theta_{n,j}^{in} \to \theta_j^{*,in}$, and $\tilde p_{n,k} \to p_k^{NE}$ a.s.\ for every $i \in \mathcal{I}^{ob}$, $j \in \mathcal{I}^{in}$, and $k \in [N]$;
\item there exists $\eta^* \in (0, 1 - \eta_{\max})$ such that the aggregate learning regret satisfies $\mathrm{\mathbf{RegLe}}(T) = \tilde O(T^{1-\eta^*})$, where $\tilde O$ hides logarithmic factors, while the exploration tax $\mathrm{\mathbf{RegEx}}(T)=\Theta(T)$, so the total regret $\Delta_i(T) = \Theta(T)$ for every $i \in [N]$.
\item the surplus capture $S_i = -\,|\beta_i|\nu_i^2/(\Pi_i^{C}-\Pi_i^{NE}) < 0$ for every $i\in\mathcal{I}^{ob}$ and $S_j = 0$ for every $j\in\mathcal{I}^{in}$ a.s.
\end{enumerate}
\end{theorem}

\paragraph{Discussion.}
As in Theorem~\ref{thm:globalconvergencetocompetitiveoutcome}, oblivious sellers exhibit near-Nash asymptotic prices, pay a persistent exploration tax on their surplus capture, and incur linear regret. The informed sellers, on the other hand, learn the true demand model \emph{despite} the presence of oblivious competitors, capture the Nash surplus ($S_j=0$), but nevertheless incur linear regret indirectly as their exploration tax carries the persistent variance $\nu_{n,k}^2 = \Theta(1)$ of oblivious neighbors. Surplus capture is therefore a more appropriate performance metric than regret in mixed markets: it accurately reflects sellers' revenue outcomes---the primary economic quantity of interest---while regret penalizes informed sellers for non-settling oblivious prices, for which informed sellers are not responsible nonetheless.

\paragraph{Interpreting the conditions.}
Condition~(i) is~\eqref{eq:conditioninformedsellersmixed} from Theorem~\ref{thm:allinformedmeanforecast} specialized to $\eta_{\min} = 0$ for the persistently exploring oblivious sellers. Condition~(ii) is the informed-side small gain: the greedy-map sensitivity $L$ and the informed-coupling magnitude $\bar\mu$ must sit below $2$, the diagonal contraction rate produced by the quadratic Lyapunov we use in the proof. Condition~(iii) is the oblivious-side small gain: the persistent excitation $C_M$ must dominate both the greedy-map sensitivity $L$ and the informed-feedback effect $\bar\mu$. 

\paragraph{Numerical experiments.}
Figure~\ref{fig:mixedmarketmaster-numerics} stress-tests Theorem~\ref{thm:mixedmarketconvergence} in three settings: a symmetric $N = 5$ mixed market sweeping $(\gamma, \nu^2)$ (panel~(\ref{fig:M1})), an asymmetric mixed market at $N \in \{3, 5, 10\}$ with one oblivious and $N-1$ informed sellers (panel~(\ref{fig:M2}))---in these two panels we make $L > 2$, so condition~(ii) fails---and a symmetric $N = 5$ market with tuned primitives so condition~(ii) holds but (iii) fails (panel~(\ref{fig:M1b})). Across all configurations the per-period learning regret decays on log--log axes (fitted negative slope), indicating that the small-gain conditions are sufficient but not necessary in practice. Numerically, mixed-market convergence likewise persists under ridge and $\varepsilon$-greedy learning (Appendix~\ref{sec:robustness-learning-rule}, Figure~\ref{fig:robustness-mse}).

\begin{figure}[!htbp]
\centering
\begin{subfigure}{.32\textwidth}
\centering
\includegraphics[width=\linewidth]{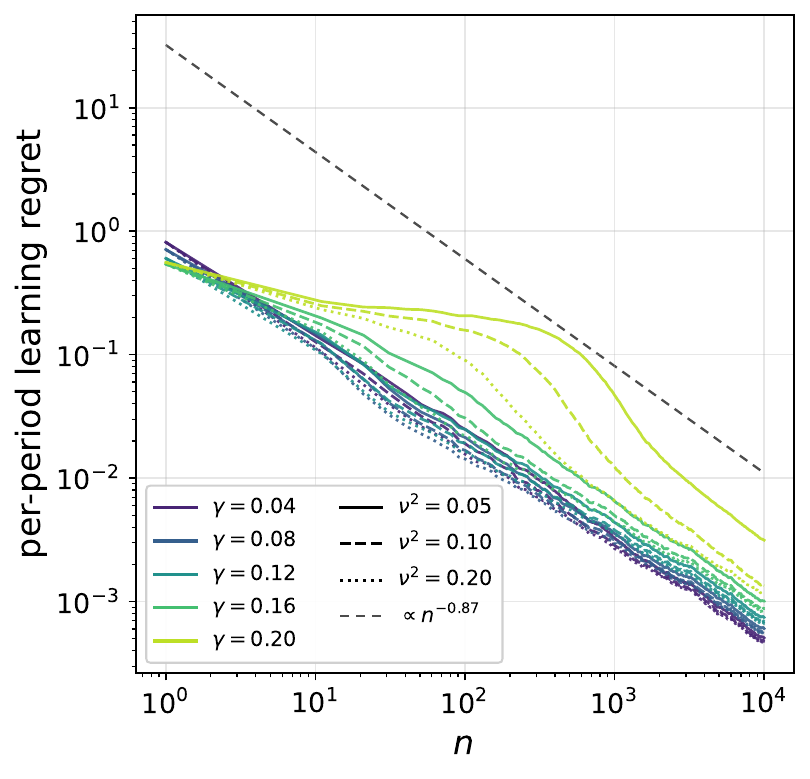}
\caption{Symm., (ii) and (iii) both fail.}
\label{fig:M1}
\end{subfigure}\hfill
\begin{subfigure}{.32\textwidth}
\centering
\includegraphics[width=\linewidth]{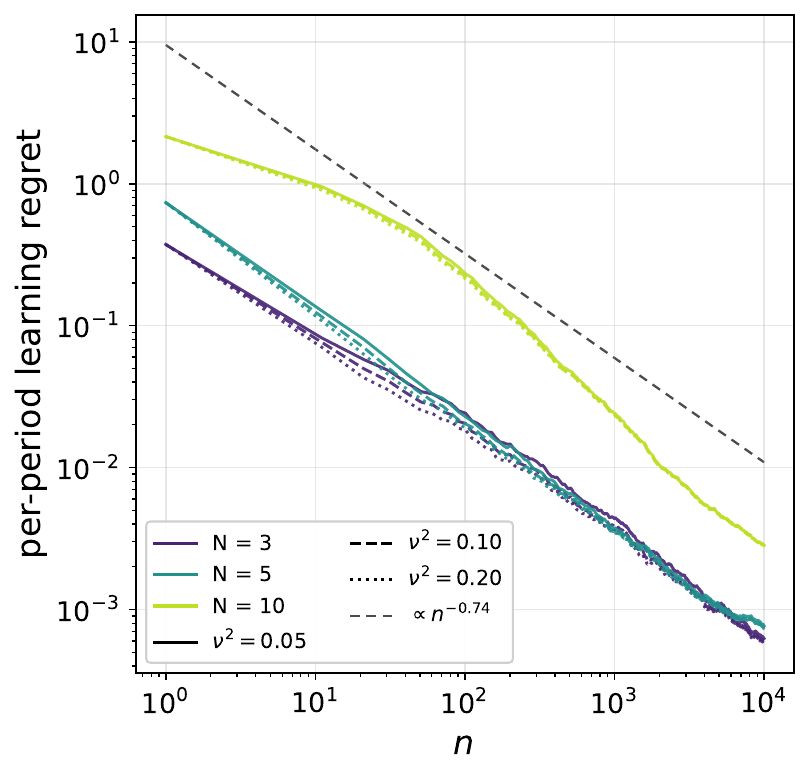}
\caption{Asymm., (ii) and (iii) both fail.}
\label{fig:M2}
\end{subfigure}\hfill
\begin{subfigure}{.32\textwidth}
\centering
\includegraphics[width=\linewidth]{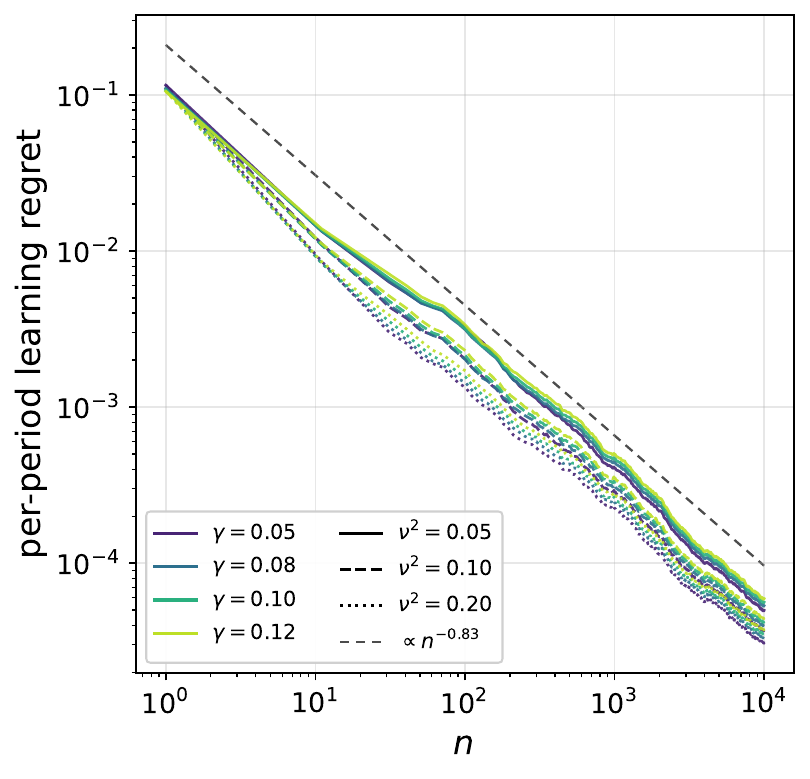}
\caption{Symm., (ii) holds, (iii) fails.}
\label{fig:M1b}
\end{subfigure}
\caption{Empirical verification of Theorem~\ref{thm:mixedmarketconvergence}: aggregate per-period learning regret on log--log axes, with a fitted power-law reference. All three panels: $\eta_j = 0.25$ for every informed seller and horizon $T = 10^4$ per configuration. \emph{Left:} $(\gamma, \nu^2)$ grid in a symmetric $N = 5$ market with $|\mathcal{I}^{ob}| = 2$, $|\mathcal{I}^{in}| = 3$, $(\alpha, \beta) = (2.5, -1)$. \emph{Middle:} $(N, \nu^2)$ grid at $N \in \{3, 5, 10\}$ with one oblivious and $N - 1$ informed sellers and heterogeneous primitives drawn as in Figure~\ref{fig:allinformed-numerics} (one draw per $N$). \emph{Right:} $(\gamma, \nu^2)$ grid in a symmetric $N = 5$ market with $|\mathcal{I}^{ob}| = 4$, $|\mathcal{I}^{in}| = 1$, $(\alpha, \beta, u) = (3.0, -2.0, 2.0)$, tuned so condition (ii) holds but (iii) fails. Across all configurations the per-period learning regret decays (fitted negative slope), so the cumulative regret is sublinear.}
\label{fig:mixedmarketmaster-numerics}
\end{figure}

\subsection{The Strategy Game}\label{sec:meta-strategy}
Sections~\ref{sec:oblivious-sellers} (\textsf{ob}--\textsf{ob}), \ref{sec:all-informed} (\textsf{in}--\textsf{in}), and~\ref{sec:mixedmarkets} (\textsf{ob}--\textsf{in}) characterize the three market compositions, which we now consolidate into a \emph{strategy game}. Each seller $i\in\{1,2\}$ chooses $s_i\in\{\textsf{oblivious},\textsf{informed}\}$, and their payoff is measured by both the regret $\Delta_i(T)$ and the asymptotic surplus capture $S_i(s_1,s_2)$ along the induced price path. Table~\ref{tab:meta-game} collects the resulting payoff pairs. A consolidated numerical study of surplus capture across the three market compositions can be found in Appendix~\ref{sec:meta-revenue}.

\begin{table}[!htbp]
\centering
\renewcommand{\arraystretch}{1.6}
\setlength{\tabcolsep}{8pt}
\begin{tabular}{c|r@{\;,\;}l|r@{\;,\;}l}
\hline
 & \multicolumn{2}{c|}{\textbf{\textsf{oblivious}}} & \multicolumn{2}{c}{\textbf{\textsf{informed}}} \\
\hline
\textbf{\textsf{oblivious}} &
$(\textcolor{red}{\Theta(T)}\;,\, \textcolor{red}{-\nu^2})$ & $(\textcolor{red}{\Theta(T)}\;,\, \textcolor{red}{-\nu^2})$ &
$(\textcolor{red}{\Theta(T)}\;,\, \textcolor{red}{-\nu^2})$ & $(\textcolor{red}{\Theta(T)}\;,\, \textcolor{green!50!black}{\;\;\;0\;\;})$ \\
\hline
\textbf{\textsf{informed}} &
$(\textcolor{red}{\Theta(T)}\;,\, \textcolor{green!50!black}{\;\;\;0\;\;})$ & $(\textcolor{red}{\Theta(T)}\;,\, \textcolor{red}{-\nu^2})$ &
$(\textcolor{green!50!black}{\boldsymbol{o(T)}}\;,\, \textcolor{green!50!black}{\;\;\;\boldsymbol{0}\;\;})$ & $(\textcolor{green!50!black}{\boldsymbol{o(T)}}\;,\, \textcolor{green!50!black}{\;\;\;\boldsymbol{0}\;\;})$ \\
\hline
\end{tabular}
\caption{The strategy game. Rows index seller~1's strategy, columns index seller~2's strategy; each cell reports the two sellers' payoff pairs $(\text{regret},\, S)$. Regret is $\Theta(T)$ in every cell except $(\textsf{informed}, \textsf{informed})$, where it is $o(T)$; surplus capture is $0$ for an informed seller and $-\nu^2 < 0$ for an oblivious seller (shorthand for $-|\beta_i|\nu_i^2/(\Pi_i^{C}-\Pi_i^{NE})$). The unique strict pure-strategy Nash equilibrium---also the unique Pareto-efficient profile---is (\textsf{informed}, \textsf{informed}).}
\label{tab:meta-game}
\end{table}

\paragraph{Unique strict Nash equilibrium.}
Since each seller's payoff has two coordinates (regret and surplus capture), we compare profiles by Pareto dominance. Starting from $(\textsf{informed},\textsf{informed})$, a unilateral deviation to \textsf{oblivious} moves the deviator from sublinear regret and zero surplus to linear regret and strictly negative surplus---worse in both coordinates---so the deviator's payoff is strictly Pareto-dominated by its equilibrium payoff; hence $(\textsf{informed},\textsf{informed})$ is a strict Nash equilibrium. It is also the unique equilibrium: from any other profile, the oblivious seller can switch to \textsf{informed}, leaving its regret no worse and raising its surplus from negative to zero---a strict improvement---so it strictly prefers to deviate and no other profile is an equilibrium. In this sense \textsf{informed} is the dominant strategy for each seller in every market composition, and the market is expected to stabilize at the all-informed outcome. 

\paragraph{Beyond duopoly.}
The two-seller strategy game extends cleanly to $N$-seller markets at every composition $|\mathcal{I}^{ob}| \in \{0, 1, \ldots, N\}$: the on-diagonal cells generalize via Theorems~\ref{thm:globalconvergencetocompetitiveoutcome} and~\ref{thm:allinformedmeanforecast}, and the off-diagonal cells via Theorem~\ref{thm:mixedmarketconvergence}. As a result, any unilateral deviation from the all-informed equilibrium moves the deviator to a strictly Pareto-dominated payoff, confirming the all-informed market as the unique strict pure-strategy Nash equilibrium of the $N$-seller strategy game.

\section{Discussion}\label{sec:discussion}

\paragraph{Summary.}
We study whether learning-based pricing algorithms should model competitors' prices in a competitive market with unknown, noisy demand. Our analysis spans the three market compositions---all-oblivious, all-informed, and mixed---and consolidates them into a strategy game over the modeling choice. The central message is that oblivious modeling does not robustly sustain supracompetitive prices: once demand noise and price exploration are present, an oblivious seller has an incentive to explore persistently despite the resulting exploration tax, and when sellers do so, the market outcome collapses from a continuum of pseudo-equilibria to the competitive equilibrium. In contrast, informed modeling reaches the same competitive outcome without this tax on their revenue and is better than oblivious modeling as a strategy. The managerial takeaway is direct: a seller is better off incorporating competitor prices into its demand model and conducting sufficient price exploration to keep learning effective, rather than relying on obliviousness as a tacit route to higher prices.

\paragraph{Limitations and Future Directions.}
Several directions warrant further investigation. First, we treat the modeling choice as binary---fully oblivious or fully informed---whereas sellers may acquire competitor information partially or at a cost; endogenizing this acquisition, and the exploration intensity alongside it, is a natural next step. On the learning side, our informed sellers rely on a running-mean forecast, and a better understanding of the forecast-rule design space remains open. Theoretically, our analysis applies most cleanly to linear demand and mean-zero dithering; while Appendix~\ref{sec:robustness-learning-rule} numerically shows that our conclusions are robust to ridge-regularized estimation and $\varepsilon$-greedy exploration, in-depth analyses of richer exploration designs, nonlinear or capacity-constrained demand, and more sophisticated learning algorithms may reveal additional insights. Overall, the interplay between model specification, learning dynamics, and market outcomes in competitive settings where algorithmic collusion is possible remains a rich and largely open area, and we hope this work provides a useful starting point for future research in this direction.

\section*{Acknowledgments}
We thank Santiago R. Balseiro, Nicola Rosaia, Hongyao Ma, Izzy Grosof, Vineet Goyal, and Julius Durmann for helpful comments and suggestions. We also thank the audience at NYC Operations Day 2025, Market Innovation Workshop 2025, INFORMS Revenue Management and Pricing Section Conference 2025, INFORMS Annual Meeting 2025, and ACM Conference on Economics and Computation 2026 for their comments on preliminary versions of this work.

\bibliography{bib}
\bibliographystyle{ims}

\appendix

\section{Additional Details}

\subsection{Mean-Dynamics ODE: Empirical Moments and Derivation}\label{sec:appendix-ode-derivation}

This appendix states the mean-dynamics ODE~\eqref{eq:ode-N-mQ} referenced in Section~\ref{sec:global-convergence} and its local-stability property (Theorem~\ref{thm:ode-local-stability}). Denote $\operatorname{Var}(z_{n,i})=\nu_i^2>0$. Track the price process through its empirical first and second moments
\[
m_{n,i}\triangleq \frac1n\sum_{s=1}^n p_{s,i},
\qquad
Q_{n,ij}\triangleq \frac1n\sum_{s=1}^n p_{s,i}p_{s,j},\qquad i,j\in[N],
\]
collected into $m_n\in\mathbb{R}^N$ and $Q_n\in\mathbb{R}^{N\times N}$. By the standard ODE method for stochastic approximation \citep{borkarStochasticApproximationDynamical2023}, the piecewise-linear interpolation of $(m_n,Q_n)$ tracks, on long time scales, the mean-dynamics ODE
\begin{equation}\label{eq:ode-N-mQ}
\begin{aligned}
\dot m_i(t) &= p_i^g(m(t),Q(t)) - m_i(t), \qquad i\in[N],\\
\dot Q_{ij}(t) &= p_i^g(m(t),Q(t))\,p_j^g(m(t),Q(t)) - Q_{ij}(t),\qquad i\neq j,\\
\dot Q_{ii}(t) &= (p_i^g(m(t),Q(t)))^2 + \nu_i^2 - Q_{ii}(t),\qquad i\in[N],
\end{aligned}
\end{equation}
where
\begin{align}
p_i^g(m,Q)&\triangleq-\frac{a_i(m,Q)}{2\,b_i(m,Q)}, \label{eq:pg-map-N-ode}\\ 
a_i(m,Q)&\triangleq\alpha_i+\sum_{j\neq i}\gamma_{i,j}m_j
\;-\;m_i\sum_{j\neq i}\gamma_{i,j}\frac{S_{ij}(m,Q)}{V_i(m,Q)}, \label{eq:a-map-N-ode}\\
b_i(m,Q)&\triangleq\beta_i+\sum_{j\neq i}\gamma_{i,j}\frac{S_{ij}(m,Q)}{V_i(m,Q)},\label{eq:b-map-N-ode}
\end{align}
and $S_{ij}(m,Q)\triangleq Q_{ij}-m_i m_j$ ($i\neq j$) and $V_i(m,Q)\triangleq Q_{ii}-m_i^2$. The oblivious projection onto $\Theta_i^{ob}$ is inactive near the Nash outcome and is omitted from the drift. The following result, proved in Appendix~\ref{proof:thm:ode-local-stability}, shows that this ODE has a unique equilibrium at Nash that is locally asymptotically stable.

\begin{theorem}[\textbf{Mean-dynamics local stability}]\label{thm:ode-local-stability}
The ODE \eqref{eq:ode-N-mQ} admits a unique equilibrium
\[
(m^*, Q^*) \;=\; \left(\mathbf p^{NE},\; \mathbf p^{NE}(\mathbf p^{NE})^\top + \mathrm{diag}(\nu_1^2,\ldots,\nu_N^2)\right),
\]
and it is locally asymptotically stable.
\end{theorem}

\subsection{Informed OLS Rate}\label{sec:informedolsrate-statement}

Consider a market in which at least one seller is informed. The following result describes the rate of convergence of the ordinary least squares (OLS) estimator for informed sellers and is the technical backbone for Theorems~\ref{thm:allinformedmeanforecast} and~\ref{thm:mixedmarketconvergence}; the proof is in Appendix~\ref{proof:thm:informedolsrate}.

\begin{theorem}[\textbf{Informed OLS rate}]\label{thm:informedolsrate}
Suppose condition~\eqref{eq:conditioninformedsellersmixed} holds. Then,
\begin{equation*}
\sum_{i \in \mathcal{I}^{in}} \mathbb{E}\,\|\hat{\theta}_{n,i}^{in}-\theta_i^{*,\,in}\|_2^2
\;=\;
O\!\left(n^{\eta_{\max}-1}\log n\right).
\end{equation*}
Moreover, for every $\varepsilon \in (0,\, 1 - \eta_{\max})$,
\begin{equation}\label{eq:thm:informedolsrate-as}
\sum_{i \in \mathcal{I}^{in}}\|\hat\theta_{n,i}^{in} - \theta_i^{*,\,in}\|_2^2 \;=\; O\!\left(n^{\eta_{\max} - 1 + \varepsilon}\right) \quad\text{a.s.}
\end{equation}
In particular, $\hat\theta_{n,i}^{in} \to \theta_i^{*,\,in}$ a.s.\ for every $i \in \mathcal{I}^{in}$, and
\begin{equation}\label{eq:thm:informedolsrate-sum}
\sum_{n \ge 1} \frac{1}{n+1}\sum_{i \in \mathcal{I}^{in}}\|\hat\theta_{n,i}^{in} - \theta_i^{*,\,in}\|_2^2 \;<\;\infty \quad\text{a.s.}
\end{equation}
\end{theorem}

\subsection{Primitive Version of Theorem~\ref{thm:mixedmarketconvergence}}\label{sec:mixedmarket-primitive}

Theorem~\ref{thm:mixedmarketconvergence} is a corollary of a primitive version of the same result, for which we introduce a few constants beyond $L_\phi^{ob}$, $C_x$, and $C_M$ inherited from Theorem~\ref{thm:globalconvergencetocompetitiveoutcome} (all restricted to the set of oblivious sellers $\mathcal{I}^{ob}$). We first define the oblivious-on-oblivious cross-coupling weight
\[
\bar\gamma^{ob} \;\triangleq\; \frac{1}{2}\left(\max_i \sum_{j \in \mathcal{I}^{ob}\setminus\{i\}} \gamma_{i,j} + \max_i \sum_{j \in \mathcal{I}^{ob}\setminus\{i\}}\gamma_{j,i}\right),
\]
the analogue of $\bar\gamma$ from Theorem~\ref{thm:globalconvergencetocompetitiveoutcome} restricted to oblivious--oblivious links. Three additional cross-coupling constants are needed to capture the feedback effects. The first measures the aggregate effect of informed sellers on oblivious sellers' regression error:
\[
\bar\Lambda \;\triangleq\; \max_{i \in \mathcal{I}^{ob}} \sum_{j \in \mathcal{I}^{in}}\frac{\gamma_{i,j}\,\gamma_j}{2|\beta_j|} \;+\; L_\phi^{in,\theta}\,\max_{i \in \mathcal{I}^{ob}}\sum_{j \in \mathcal{I}^{in}}\gamma_{i,j}.
\]
The second represents how a perturbation in the running mean $m_{n,k}$ propagates through informed sellers' best responses and feeds into oblivious sellers' regression:
\[
\bar\Psi \;\triangleq\; \max_{k \in [N]} \sum_{i \in \mathcal{I}^{ob}}\sum_{j \in \mathcal{I}^{in}\setminus\{k\}}\frac{\gamma_{i,j}\,\gamma_{j,k}}{2|\beta_j|}.
\]
The third is the small-gain margin on the informed side: the diagonal self-stabilizing rate---combining the effects of oblivious sellers and running-mean forecasts on informed sellers' regression---minus the informed--informed cross-coupling drag,
\[
\bar\kappa \;\triangleq\; \min\,\left\{2 - L_\phi^{ob},\;\; 2 - L_\phi^{in,\theta} - \max_{j \in \mathcal{I}^{in}}\frac{\gamma_j}{2|\beta_j|}\right\} \;-\; \max_{k \in [N]}\sum_{j \in \mathcal{I}^{in}\setminus\{k\}}\frac{\gamma_{j,k}}{2|\beta_j|}.
\]
The proof of the following result is in Appendix~\ref{proof:thm:mixedmarketconvergence-primitive}.

\begin{theorem}[\textbf{Mixed-market convergence, primitive version}]\label{thm:mixedmarketconvergence-primitive}
Suppose:
\begin{itemize}
\item[(i)] (Exploration regularity, informed side) $\eta_{\max} < 1/2$;
\item[(ii)] (Small-gain, informed side) $\bar\kappa > 0$;
\item[(iii)] (Small-gain, oblivious side) $C_M > C_x\left(\, 2\bar\gamma^{ob}\, L_\phi^{ob} + \bar\Lambda + L_\phi^{ob}\, \bar\Psi/\bar\kappa \,\right)$.
\end{itemize}
Then:
\begin{enumerate}[label=(\roman*)]
\item $\hat\theta_{n,i}^{ob} \to \theta_i^{*,ob}$, $\hat\theta_{n,j}^{in} \to \theta_j^{*,in}$, and $\tilde p_{n,k} \to p_k^{NE}$ a.s.\ for every $i \in \mathcal{I}^{ob}$, $j \in \mathcal{I}^{in}$, and $k \in [N]$;
\item with $c^* \triangleq \sup_{\lambda > 0} \min\left\{1 + \left[C_M - C_x(2\bar\gamma^{ob} L_\phi^{ob} + \bar\Lambda) - \lambda L_\phi^{ob}\right]/C_x^2,\;\; \bar\kappa - C_x \bar\Psi/\lambda\right\} > 0$,\footnote{Because $c^*(\lambda)$ is the pointwise minimum of a strictly decreasing affine function of $\lambda$ and a strictly increasing function of $\lambda$, its supremum is attained at the unique crossing point of the two arguments; setting them equal reduces to a quadratic in $\lambda$, so $c^*$ admits an explicit closed-form expression in terms of the primitives.} the aggregate learning regret satisfies
\[
\mathrm{\mathbf{RegLe}}(T)
=
\begin{cases}
O\!\left(T^{\eta_{\max}}\,\log T\right), & c^* > 1 - \eta_{\max},\\[3pt]
O\!\left(T^{\eta_{\max}}\,(\log T)^2\right), & c^* = 1 - \eta_{\max},\\[3pt]
O\!\left(T^{1-c^*}\right), & c^* < 1 - \eta_{\max},
\end{cases}
\]
while the exploration tax $\mathrm{\mathbf{RegEx}}(T)=\Theta(T)$, so the total regret $\Delta_i(T) = \Theta(T)$ for every $i \in [N]$.
\item the surplus capture $S_i = -\,|\beta_i|\nu_i^2/(\Pi_i^{C}-\Pi_i^{NE}) < 0$ for every $i\in\mathcal{I}^{ob}$ and $S_j = 0$ for every $j\in\mathcal{I}^{in}$ a.s.
\end{enumerate}
\end{theorem}

\subsection{Numerical Revenue Comparison}\label{sec:meta-revenue}

Table~\ref{tab:meta-revenue-summary} reports realized average revenue $\bar R_{T,i}$ and the surplus-capture ratio $S_i$ for each composition in a symmetric duopoly, corroborating the surplus-capture items of Theorems~\ref{thm:globalconvergencetocompetitiveoutcome}, \ref{thm:allinformedmeanforecast}, and~\ref{thm:mixedmarketconvergence}. In the \textsf{ob}--\textsf{ob} composition, both sellers pay the persistent exploration tax, with realized revenues sitting below $\Pi^{NE}$ in proportion to $\nu^2$; in the \textsf{ob}--\textsf{in} composition, the oblivious seller pays the tax while the informed seller captures $\Pi^{NE}$ under the running-mean forecast; and in the \textsf{in}--\textsf{in} composition, both sellers' realized revenues are equal to $\Pi^{NE}$. The last \textsf{ob}--\textsf{in} row is the numerical counterpart of Proposition~\ref{prop:stackelberg} (Appendix~\ref{sec:stackelberg-benchmark}): under the clairvoyant forecast, the informed side's surplus capture climbs from $0$ to $0.31$.

\begin{table}[!htbp]
\centering
\begin{tabular}{llcccc}
\toprule
composition & exploration / forecast & $\bar R_{T,1}$ {\scriptsize[5--95\%]} & $\bar R_{T,2}$ {\scriptsize[5--95\%]} & $S_1$ & $S_2$ \\
\midrule
\textsf{ob}--\textsf{ob} & linear, $\nu^2 = 0.20$ & $2.99\;[2.95, 3.03]$ & $2.99\;[2.96, 3.03]$ & $-0.27$ & $-0.27$ \\
\textsf{ob}--\textsf{ob} & linear, $\nu^2 = 0.05$ & $3.16\;[3.12, 3.23]$ & $3.16\;[3.12, 3.23]$ & $-0.04$ & $-0.04$ \\
\midrule
\textsf{ob}--\textsf{in} & running mean, $\eta = 0.5$  & $3.09\;[3.05, 3.16]$ & $3.19\;[3.16, 3.22]$ & $-0.13$ & $\phantom{-}0.00$ \\
\textsf{ob}--\textsf{in} & running mean, $\eta = 0.25$ & $3.09\;[3.06, 3.13]$ & $3.18\;[3.15, 3.21]$ & $-0.14$ & $-0.01$ \\
\textsf{ob}--\textsf{in} & perfect prediction & $3.14\;[3.09, 3.21]$ & $3.41\;[3.37, 3.45]$ & $-0.07$ & $\phantom{-}0.31$ \\
\midrule
\textsf{in}--\textsf{in} & decaying, $\eta = 0.5$ & $3.19\;[3.15, 3.23]$ & $3.19\;[3.15, 3.23]$ & $\phantom{-}0.00$ & $\phantom{-}0.00$ \\
\textsf{in}--\textsf{in} & decaying, $\eta = 0.7$ & $3.19\;[3.12, 3.25]$ & $3.19\;[3.12, 3.26]$ & $\phantom{-}0.00$ & $\phantom{-}0.00$ \\
\bottomrule
\end{tabular}
\caption{Revenue comparison across the three duopoly compositions of oblivious (\textsf{ob}) and informed (\textsf{in}) sellers. Brackets are cross-seed $5\%$--$95\%$ ranges over $S = 200$ seeds. Symmetric revenue duopoly with $\alpha = 2.5$, $\beta = -1$, $\gamma = 0.6$, $[l, u] = [0.5, 3.5]$; $T = 60{,}000$. The \textsf{ob}--\textsf{ob} cells use constant exploration $\nu^2 \in \{0.20, 0.05\}$. The \textsf{ob}--\textsf{in} cells use the running-mean forecast with $\nu_{n,1}^2 \equiv 0.10$ for the oblivious seller and $\nu_{n,2}^2 = 0.10\,(n+1)^{-\eta}$ for $\eta \in \{0.5, 0.25\}$ for the informed seller; the final \textsf{ob}--\textsf{in} row uses the perfect-prediction forecast for the informed seller at $\eta = 0.5$. The \textsf{in}--\textsf{in} cells use running-mean forecast with $\nu_{n,i}^2 = 0.30\,(n+1)^{-\eta}$ and $\eta \in \{0.5, 0.7\}$.}
\label{tab:meta-revenue-summary}
\end{table}

\subsection{Clairvoyant Forecasts and the Stackelberg Benchmark}\label{sec:stackelberg-benchmark}

The running-mean forecast in Section~\ref{sec:convergence} is limited by its inability to observe a competitor's \emph{next-period} exploration $z_{n+1,j}$. In comparison, an informed seller in possession of sharper competitor-price information---say from market-microstructure data, public algorithmic disclosures, or stronger forecasting models---may be able to turn this predictive advantage into extra revenue. We formalize this intuition through a mixed duopoly, where seller~1 is oblivious and seller~2 is informed and uses the \emph{perfect-prediction} forecast 
$$
\hat p_{n+1,1} = p_{n+1,1}.
$$
Note that the Stackelberg system $p_1^{S} = \argmax_{p_1} R_1(p_1, \mathrm{BR}_2(p_1))$ and $p_2^{S} = \mathrm{BR}_2(p_1^{S})$ has solutions
\begin{equation}\label{eq:stackelberg-prices-app}
p_1^{S} \;=\; \frac{-2\beta_2\alpha_1+\gamma_{1,2}\alpha_2}{4\beta_1\beta_2-2\gamma_{1,2}\gamma_{2,1}}
\quad \text{and} \quad
p_2^{S} \;=\; \frac{\alpha_2+\gamma_{2,1}p_1^{S}}{-2\beta_2}.
\end{equation}
For ease of exposition, we assume that the Stackelberg-induced least-squares limits lie in the interiors of the sellers' parameter sets. The following result shows that the system converges to the Stackelberg outcome; the proof is in Appendix~\ref{proof:prop:stackelberg}. 

\begin{proposition}[Stackelberg lock-in]\label{prop:stackelberg}
Suppose $\eta_1 = 0$ and $\eta_2\in(0,1/2)$. Then, $\tilde p_{n,i} \to p_i^{S}$ for $i\in\{1,2\}$ almost surely. In particular, $S_2>0$.
\end{proposition}

\paragraph{Remark.}
The informed seller is able to capture extra revenue because it asymptotically learns to best respond to competitor prices \emph{exactly}. Because $z_{n+1,1}$ is unobservable in practice, the Stackelberg follower revenue should be read as an \emph{informational upper bracket}: any sharper-but-implementable forecast delivers a realized informed revenue somewhere between Nash and Stackelberg follower, which only strengthens the incentive to be informed in the strategy game of Section~\ref{sec:meta-strategy}. The result extends immediately to $N-1$ oblivious leaders and a single informed follower. However, when there are two or more mutually clairvoyant informed sellers, ``perfect mutual prediction'' becomes circular and ill-posed (e.g., $p_i$ depends on $p_j$ depends on $p_i$ at the same time).

\subsection{Robustness to the Learning Rule}\label{sec:robustness-learning-rule}

Our theoretical analyses center around the estimate--exploit--explore paradigm of Section~\ref{sec:learning-dynamics}: iterated ordinary least squares (OLS) with additive, mean-zero dithering. In this appendix, we numerically show that the qualitative conclusions established in this paper generalize beyond this modeling choice. For the \emph{estimate} step, we replace OLS with \emph{ridge} least squares under a shrinking penalty $\lambda_n = \lambda_0(n+1)^{-1/2}$; a vanishing penalty keeps the estimator consistent while stabilizing learning in small-sample regimes. For the \emph{explore} step, we replace additive dithering with \emph{$\varepsilon$-greedy} exploration, under which a seller posts a uniformly random price in $[l,u]$ with probability $\varepsilon_n$ and its greedy price otherwise. All other primitives match the symmetric duopoly of Figure~\ref{fig:pe-continuum-region-revenue} ($\alpha=2.5$, $\beta=-1$, $\gamma=0.4$, $[l,u]=[0.5,2.5]$).

\begin{figure}[!htbp]
\centering
\begin{minipage}[t]{.49\textwidth}
\centering
{\small\textbf{Ridge + dither}}\\[2pt]
\begin{subfigure}[t]{.49\linewidth}
\centering
\includegraphics[width=\linewidth]{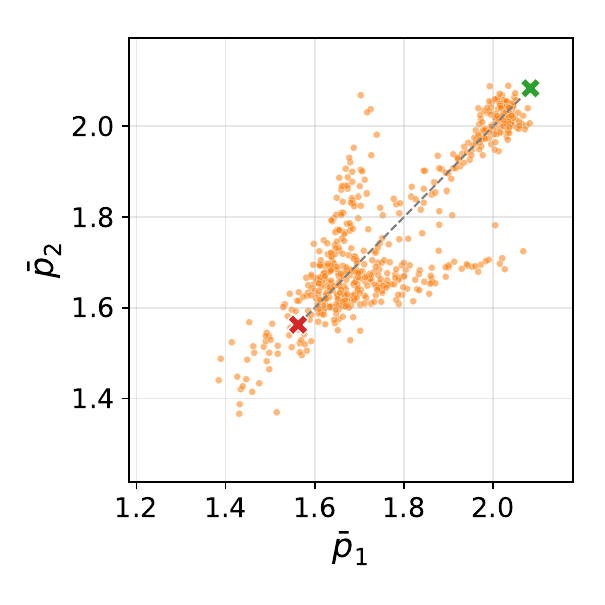}
\caption{Symmetric.}
\label{fig:robustness-ridge-sym}
\end{subfigure}\hfill
\begin{subfigure}[t]{.49\linewidth}
\centering
\includegraphics[width=\linewidth]{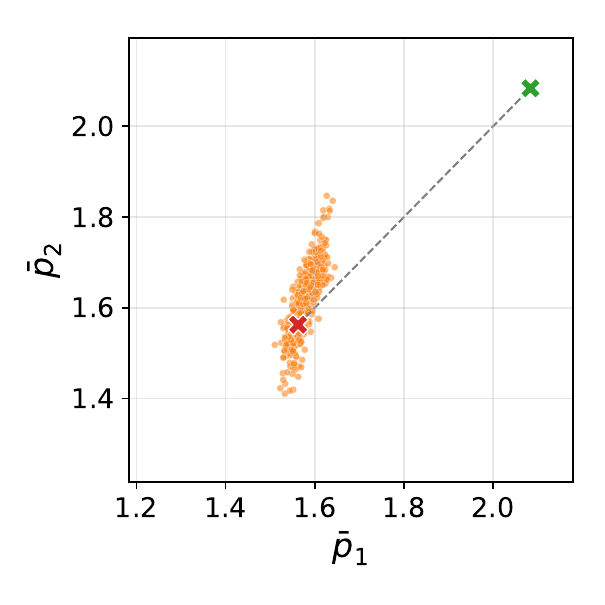}
\caption{Asymmetric.}
\label{fig:robustness-ridge-asym}
\end{subfigure}
\end{minipage}\hfill
\begin{minipage}[t]{.49\textwidth}
\centering
{\small\textbf{$\varepsilon$-greedy}}\\[2pt]
\begin{subfigure}[t]{.49\linewidth}
\centering
\includegraphics[width=\linewidth]{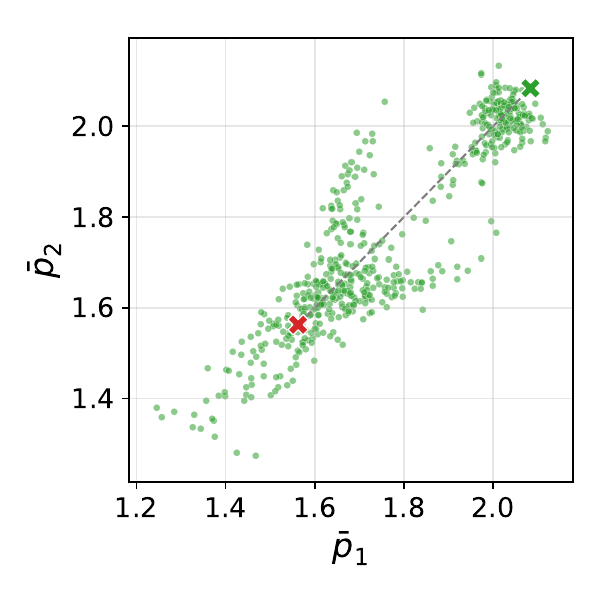}
\caption{Symmetric.}
\label{fig:robustness-eps-sym}
\end{subfigure}\hfill
\begin{subfigure}[t]{.49\linewidth}
\centering
\includegraphics[width=\linewidth]{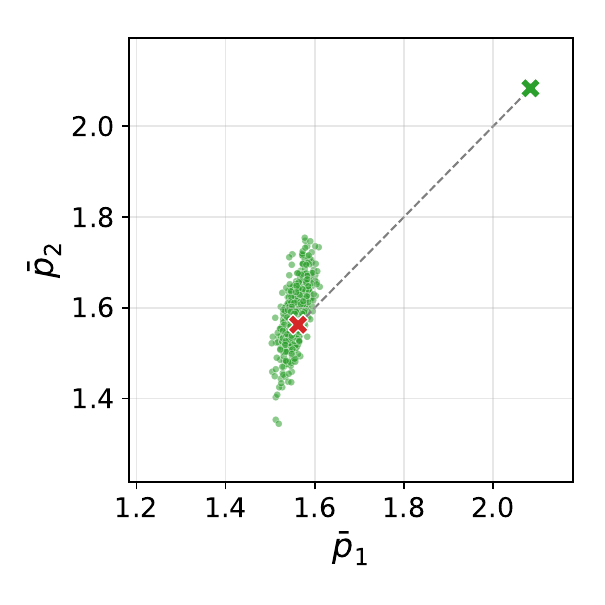}
\caption{Asymmetric.}
\label{fig:robustness-eps-asym}
\end{subfigure}
\end{minipage}
\caption{Long-run average prices $(\bar p_1,\bar p_2)$ of an all-oblivious symmetric duopoly under ridge regression (two left panels) and $\varepsilon$-greedy exploration (two right panels), each with a symmetric and an asymmetric exploration schedule. Red and green markers denote the competitive ($\mathbf p^{NE}$) and collusive ($\mathbf p^{C}$) benchmarks. The ridge panels add dithering with per-seller variance $\nu_{n,i}^2 = c_i(n+1)^{-\eta_i}$---$(c,\eta) = (0.3, 0.85)$ for both sellers in the symmetric case, and seller~1 made variance-dominant via $(c_1,\eta_1) = (0.5, 0.25)$ versus $(c_2,\eta_2) = (0.25, 0.8)$ in the asymmetric case---and estimate with the shrinking penalty $\lambda_n = \lambda_0(n+1)^{-1/2}$, $\lambda_0 = 0.2$. The $\varepsilon$-greedy panels post a uniform random price with probability $\varepsilon_{n,i} = \min\{1, c_i(n+1)^{-\eta_i}\}$---$(c,\eta) = (0.5, 0.85)$ for both sellers in the symmetric case and $(c_1,\eta_1) = (0.8, 0.15)$ versus $(c_2,\eta_2) = (0.4, 0.8)$ in the asymmetric case. Each panel uses horizon $T = 4\times 10^4$. Under symmetric exploration the cloud fills the pseudo-equilibrium continuum; under asymmetric (variance-dominant) exploration it collapses toward $\mathbf p^{NE}$.}
\label{fig:robustness-continuum}
\end{figure}

\begin{figure}[!htbp]
\centering
\includegraphics[width=\textwidth]{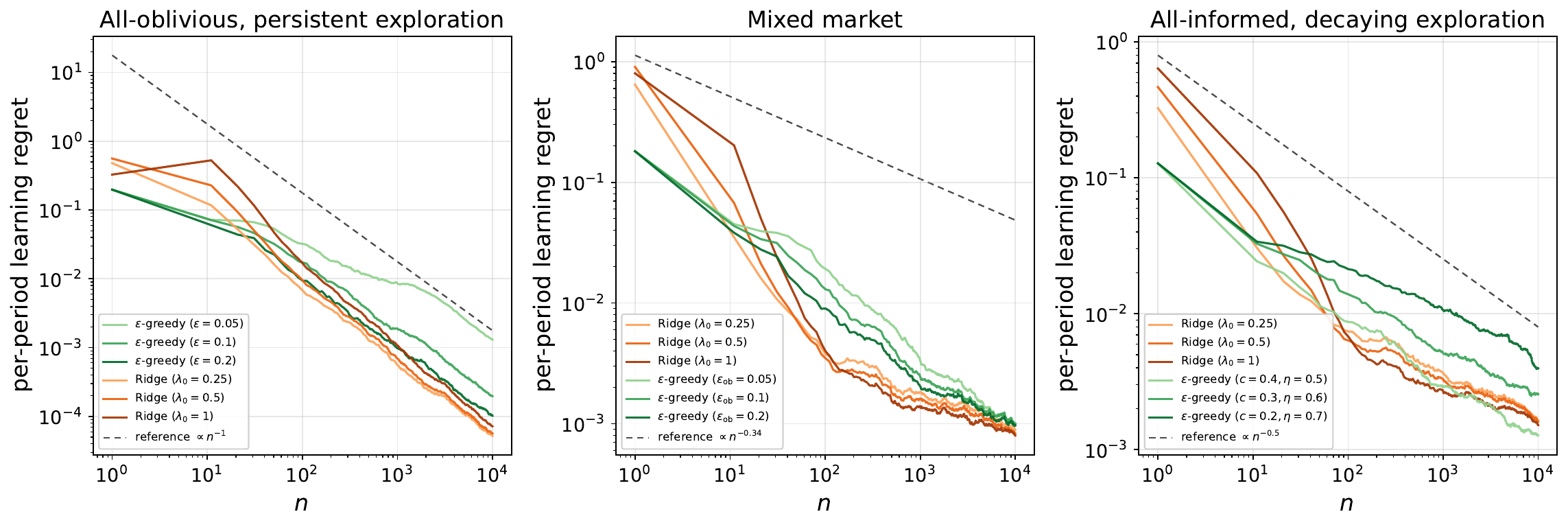}
\caption{Aggregate per-period learning regret on log--log axes; each panel overlays three ridge and three $\varepsilon$-greedy configurations, with ridge using the shrinking penalty $\lambda_n = \lambda_0(n+1)^{-1/2}$. \emph{Left:} all-oblivious market under persistent exploration---constant dither $\nu^2 = 0.2$ (ridge, $\lambda_0 \in \{0.25, 0.5, 1\}$) or constant $\varepsilon \in \{0.05, 0.1, 0.2\}$. \emph{Middle:} mixed duopoly with a persistently exploring oblivious seller (schedule as in the left panel) and an informed seller exploring at the decaying rate $\nu^2 = 0.1(n+1)^{-1/2}$ or $\varepsilon_n = 0.4(n+1)^{-0.6}$. \emph{Right:} all-informed market under decaying exploration---$\nu^2 = 0.1(n+1)^{-1/2}$ (ridge, $\lambda_0 \in \{0.25, 0.5, 1\}$) or $\varepsilon_n = c(n+1)^{-\eta}$ with $(c,\eta) \in \{(0.4, 0.5), (0.3, 0.6), (0.2, 0.7)\}$. Primitives match Figure~\ref{fig:robustness-continuum}; the regret panels use horizon $T = 10^4$ per configuration. Every configuration's per-period learning regret decays, tracking its regime's rate ($n^{-1}$ for all-oblivious, a fitted power for the mixed market, and $n^{-1/2}$ for all-informed).}
\label{fig:robustness-mse}
\end{figure}

\paragraph{Continuum and variance dominance.}
Figure~\ref{fig:robustness-continuum} revisits the pseudo-equilibrium continuum of Figure~\ref{fig:pe-continuum-region-revenue} for both rules. Under \emph{symmetric} sublinear exploration the long-run average prices fill the continuum, whereas under \emph{asymmetric} exploration---one seller variance-dominant in the sense of Theorem~\ref{thm:multi-seller-VD}---the cloud collapses tightly around the competitive point. Both the non-identifiability and its (partial) resolution by variance dominance are therefore properties of the exploration geometry and generalize beyond the OLS estimator and additive dithering.

\paragraph{Convergence to the competitive outcome.}
Figure~\ref{fig:robustness-mse} tracks the aggregate cumulative learning regret $\mathrm{\mathbf{RegLe}}(T)$ in the three convergent regimes: an all-oblivious market under persistent exploration (Theorem~\ref{thm:globalconvergencetocompetitiveoutcome}), a mixed duopoly with one persistently-exploring oblivious seller and one decaying-exploration informed seller (Theorem~\ref{thm:mixedmarketconvergence}), and an all-informed market under decaying exploration (Theorem~\ref{thm:allinformedmeanforecast}). Under both alternative rules the per-period learning regret decays at the rate predicted for each regime ($n^{-1}$, a fitted power, and $n^{-1/2}$, respectively). The exploration tax likewise generalizes with the rule: the per-period tax is of order $\nu_{n,i}^2$ for additive dither and of order $\varepsilon_{n,i}$ for $\varepsilon$-greedy (which posts a random price with probability $\varepsilon_{n,i}$), so persistent exploration leaves a linear tax and decaying exploration a sublinear one, exactly as in the OLS analysis. 

\section{Proofs}

\subsection{Proof of Theorem~\ref{thm:impactdemandnoise}}\label{proof:thm:impactdemandnoise}
Let $(\mathcal{F}_n)_{n \ge 2}$ be the filtration generated by $\left\{z_{m+1,i}, \varepsilon_{m,i}: 1 \le m \le n, i \in [N]\right\}$, i.e., $p_{m,i} \in \mathcal{F}_{m-1}$ and $\varepsilon_{m,i} \in \mathcal{F}_{m}$ for all $m \ge 2$. We drop the subscript $i$ for notational simplicity. Denote the numerator of $w_{n}$ by
\[
\begin{aligned}
C_n &= \sum_{m=1}^{n} \left(p_{m} - \bar{p}_{n}\right) \varepsilon_{m} = \sum_{m=2}^{n} \left(1 - \frac{1}{m}\right) (p_m - \bar{p}_{m-1}) (\varepsilon_m - \bar{\varepsilon}_{m-1}),
\end{aligned}
\]
where $\bar{\varepsilon}_{m-1} = \tfrac{1}{m-1} \sum_{k=1}^{m-1} \varepsilon_{k}$. Denoting $a_m = \left(1 - \frac{1}{m}\right) (p_{m} - \bar{p}_{m-1}) \in \mathcal{F}_{m-1}$, we can write
\[
J_n = \sum_{m=2}^{n} \frac{m}{m-1} a_m^2, \qquad C_n = D_n - R_n,
\]
where $D_n = \sum_{m=2}^{n} a_m \varepsilon_{m}$ and $R_n = \sum_{m=2}^{n} a_m \bar{\varepsilon}_{m-1}$. Since $\tfrac{m}{m-1} \in [1,2]$ for all $m \ge 2$, $\sum_{m=2}^{n} a_m^2 \le J_n \le 2 \sum_{m=2}^{n} a_m^2$. Note that $(D_n, \mathcal{F}_n)_{n \ge 2}$ is a square-integrable martingale with uniformly bounded increments, and
\[
V_n \triangleq \sum_{m=2}^{n} \mathbb{E}[(D_m - D_{m-1})^2 | \mathcal{F}_{m-1}] = \sigma^2 \sum_{m=2}^{n} a_m^2.
\]
By the strong law of large numbers for martingales (e.g., \citet{Williams_1991}, pp.122-124): if $V_n \rightarrow \infty$, then $D_n / V_n \rightarrow 0$ a.s.; if $V_n \rightarrow V_\infty$ for some $V_\infty < \infty$, then $D_n \rightarrow D_\infty$ a.s. for some finite $D_\infty$. Since $J_n$ is squeezed by $\sum_{m=2}^{n} a_m^2$ up to constant factors, the same conclusions can be stated in terms of $J_n$. This handles $D_n$. For $R_n$, similarly it suffices to analyze $\tfrac{\sum_{m=2}^{n} a_m \bar{\varepsilon}_{m-1}}{\sum_{m=2}^{n} a_m^2}$. By the law of the iterated logarithm for martingales (e.g., \citet{freedman1975ontail}), $\bar{\varepsilon}_{m} = O\left(\sqrt{\tfrac{\log \log m}{m}}\right)$ a.s. By Cauchy-Schwarz,
\[
\left|\frac{\sum_{m=2}^{n} a_m \bar{\varepsilon}_{m-1}}{\sum_{m=2}^{n} a_m^2}\right| \le \sqrt{\frac{\sum_{m=2}^{n} \bar{\varepsilon}_{m-1}^2}{\sum_{m=2}^{n} a_m^2}} \preceq \sqrt{\frac{\log n \log \log n}{\sum_{m=2}^{n} a_m^2}} \quad\text{a.s.}
\] If $\tfrac{J_n}{\log n \log \log n} \rightarrow \infty$, then $\tfrac{\sum_{m=2}^{n} a_m \bar{\varepsilon}_{m-1}}{\sum_{m=2}^{n} a_m^2} \rightarrow 0$ a.s. If $\sum_{m=2}^{\infty} a_m^2 (\log m)^{1+\delta} < \infty$, similarly
\[
\begin{aligned}
&\sum_{m=2}^{n} |a_m \bar{\varepsilon}_{m-1}| \le \sqrt{\sum_{m=2}^{n} a_m^2 (\log m)^{1+\delta}}\cdot \sqrt{\sum_{m=2}^{n} \frac{\bar{\varepsilon}_{m-1}^2}{(\log m)^{1+\delta}}} < \infty \quad\text{a.s.}
\end{aligned}
\] Since $\sum_{m=2}^{n} |a_m \bar{\varepsilon}_{m-1}|$ is increasing and bounded, it converges and $\sum_{m=2}^{\infty} a_m \bar{\varepsilon}_{m-1}$ converges absolutely. $\sum_{m=2}^{n} a_m^2$ is likewise increasing and bounded, so it converges. Therefore, $\tfrac{\sum_{m=2}^{n} a_m \bar{\varepsilon}_{m-1}}{\sum_{m=2}^{n} a_m^2}$ converges a.s. to a finite limit. Together, if $\tfrac{J_n}{\log n \log \log n} \rightarrow \infty$, then
\[
w_n = C_n / J_n = D_n / J_n - R_n / J_n \rightarrow 0 \quad\text{a.s.}
\]
If $\sum_{m=2}^{n} \left(p_{m} - \bar{p}_{m-1}\right)^2 (\log m)^{1+\delta} < \infty$ for some $\delta > 0$, then
\[
w_n = C_n / J_n = D_n / J_n - R_n / J_n \rightarrow w_{\infty} \quad\text{a.s. for some finite } w_{\infty}.
\]

\subsection{Proof of Theorem~\ref{thm:cesaro-main}}\label{proof:thm:cesaro-main}

Let $\{x_m\}_{m\ge 1}$ be a bounded sequence in $[l,u]\subset\mathbb{R}$. For $1\le T<S$, define $y_T \triangleq \tfrac{1}{T}\sum_{m=1}^T x_m$, $\tilde y_{T,S} \triangleq \tfrac{1}{S-T}\sum_{m=T+1}^S x_m$, $V_T \triangleq \tfrac{1}{T}\sum_{m=1}^T(x_m-y_T)^2$, and $\tilde V_{T,S}\triangleq \tfrac{1}{S-T}\sum_{m=T+1}^S(x_m-\tilde y_{T,S})^2$.

\begin{lemma}[Mixture-variance]\label{lem:mixture}
For any $1\le T<S$ and any bounded sequence $\{x_m\}$,
\begin{equation}\label{eq:mixture-variance-y}
V_S \;=\; \frac{T}{S}\,V_T \;+\; \frac{S-T}{S}\,\tilde V_{T,S} \;+\; \frac{T}{S-T}\,(y_S-y_T)^2.
\end{equation}
In particular,
\begin{equation}\label{eq:key-sqrt}
|y_S-y_T|\;\le\;\sqrt{\frac{S-T}{T}\,V_S}.
\end{equation}
\end{lemma}
\begin{proof}[Proof of Lemma~\ref{lem:mixture}]
Splitting the sum for $y_S$ at index $T$ gives
\[
y_S=\frac{T}{S}y_T+\frac{S-T}{S}\tilde y_{T,S},
\]
from which $y_T-y_S=-\tfrac{S-T}{S}(\tilde y_{T,S}-y_T)$ and $\tilde y_{T,S}-y_S=\tfrac{T}{S}(\tilde y_{T,S}-y_T)$. Expanding the terms yields
\[
\sum_{m=1}^T(x_m-y_S)^2 = TV_T + T(y_T-y_S)^2, \qquad \sum_{m=T+1}^S(x_m-y_S)^2 = (S-T)\tilde V_{T,S} + (S-T)(\tilde y_{T,S}-y_S)^2.
\]
Substituting $y_T-y_S=-\tfrac{S-T}{S}(\tilde y_{T,S}-y_T)$ and $\tilde y_{T,S}-y_S=\tfrac{T}{S}(\tilde y_{T,S}-y_T)$, we have
\[
T(y_T-y_S)^2+(S-T)(\tilde y_{T,S}-y_S)^2=\frac{TS}{S-T}(y_S-y_T)^2.
\]
Together,
\[
SV_S= \sum_{m=1}^T(x_m-y_S)^2 + \sum_{m=T+1}^S(x_m-y_S)^2 = TV_T+(S-T)\tilde V_{T,S} + \frac{TS}{S-T}(y_S-y_T)^2,
\]
which is~\eqref{eq:mixture-variance-y} after dividing by $S$. The inequality follows from dropping the two variance terms and taking square roots.
\end{proof}

\begin{lemma}[Ces\`aro convergence]\label{lem:cesaro-convergence}
Let $\{x_m\}$ be a sequence in $[l,u]$. Suppose there exist $\rho>1$ and $T_0\ge 1$ such that, along $T_k\triangleq\lceil T_0\rho^k\rceil$,
\begin{equation}\label{eq:geometric-grid-summability}
\sum_{k\ge 0}\sqrt{V_{T_k}}\;<\;\infty.
\end{equation}
Then $y_T\to\bar y$ as $T\to\infty$ for some $\bar y\in[l,u]$. In particular, \eqref{eq:geometric-grid-summability} is satisfied if $V_T\le C(\log T)^{-(2+\epsilon)}$ for all large $T$ for some constants $C,\epsilon>0$.
\end{lemma}
\begin{proof}[Proof of Lemma~\ref{lem:cesaro-convergence}]
Let $c_\rho\triangleq\rho+1/T_0$. Since $T_{k+1}/T_k\le c_\rho$, we have $(T_{k+1}-T_k)/T_k\le c_\rho-1$. Applying~\eqref{eq:key-sqrt} with $(T,S)=(T_k,T_{k+1})$,
\[
|y_{T_{k+1}}-y_{T_k}|\;\le\;\sqrt{c_\rho-1}\,\sqrt{V_{T_{k+1}}}.
\]
Summing and using the hypothesis, $\{y_{T_k}\}$ has finite total variation, hence Cauchy, hence converges to some $\bar y\in[l,u]$. For arbitrary $S\in[T_k,T_{k+1})$, applying~\eqref{eq:key-sqrt} with $(T,S)=(T_k,S)$ gives $|y_S-y_{T_k}|\le\sqrt{(c_\rho-1)V_S}$. Applying identity~\eqref{eq:mixture-variance-y} to $\{1,\ldots,T_{k+1}\}$ split at index $S$ gives $V_{T_{k+1}}\ge (S/T_{k+1})V_S$, i.e., $V_S\le c_\rho V_{T_{k+1}}$. Combining, $|y_S-y_{T_k}|\le\sqrt{(c_\rho-1)c_\rho V_{T_{k+1}}}\to 0$ as $k\to\infty$; combined with $y_{T_k}\to\bar y$, $y_S\to\bar y$. The sufficient condition follows by noting $\log T_k\ge k\log\rho$ and $\sqrt{V_{T_k}}\le\sqrt C\,(k\log\rho)^{-(1+\epsilon/2)}$, the sum of which converges.
\end{proof}

Define, for each seller $i$ and horizon $n$,
\[
\begin{aligned}
D_{n,i}&\triangleq\sum_{m\le n}(\tilde p_{m,i}-\bar{\tilde p}_{n,i})^2,\\
C_{n,i}&\triangleq2\sum_{m\le n}(\tilde p_{m,i}-\bar{\tilde p}_{n,i})(z_{m,i}-\bar z_{n,i}),\\
E_{n,i}&\triangleq\sum_{m\le n}(z_{m,i}-\bar z_{n,i})^2,
\end{aligned}
\]
so that 
\begin{equation}\label{eq:J-decomp-proof}
J_{n,i}\;=\;D_{n,i}\;+\;C_{n,i}\;+\;E_{n,i}.
\end{equation}
To apply Lemma~\ref{lem:cesaro-convergence} to $x_m=\tilde p_{m,i}$ along any geometric grid $T_k\triangleq\lceil T_0\rho^k\rceil$, it suffices to verify $\sum_k\sqrt{V_i(T_k)}<\infty$ a.s. for $V_i(n)\triangleq D_{n,i}/n$. From \eqref{eq:J-decomp-proof} and $E_{n,i}\ge 0$ we have
\[
\sqrt{V_i(n)}\le\sqrt{J_{n,i}/n}+\sqrt{|C_{n,i}|/n}.
\]
Summing over the grid, the first term is finite by Lemma~\ref{lem:cesaro-convergence}, and the second term telescopes to
\[
C_{n,i}=2\sum_{m\le n}\tilde p_{m,i}z_{m,i}-2n\bar{\tilde p}_{n,i}\bar z_{n,i}.
\]
Both $\sum_{m\le n}\tilde p_{m,i}z_{m,i}$ and $\sum_{m\le n}z_{m,i}$ are martingales with bounded increments, so by the law of the iterated logarithm for martingales (e.g., \citet{freedman1975ontail}) both terms and thus also $|C_{n,i}|$ are $O(\sqrt{n\log\log n})$ a.s. Therefore, along the grid,
\[
\sqrt{|C_{T_k,i}|/T_k}=O\left((\log\log T_k/T_k)^{1/4}\right)=O\left((k/\rho^k)^{1/4}\right),
\]
which is summable. The result follows by Lemma~\ref{lem:cesaro-convergence}.

\subsection{Proof of Theorem~\ref{thm:multi-seller-VD}}\label{proof:thm:multi-seller-VD}

\begin{lemma}[Variance-dominance]\label{lem:vd-crosscov}
For any sellers $i\ne j$ and every $n\ge 2$, $|r_{n,i\leftarrow j}|\le\sqrt{J_{n,j}/J_{n,i}}$. In particular, on the event $J_{n,j}/J_{n,i}\to0$ we have $r_{n,i\leftarrow j}\to0$ a.s.
\end{lemma}

\begin{proof}[Proof of Lemma~\ref{lem:vd-crosscov}]
Cauchy--Schwarz gives $|r_{n,i\leftarrow j}|\le\sqrt{J_{n,i}J_{n,j}}/J_{n,i}$ and the result follows.
\end{proof}

\begin{lemma}[Exploration bookkeeping]\label{lem:expl-bookkeeping}
Suppose all sellers satisfy (Max-Var). Then, $\tfrac1n\sum_{m\le n}\xi_{m}^2\to0$ a.s. for any $\xi_{n}=\sum_{j \in [N]} a_j z_{n,j}$.
\end{lemma}

\begin{proof}[Proof of Lemma~\ref{lem:expl-bookkeeping}]
Recall $\bar z_{n,i} \to 0$ a.s. By \eqref{eq:J-decomp-proof} we have $E_{n,i}\le J_{n,i}+|C_{n,i}|$, and the proof of Theorem~\ref{thm:cesaro-main} and (Max-Var) give $E_{n,i}/n\to0$ a.s., so $\sum_{m\le n}z_{m,i}^2=E_{n,i}+n\bar z_{n,i}^2 \to0$ a.s. For any $\xi_n=\sum_{j \in [N]} a_j z_{n,j}$, Cauchy--Schwarz gives $\xi_m^2\le(\sum_j a_j^2)\sum_j z_{m,j}^2$, so averaging over $m\le n$ and interchanging the finite sums yields
\[
\frac1n\sum_{m\le n}\xi_m^2\le\left(\sum_j a_j^2\right)\sum_j\frac1n\sum_{m\le n}z_{m,j}^2\to0 \quad\text{a.s.}
\]
\end{proof}

Without loss of generality let $j^\star = 1$. For each seller $i$, $\bar{p}_{n,i}\to\bar q_i$ for some $\bar q_i\in[l,u]$ by Theorem~\ref{thm:cesaro-main} (so does $\bar{\tilde p}_{n,i}$ since $\bar{z}_{n,i} \to 0$ a.s.). For each dominated seller $j\ne 1$, (Var-Dom) gives $J_{n,j}/J_{n,1}\to 0$, so Lemma~\ref{lem:vd-crosscov} gives $r_{n,1\leftarrow j}\to 0$ a.s. By~\eqref{eq:obliviousgreedynextprice} and the interior-parameter assumption, $\hat\alpha_{n,1}\to\alpha_1+\sum_{j\ne 1}\gamma_{1,j}\bar q_j$, $\hat\beta_{n,1}\to\beta_1+\sum_{j\ne 1}\gamma_{1,j}r_{n,1\leftarrow j}\to\beta_1$, and $\tilde p_{n,1}\to\mathrm{BR}_1(\bar q_2,\ldots,\bar q_N)$, all a.s. By the proof of Proposition~\ref{prop:dynamicbenchmark}, the \emph{realized} regret satisfies
\begin{equation}\label{eq:br-form}
\widehat\Delta_i(T)\;\triangleq\;\sum_{n\le T}\left[R_i(\mathrm{BR}_i(\mathbf p_{n,-i}),\mathbf p_{n,-i})-R_i(p_{n,i},\mathbf p_{n,-i})\right]=-\beta_i\sum_{n\le T}g_{n,i}^2
\end{equation}
with $\Delta_i(T)=\mathbb{E}\,\widehat\Delta_i(T)$ and $g_{n,i}\triangleq\mathrm{BR}_i(\mathbf p_{n,-i})-p_{n,i}$. Since
\begin{equation}\label{eq:g-decomp}
g_{n,i}\;=\;\underbrace{\left(\mathrm{BR}_i(\tilde{\mathbf p}_{n,-i})-\tilde p_{n,i}\right)}_{\triangleq\,h_{n,i}}\;+\;\underbrace{\left(\sum_{j\ne i}\frac{\gamma_{i,j}}{-2\beta_i}z_{n,j}-z_{n,i}\right)}_{\triangleq\,\xi_{n,i}} .
\end{equation}
By \eqref{eq:br-form}--\eqref{eq:g-decomp},
\begin{equation}\label{eq:regret-dominant-decomp}
\widehat\Delta_1(T)\;\le\;2(-\beta_1)\left(\sum_{n\le T}h_{n,1}^2+\sum_{n\le T}\xi_{n,1}^2\right).
\end{equation}
Note that $\xi_{n,i}=\sum_j a_{i,j}z_{n,j}$ has coefficients $a_{i,i}=-1$ and $a_{i,j}=\gamma_{i,j}/(-2\beta_i)$ for $j\ne i$, so the second sum is $o(T)$ by Lemma~\ref{lem:expl-bookkeeping}. Writing $q_1^\star\triangleq\mathrm{BR}_1(\bar q_{-1})$ and recalling $\tilde p_{n,1}\to q_1^\star$ yields
\[
h_{n,1}=\mathrm{BR}_1(\tilde{\mathbf p}_{n,-1})-\tilde p_{n,1}=\sum_{j\ne1}\frac{\gamma_{1,j}}{-2\beta_1}(\tilde p_{n,j}-\bar q_j)+(q_1^\star-\tilde p_{n,1}).
\]
Viewing this as a linear combination $h_{n,1}=\sum_{j=1}^N b_j\,y_{n,j}$ with $b_j=\tfrac{\gamma_{1,j}}{-2\beta_1}$ and $y_{n,j}=\tilde p_{n,j}-\bar q_j$ for $j\ne1$ and $b_1=1$ and $y_{n,1}=q_1^\star-\tilde p_{n,1}$, Cauchy--Schwarz gives
\[
\frac1T\sum_{n\le T}h_{n,1}^2 \lesssim\sum_{j\ne1}\frac1T\sum_{n\le T}(\tilde p_{n,j}-\bar q_j)^2+\frac1T\sum_{n\le T}(q_1^\star-\tilde p_{n,1})^2.
\]
The second term goes to 0 since $\tilde p_{n,1}\to q_1^\star$. For the first term, note that
\[
\sum_{m\le n}(\tilde p_{m,j}-\bar q_j)^2=\sum_{m\le n}(\tilde p_{m,j}-\bar{\tilde p}_{n,j})^2+n(\bar{\tilde p}_{n,j}-\bar q_j)^2,
\]
where the first sum is $D_{n,j}$. Since $D_{n,j}/n\le(J_{n,j}+|C_{n,j}|)/n\to0$ by (Max-Var) and $\bar{\tilde p}_{n,j}\to\bar q_j$, $\tfrac1n\sum_{m\le n}(\tilde p_{m,j}-\bar q_j)^2\to0$. Therefore, $\sum_{n\le T}h_{n,1}^2=o(T)$, and by~\eqref{eq:regret-dominant-decomp} $\widehat\Delta_1(T)=o(T)$ a.s. Bounded convergence then gives $\Delta_1(T)=o(T)$. 

Fix a dominated seller $i$ and set $\Delta q_i\triangleq\mathrm{BR}_i(\bar q_{-i})-\bar q_i$. If $\Delta q_i=0$, then $i$ behaves like the dominant seller and the argument above applied verbatim gives $\Delta_i(T)=o(T)$. Suppose $\Delta q_i\ne0$. Since $\mathrm{BR}_i(\tilde{\mathbf p}_{n,-i})=\mathrm{BR}_i(\bar q_{-i})+\sum_{j\ne i}\tfrac{\gamma_{i,j}}{-2\beta_i}(\tilde p_{n,j}-\bar q_j)$ and $h_{n,i}=\Delta q_i+\sum_{j\ne i}\tfrac{\gamma_{i,j}}{-2\beta_i}(\tilde p_{n,j}-\bar q_j)-(\tilde p_{n,i}-\bar q_i)$, where $\delta_{n,i}$ denotes the last two terms together, \eqref{eq:g-decomp} gives $g_{n,i}=\Delta q_i+\zeta_{n,i}$, where $\zeta_{n,i}=\delta_{n,i}+\xi_{n,i}$. Both $\delta_{n,i}$ and $\xi_{n,i}$ have vanishing Ces\`aro means, so $\bar\zeta_{T,i}\triangleq\tfrac1T\sum_{n\le T}\zeta_{n,i}\to0$. Expanding the average squared gap then gives
\[
\frac1T\sum_{n\le T}g_{n,i}^2\ge\Delta q_i^2+2\Delta q_i\bar\zeta_{T,i}\to\Delta q_i^2.
\]
So $\liminf_T\tfrac1T\sum_{n\le T}g_{n,i}^2\ge\Delta q_i^2>0$, i.e., $\liminf_T\widehat\Delta_i(T)/T>0$ a.s. Since prices are bounded, the realized regret is $\widehat\Delta_i(T)=\Theta(T)$ a.s., and Fatou gives $\liminf_T\Delta_i(T)/T>0$ and thus $\Delta_i(T)=\Theta(T)$.

Suppose seller~$i$'s projection is asymptotically slack---implying its aggregate contamination $s_{n,i}=\sum_{j\ne i}\gamma_{i,j}r_{n,i\leftarrow j}$ is eventually bounded---\eqref{eq:obliviousgreedynextprice} with $\tau_{n,i}\triangleq s_{n,i}+w_{n,i}$ yields
\[
\frac{\tau_{n,i}(2\tilde p_{n+1,i}-\bar p_{n,i})}{-2\beta_i}=\tilde p_{n+1,i}-\mathrm{BR}_i(\bar{\mathbf p}_{n,-i})-\frac{\bar\varepsilon_{n,i}}{-2\beta_i}.
\]
Averaging over $n<T$, the right-hand side converges to $-\Delta q_i$, and the left becomes
\[
\frac{1}{-2\beta_i} \frac1T\sum_{n<T}\tau_{n,i}\left(\bar q_i+\psi_{n,i}\right),
\]
where $\psi_{n,i}\triangleq 2(\tilde p_{n+1,i}-\bar q_i)-(\bar p_{n,i}-\bar q_i)$. Cauchy--Schwarz gives
\[
\left|\frac1T\textstyle\sum_{n<T}\tau_{n,i}\psi_{n,i}\right|\le\left(\frac1T\textstyle\sum_{n<T} \tau_{n,i}^2\right)^{1/2}\left(\frac1T\textstyle\sum_{n<T} \psi_{n,i}^2\right)^{1/2},
\]
in which the first factor is asymptotically bounded. For the second, we have
\[
\frac1T\textstyle\sum_n \psi_{n,i}^2 \lesssim\frac1T\textstyle\sum_n(p_{n,i}-\bar q_i)^2+\frac1T\textstyle\sum_n z_{n,i}^2+\frac1T\textstyle\sum_n(\bar p_{n,i}-\bar q_i)^2,
\]
where $\frac1T\sum_n(p_{n,i}-\bar q_i)^2=J_{T,i}/T+(\bar p_{T,i}-\bar q_i)^2\to0$, $\frac1T\sum_n z_{n,i}^2\to0$, and $\frac1T\sum_n(\bar p_{n,i}-\bar q_i)^2\to0$. Therefore, $\frac1T\sum_n \tau_{n,i}\psi_{n,i}$ vanishes, and we have $\bar q_i\cdot\tfrac1T\sum_n \tau_{n,i}\to 2\beta_i\Delta q_i$. Since $w_{n,i}\to0$ by Theorem~\ref{thm:impactdemandnoise}, we can replace $\tfrac1T\sum_n\tau_{n,i}$ by $\tfrac1T\sum_n s_{n,i}$, whose limit $\bar s_i$ thus exists. Hence, $\Delta q_i=\tfrac{\bar q_i}{2\beta_i}\,\bar s_i$ and, as $\bar q_i>0$, the degenerate case $\Delta q_i=0$ is equivalent to $\bar s_i=0$. Combining with the two regimes above---$\Delta_i(T)=o(T)$ when $\Delta q_i=0$ and $\Delta_i(T)=\Theta(T)$ when $\Delta q_i\ne0$---completes the proof.

\subsection{Proof of Theorem~\ref{thm:globalconvergencetocompetitiveoutcome}}\label{proof:thm:globalconvergencetocompetitiveoutcome}
For notational simplicity, we drop the superscript $ob$. Recall that each oblivious seller $i$ models the demand as $d_{n,i}=x_{n,i}^\top \theta_i^* + \eta_{n,i}$, where $\eta_{n,i} = \varepsilon_{n,i} + \sum_{j\ne i}\gamma_{i,j}\left(p_{n,j}-p_j^{NE}\right)$. For each seller $i$, define the adaptive covariance matrices $S_{n,i}=\sum_{m=1}^n x_{m,i}x_{m,i}^\top$, and note that $S_{n,i}$ is invertible for all $n \ge 2$ since the first two prices are different. The unprojected least squares estimator writes $\tilde \theta_{n,i} = \theta_i^* + S_{n,i}^{-1} \sum_{m=1}^n x_{m,i} \eta_{m,i}$. Define the estimation error $\tilde{e}_{n,i} = \tilde \theta_{n,i} - \theta_i^*$ and $e_{n,i} = \hat \theta_{n,i} - \theta_i^*$. Then, we have
\begin{align*}
\tilde{\theta}_{n+1, i} - \tilde{\theta}_{n,i} &= S_{n+1,i}^{-1} x_{n+1,i} (d_{n+1,i} - x_{n+1,i}^\top \tilde{\theta}_{n,i}) \\
\tilde{e}_{n+1,i} - \tilde{e}_{n,i} &= - S_{n+1,i}^{-1} x_{n+1,i} x_{n+1,i}^\top \tilde{e}_{n,i} + S_{n+1,i}^{-1} x_{n+1,i} \eta_{n+1,i}
\end{align*}
Define Lyapunov functions $V_{n,i} = \tilde{e}_{n,i}^\top S_{n,i} \tilde{e}_{n,i}$, $W_{n,i} = \frac{1}{n} V_{n,i}$, $V_n = \sum_{i=1}^{N} V_{n,i}$ and $W_n = \sum_{i=1}^{N} W_{n,i}$. Then,
\[
\begin{aligned}
V_{n+1, i} &= V_{n,i}- (x_{n+1, i}^\top \tilde e_{n,i})^2 (1 - x_{n+1, i}^\top S_{n+1,i}^{-1} x_{n+1, i})\\
&\quad + 2 \eta_{n+1,i} x_{n+1,i}^\top \tilde e_{n,i} (1 - x_{n+1, i}^\top S_{n+1,i}^{-1} x_{n+1, i})+ x_{n+1,i}^\top S_{n+1,i}^{-1} x_{n+1,i} \eta_{n+1,i}^2.
\end{aligned}
\]
The Sherman-Morrison formula applied to $S_{n,i}^{-1}$ gives $1 - x_{n+1, i}^\top S_{n+1,i}^{-1} x_{n+1, i} = \tfrac{1}{1 + x_{n+1,i}^\top S_{n,i}^{-1} x_{n+1,i}}$. Denoting $h_{n+1,i} = x_{n+1,i}^\top S_{n,i}^{-1} x_{n+1,i}$, we have
\[
\begin{aligned}
&W_{n+1, i} = \frac{n}{n+1} W_{n,i} - \frac{(x_{n+1, i}^\top \tilde e_{n,i})^2}{(n+1)(1 + h_{n+1,i})}\\
&\quad + \frac{2 \eta_{n+1,i} x_{n+1,i}^\top \tilde e_{n,i}}{(n+1)(1 + h_{n+1,i})} + \frac{h_{n+1,i}}{(n+1)(1 + h_{n+1,i})} \eta_{n+1,i}^2.
\end{aligned}
\]
Since $p_{n+1,j}=\phi^{ob}(\theta_j^*+e_{n,j})+z_{n+1,j}$ and $\mathbb E[z_{n+1,j}| \mathcal F_n]=0$, we have
\[
\left|\mathbb E[p_{n+1,j}-p_j^{NE}\mid\mathcal F_n]\right| = \left|\phi^{ob}(\theta_j^*+e_{n,j})-\phi^{ob}(\theta_j^*)\right| \le L_\phi^{ob} \norm{e_{n,j}}_2 \le L_\phi^{ob} \norm{\tilde e_{n,j}}_2,
\]
where the last inequality holds because projection is non-expansive. Since
\[
\mathbb E[\eta_{n+1,i}x_{n+1,i}^\top \tilde e_{n,i}\mid\mathcal F_n] =\sum_{j\ne i}\gamma_{i,j}\mathbb E[(p_{n+1,j}-p_j^{NE})x_{n+1,i}^\top \tilde e_{n,i}\mid\mathcal F_n]
\]
and $\mathbb E[(p_{n+1,j}-p_j^{NE})x_{n+1,i}^\top \tilde e_{n,i}\mid\mathcal F_n] =\mathbb E[p_{n+1,j}-p_j^{NE}\mid\mathcal F_n]\cdot \mathbb E[x_{n+1,i}^\top \tilde e_{n,i}\mid\mathcal F_n]$ because $(p_{n+1,j}-p_j^{NE})\perp x_{n+1,i}$ for $j \neq i$ conditional on $\mathcal F_n$, we have
\[
\left|\mathbb E[\eta_{n+1,i}x_{n+1,i}^\top \tilde e_{n,i}\mid\mathcal F_n]\right| \le \frac{1}{2} L_\phi^{ob} C_x \sum_{j\ne i}\gamma_{i,j} \left(\norm{\tilde e_{n,i}}_2^2 + \norm{\tilde e_{n,j}}_2^2\right)
\]
by Cauchy-Schwarz and Young's inequality. By~\eqref{eq:globalconvergencespectrallowerboundcondition} we have $\mathbb{E}\left[ (x_{n+1, i}^\top \tilde e_{n,i})^2 \middle| \mathcal{F}_n\right] \ge C_{M} \norm{\tilde e_{n,i}}_2^2$. By Lemma~\ref{lem:spectralboundadaptivecovariancematrix}, for any $0 < \delta < C_M$, eventually a.s.,
\[
h_{n+1,i} \le C_x^2 \lambda_{\max}(S_{n,i}^{-1})= \frac{C_x^2}{C_{M} - \delta} \frac{1}{n} \le \delta_h
\]
for arbitrarily small $\delta_h > 0$. Denoting $C_\eta^2 \triangleq \sup_{n,i}\mathbb E[\eta_{n+1,i}^2\mid\mathcal F_n]<\infty$ and taking conditional expectation of the $W_{n+1}$ recursion gives that, for all large $n$ almost surely,
\[
\begin{aligned}
\mathbb E[W_{n+1}\mid\mathcal F_n] &\le \frac{n}{n+1} W_n- \frac{C_M }{(n+1)(1+\delta_h)} \sum_{i=1}^N \norm{\tilde e_{n,i}}_2^2\\
&\quad + \frac{L_\phi^{ob} C_x }{n+1} \sum_{i=1}^N \sum_{j\ne i}\gamma_{i,j} (\norm{\tilde e_{n,i}}_2^2 + \norm{\tilde e_{n,j}}_2^2) + \frac{N C_x^2 C_\eta^2}{(C_M-\delta)}\frac{1}{n(n+1)}.
\end{aligned}
\]
Since
\[
\sum_{i=1}^N \sum_{j\ne i}\gamma_{i,j} \left(\norm{\tilde e_{n,i}}_2^2 + \norm{\tilde e_{n,j}}_2^2\right)= \sum_{i=1}^N (\gamma_i + \gamma_i^{\mathrm{col}}) \norm{\tilde e_{n,i}}_2^2 \le 2\bar\gamma \sum_{i=1}^N \norm{\tilde e_{n,i}}_2^2,
\]
where $\gamma_i^{\mathrm{col}}\triangleq \sum_{j\neq i}\gamma_{j,i}$, we have
\[
\mathbb E[W_{n+1}\mid\mathcal F_n] \le \frac{n}{n+1} W_n+ \frac{\kappa}{n+1} \sum_{i=1}^N \norm{\tilde e_{n,i}}_2^2 + \frac{C_\delta}{n(n+1)},
\]
where $\kappa = 2 L_\phi^{ob} C_x \bar\gamma - \tfrac{C_M}{1+\delta_h}$ and $C_\delta = \tfrac{N C_x^2 C_\eta^2}{C_M-\delta}$. By Lemma~\ref{lem:spectralboundadaptivecovariancematrix}, for all $\delta \in (0, C_M)$, eventually a.s., $\sum_{i=1}^N \norm{\tilde e_{n,i}}_2^2 \le \tfrac{1}{C_{M} - \delta} W_n$ and $\sum_{i=1}^N \norm{\tilde e_{n,i}}_2^2 \ge \tfrac{1}{C_x^2} W_n$. We consider three cases of condition~\eqref{eq:globalconvergencecondition}. Suppose $2 \bar\gamma L_\phi^{ob} C_x < C_{M}$. Then, since $\delta$ and $\delta_h$ can be arbitrarily small, $\kappa < 0$ eventually, and eventually a.s.,
\[
\mathbb{E}[W_{n+1} | \mathcal{F}_n] \le W_n- \left(1 - \frac{\kappa}{C_x^2}\right)\frac{1}{n+1} W_n + \frac{C_\delta}{n(n+1)}.
\]
Because $\kappa < 0$, by \citet{robbinsConvergenceTheoremNonnegative1971} $W_n$ converges a.s. and $\sum_{n=1}^{\infty} \tfrac{1 - \kappa / C_x^2}{n+1} W_n < \infty$ a.s., implying that $W_n \rightarrow 0$ a.s. If $2 \bar\gamma L_\phi^{ob} C_x = C_{M}$, then $\kappa$ goes to 0 as $\delta_h \rightarrow 0$ at the rate of $O(1/n)$. Then, the second term in the $W_{n+1}$ recursion can be absorbed into the last term, and we still have $W_n \rightarrow 0$ a.s. If $\bar\gamma L_\phi^{ob} C_x < C_{M} < 2 \bar\gamma L_\phi^{ob} C_x$, then $\kappa > 0$ infinitely often, and for all $n$ large enough, $\Delta_\delta \triangleq 1 - \tfrac{\kappa}{C_{M} - \delta} \in (0,1)$. We have
\[
\mathbb{E}[W_{n+1} | \mathcal{F}_n] \le W_n - \frac{\Delta_\delta}{n+1} W_n + \frac{C_\delta}{n(n+1)}.
\]
Then, the Robbins--Siegmund conditions are satisfied and $W_n \rightarrow 0$ a.s. Since $\sum_{i=1}^N \norm{\tilde e_{n,i}}_2^2 \le \tfrac{1}{C_{M} - \delta} W_n$, we have that $\tilde \theta_{n,i} \rightarrow \theta_i^*$ a.s. for all $i$ and the same for $\hat \theta_{n,i}$ since the projections are slack eventually. By the continuity of $\phi^{ob}$, we have $\tilde p_{n,i} \rightarrow p_i^{NE}$ a.s. for all $i$. Taking the unconditional expectation on the $W_{n+1}$ recursion and Lemma~\ref{lem:nonhomogeneouslinearrecursion} give the desired rates, which carry over to the MSE by $\sum_{i=1}^{N} \norm{\tilde e_{n,i}}_2^2 \le \tfrac{1}{C_{M} - \delta} W_n$ and asymptotically slack projections. Since $\phi^{ob}$ is uniformly Lipschitz continuous on the compact set $\bigcup_i\Theta_i^{ob}$, the same bounds hold for the greedy prices. By Proposition~\ref{prop:dynamicbenchmark},
\[
\mathrm{\mathbf{RegLe}}(T) \asymp \sum_{n\le T}\mathbb{E}\|\tilde{\mathbf p}_n - \mathbf p^{NE}\|_2^2.
\]
Summing the greedy-price MSE rate just established over $n\le T$ gives the aggregate learning regret of Theorem~\ref{thm:globalconvergencetocompetitiveoutcome}: $O(\log T)$, $O((\log T)^2)$, and $O(T^{2\rho-1})$ for $\rho<1/2$, $\rho=1/2$, and $\rho\in(1/2,1)$, respectively. Since oblivious sellers explore persistently ($\nu_{n,i}^2=\Theta(1)$), the exploration tax is $\Delta_i^{\mathrm{ex}}(T)=\Theta(T)$, so the total regret $\Delta_i(T)=\Theta(T)$.

For the surplus capture of part~(iii), we record a revenue lemma also used for Theorems~\ref{thm:allinformedmeanforecast} and~\ref{thm:mixedmarketconvergence}.

\begin{lemma}[Asymptotic per-period revenue]\label{lem:asymp-revenue}
Suppose $\tilde p_{n,i}\to p_i^{NE}$ a.s.\ for every $i\in[N]$. Then, the \emph{realized} per-period revenue time-average $\bar R_{T,i} \;\triangleq\; \tfrac{1}{T}\sum_{n=1}^T R_{n,i}$ satisfies $\bar R_{T,i} = \Pi_i^{NE} - |\beta_i|\bar\nu_{T,i}^2 + o(1)$ a.s. as $T\to\infty$, where $\bar\nu_{T,i}^2 \triangleq T^{-1}\sum_{n=1}^T \nu_{n,i}^2$.
\end{lemma}

\begin{proof}[Proof of Lemma~\ref{lem:asymp-revenue}]
The realized per-period revenue at time $n$ is $R_{n,i} \triangleq R_i(\mathbf p_n) + p_{n,i}\,\varepsilon_{n,i}$ with $\mathbb{E}[R_{n,i}\mid\mathcal{F}_{n-1},\mathbf p_n] = R_i(\mathbf p_n)$. Note that
\[
R_i(\mathbf p) - \Pi_i^{NE} = p_i^{NE}\sum_{j\ne i} \gamma_{i,j}(p_j - p_j^{NE})+ \beta_i (p_i - p_i^{NE})^2 + \sum_{j\ne i} \gamma_{i,j} (p_i - p_i^{NE})(p_j - p_j^{NE}).
\]
Since $p_{n,i}=\tilde p_{n,i}+z_{n,i}$~\eqref{eq:realized-price}, we have $p_{n,i} - p_i^{NE} = \Delta p_{n,i} + z_{n,i}$ with $\Delta p_{n,i}\triangleq\tilde p_{n,i} - p_i^{NE}\in\mathcal{F}_{n-1}$. Conditioning on $\mathcal{F}_{n-1}$,
\[
\mathbb{E}[R_{n,i} - \Pi_i^{NE} \mid \mathcal{F}_{n-1}] =p_i^{NE}\sum_{j\ne i}\gamma_{i,j}\,\Delta p_{n,j}+ \beta_i[(\Delta p_{n,i})^2 + \nu_{n,i}^2] + \sum_{j\ne i} \gamma_{i,j}\,\Delta p_{n,i}\,\Delta p_{n,j}.
\]
The martingale residual $R_{n,i} - \mathbb{E}[R_{n,i}\mid\mathcal{F}_{n-1}]$ has bounded martingale differences, so $\tfrac{1}{T}\sum_{n=1}^T (R_{n,i} - \mathbb{E}[R_{n,i}\mid\mathcal{F}_{n-1}]) \to 0$ a.s. by the strong law. Combining, $\bar R_{T,i} - \Pi_i^{NE} = -|\beta_i|\bar\nu_{T,i}^2 + R^{(\Delta)}_{T,i} + o(1)$ a.s., where $R^{(\Delta)}_{T,i} \triangleq p_i^{NE}\sum_{j\ne i}\gamma_{i,j}\tfrac{1}{T}\sum_{n=1}^T\Delta p_{n,j}+ \beta_i\tfrac{1}{T}\sum_{n=1}^T (\Delta p_{n,i})^2 + \sum_{j\ne i} \gamma_{i,j}\tfrac{1}{T}\sum_{n=1}^T\Delta p_{n,i}\Delta p_{n,j} \to 0$ because $\Delta p_{n,k}\to 0$ a.s.\ for every $k$. 
\end{proof}

By part~(i), $\tilde p_{n,i}\to p_i^{NE}$ a.s.; because oblivious exploration is constant ($\nu_{n,i}^2\equiv\nu_i^2$), $\bar\nu_{T,i}^2=\nu_i^2$, and Lemma~\ref{lem:asymp-revenue} gives $\bar R_{T,i}\to\Pi_i^{NE}-|\beta_i|\nu_i^2$ a.s. Dividing by $\Pi_i^{C}-\Pi_i^{NE}>0$ yields the surplus capture $S_i=-|\beta_i|\nu_i^2/(\Pi_i^{C}-\Pi_i^{NE})<0$ of part~(iii).

\subsection{Proof of Theorem~\ref{thm:ode-local-stability}}\label{proof:thm:ode-local-stability}
Let $(m^*,Q^*)$ be any equilibrium of \eqref{eq:ode-N-mQ}. From $\dot m=0$ we have $m_i^* = p_i^g(m^*,Q^*)$ for all $i$. From $\dot Q=0$ we obtain
\[
\begin{aligned}
&Q_{ij}^* = p_i^g(m^*,Q^*)p_j^g(m^*,Q^*) = m_i^* m_j^* \quad (i\neq j),\\
&Q_{ii}^* = (p_i^g(m^*,Q^*))^2+\nu_i^2 = (m_i^*)^2+\nu_i^2.
\end{aligned}
\]
Therefore, $S_{ij}(m^*,Q^*)=Q_{ij}^*-m_i^* m_j^*=0$ for all $i\neq j$, and $V_i(m^*,Q^*)=Q_{ii}^*-(m_i^*)^2=\nu_i^2$ for all $i$. Plugging $S_{ij}=0$ into \eqref{eq:a-map-N-ode}--\eqref{eq:b-map-N-ode} yields $b_i(m^*,Q^*)=\beta_i$ and $a_i(m^*,Q^*)=\alpha_i+\sum_{j\neq i}\gamma_{i,j}m_j^*$. Using \eqref{eq:pg-map-N-ode} and $m_i^*=p_i^g(m^*,Q^*)$ then gives, for each $i$,
\[
2\beta_i\,m_i^*+\sum_{j\neq i}\gamma_{i,j}m_j^*=-\alpha_i,
\]
which coincides with the first-order condition for the full-information Nash equilibrium; by Section~\ref{sec:solution-concept}, $m^*=\mathbf p^{NE}$ is the unique solution. The result for $Q^*$ then follows.

For local asymptotic stability, we linearize \eqref{eq:ode-N-mQ} around $(m^*,Q^*)$. Define $y_i(m,Q)\triangleq p_i^g(m,Q)-m_i$ and the centered second-moment deviations $v_i\triangleq (Q_{ii}-m_i^2)-\nu_i^2$ and $s_{ij}\triangleq Q_{ij}-m_im_j$ for $i\neq j$. In a neighborhood of $(m^*,Q^*)$ this defines a smooth change of coordinates from $(m,Q)$ to $(m,v,s)$. By \eqref{eq:ode-N-mQ}, we have
\begin{equation}\label{eq:ode-vs-exact-rep}
\dot v_i = -v_i + y_i(m,Q)^2,\qquad
\dot s_{ij} = -s_{ij} + y_i(m,Q)\,y_j(m,Q)\quad (i\neq j),
\end{equation}
together with $\dot m_i=y_i(m,Q)$. At equilibrium, $y(m^*,Q^*)=0$, $v^*= 0$, and $s^*= 0$. Let $J$ denote the Jacobian of the vector field of \eqref{eq:ode-N-mQ} in $(m,v,s)$ coordinates, evaluated at $(m^*,0,0)$. Since the maps $y\mapsto y_i^2$ and $(y_i,y_j)\mapsto y_i y_j$ have derivative zero at $y=0$, linearizing \eqref{eq:ode-vs-exact-rep} yields $\delta\dot v = -\delta v$ and $\delta\dot s = -\delta s$. Therefore, $J$ has the block upper-triangular form $J = \left(\begin{smallmatrix} J_{mm} & *\\
0 & -I \end{smallmatrix}\right)$, where the $-I$ block corresponds to $(v,s)$. Hence, the spectrum of $J$ is the union of $\mathrm{spec}(J_{mm})$ and $\{-1\}$, and it remains to show that $J_{mm}$ is Hurwitz. The mean block is the Jacobian of $m\mapsto \dot m$ at equilibrium with $(v,s)$ held fixed. Since $s^*\equiv 0$, the greedy map simplifies to the linear best response
\[
p_i^g(m,Q^*)= -\frac{\alpha_i+\sum_{j\neq i}\gamma_{i,j}m_j}{2\beta_i},
\]
so $\tfrac{\partial p_i^g}{\partial m_j}(m^*,Q^*)= -\tfrac{\gamma_{i,j}}{2\beta_i}\ (i\neq j)$ and $\tfrac{\partial p_i^g}{\partial m_i}(m^*,Q^*)=0$. It follows that $(J_{mm})_{ii}=-1$ and $(J_{mm})_{ij}=-\frac{\gamma_{i,j}}{2\beta_i}$ for $i\neq j$. Define $A$ by $A_{ii}=0$ and $A_{ij}=\tfrac{\gamma_{i,j}}{-2\beta_i}$ for $i\neq j$. We therefore have $J_{mm} = A - I$. Because $\beta_i<0$ for all $i$, $A$ is entrywise non-negative and its row sums satisfy $\sum_{j}A_{ij}=\sum_{j\neq i}\tfrac{\gamma_{i,j}}{-2\beta_i}<\tfrac12$. Hence, $\|A\|_\infty<1/2$ and every eigenvalue $\mu$ of $A$ satisfies $|\Re(\mu)|\le \|A\|_\infty<1/2$. Thus, every eigenvalue $\lambda$ of $J_{mm}$ satisfies
\[
\Re(\lambda)=\Re(\mu-1)\le |\Re(\mu)|-1 < -1/2.
\]
Therefore, $J_{mm}$ is Hurwitz, and consequently the full Jacobian $J$ is Hurwitz. Standard linearization theory implies local asymptotic stability of $(m^*,Q^*)$ for \eqref{eq:ode-N-mQ}.

\subsection{Proof of Proposition~\ref{prop:dynamicbenchmark}}\label{proof:prop:dynamicbenchmark}

Denote $p_{n,i}^* = \mathrm{BR}_i(\mathbf{p}_{n,-i})$. Then,
\[
R_i(p_{n,i}^*, \mathbf{p}_{n,-i}) - R_i(p_{n,i}, \mathbf{p}_{n,-i})= \alpha_i (p_{n,i}^* - p_{n,i}) + \beta_i [(p_{n,i}^*)^2 - p_{n,i}^2] + \sum_{j\neq i}\gamma_{i,j}p_{n,j}(p_{n,i}^* - p_{n,i}).
\]
The first-order condition $\alpha_i + 2\beta_i p_{n,i}^* + \sum_{j\neq i}\gamma_{i,j} p_{n,j} = 0$ yields
\[
\alpha_i p_{n,i}^* + 2 \beta_i (p_{n,i}^*)^2 + \sum_{j\neq i}\gamma_{i,j} p_{n,j}\, p_{n,i}^* = 0.
\]
So,
\[
R_i(p_{n,i}^*, \mathbf{p}_{n,-i}) - R_i(p_{n,i}, \mathbf{p}_{n,-i}) = - \beta_i (p_{n,i}^* - p_{n,i})^2= - \beta_i \left[\sum_{j\neq i}\frac{\gamma_{i,j}}{-2\beta_i}(p_{n,j} - p_j^{NE}) - (p_{n,i} - p_i^{NE})\right]^2.
\]
Summing over $n$ and taking $\mathbb{E}[\cdot]$ gives~\eqref{eq:Delta-vi}. Define $\mathbf{e}_n \triangleq \mathbf{p}_n - \mathbf{p}^{NE}$ and $v_i \in \mathbb{R}^N$ with $v_{i,i} = -1$ and $v_{i,j} = \gamma_{i,j}/(-2\beta_i)$ for $j\neq i$ so~\eqref{eq:Delta-vi} reads $\Delta_i(T) = -\beta_i\,\mathbb E\!\left[\sum_{n=1}^T (v_i^\top\mathbf{e}_n)^2\right]$. The per-seller upper bound is then immediate by Cauchy--Schwarz. For the aggregate equivalence, stack the row vectors $v_i$ into the matrix $V \triangleq D_\beta^{-1}\Gamma$, where $D_\beta \triangleq \mathrm{diag}(-2\beta_i)$ and $\Gamma$ is the Nash matrix from Section~\ref{sec:solution-concept} ($\Gamma_{ii}=2\beta_i$, $\Gamma_{ij}=\gamma_{i,j}$ for $i\neq j$), and let $D \triangleq \mathrm{diag}(-\beta_i)\succ 0$. Summing~\eqref{eq:Delta-vi} over $i$, $\sum_{i=1}^{N} \Delta_i(T)=\mathbb{E}[\sum_{n=1}^{T} \mathbf{e}_n^\top V^\top D V\mathbf{e}_n]$. Since $-\beta_i>\gamma_i$ for all $i$, $\Gamma$ and thus $V$ are invertible, and so $V^\top D V \succ 0$. Therefore,
\[
\lambda_{\min}(V^\top D V)\|\mathbf{e}_n\|_2^2 \le \mathbf{e}_n^\top V^\top D V \mathbf{e}_n \le \lambda_{\max}(V^\top D V)\|\mathbf{e}_n\|_2^2
\]
and the result follows. Decomposition~\eqref{eq:regret-decomp} follows by the fact that $z_{n,i}$ are independent across sellers with mean zero and variance $\nu_{n,i}^2$.

\subsection{Proof of Theorem~\ref{thm:allinformedmeanforecast}}\label{proof:thm:allinformedmeanforecast}

\paragraph*{Step 1: background.} For notational simplicity, we drop the superscript $in$. Define the vector-valued map $\boldsymbol\Phi(\theta;\mathbf q) \in \mathbb R^N$ with $i$\textsuperscript{th} coordinate $\boldsymbol\Phi_i(\theta_i; \mathbf q) \triangleq -\tfrac{\alpha_i + \sum_{j \neq i}\gamma_{i,j} q_j}{2\beta_i}$. The Nash first-order condition gives $\boldsymbol\Phi_i(\theta^*; \mathbf q) - p_i^{NE} = \sum_{j \neq i}(-\tfrac{\gamma_{i,j}}{2\beta_i})(q_j - p_j^{NE})$. Stacking into vector form yields
\begin{equation}\label{eq:cor:linear-BR}
\boldsymbol\Phi(\theta^*;\mathbf q) \;-\; \mathbf p^{NE} \;=\; B\,(\mathbf q - \mathbf p^{NE}).
\end{equation}
Since $\Theta^{in}_i$ has positive distance from $\{\beta_i = 0\}$ (Section~\ref{sec:learning-dynamics}), by the mean-value theorem
\begin{equation}\label{eq:cor:phi-lipschitz}
\|\boldsymbol\Phi(\hat\theta_n; \mathbf m_n) - \boldsymbol\Phi(\theta^*; \mathbf m_n)\|_2 \;\le\; K\,\sum_{i \in [N]}\|\hat\theta_{n,i} - \theta^*_i\|_2
\end{equation}
for some constant $K>0$. Since $\mathbf p_{n+1} = \boldsymbol\Phi(\hat\theta_n;\mathbf m_n) + \mathbf z_{n+1}$, by \eqref{eq:cor:linear-BR} with $\mathbf q = \mathbf m_n$ we have
\begin{align}
\mathbf p_{n+1} - \mathbf p^{NE}
&\;=\; B\,(\mathbf m_n - \mathbf p^{NE}) \;+\; \boldsymbol\delta_n \;+\; \mathbf z_{n+1}, \label{eq:cor:price-decomp}
\end{align}
where $\boldsymbol\delta_n \;\triangleq\; \boldsymbol\Phi(\hat\theta_n; \mathbf m_n) - \boldsymbol\Phi(\theta^*; \mathbf m_n)$. Since $\boldsymbol\delta_n$ is $\mathcal F_n$-measurable, by \eqref{eq:cor:phi-lipschitz}, Cauchy--Schwarz, and Theorem~\ref{thm:informedolsrate} we have
\begin{equation}\label{eq:cor:delta-bound}
\mathbb{E}\|\boldsymbol\delta_n\|_2^2 \;\le\; K^2 N\,\sum_{i \in [N]}\mathbb{E}\|\hat\theta_{n,i} - \theta^*_i\|_2^2 \;=\; O\!\left(n^{\eta_{\max} - 1}\log n\right).
\end{equation}

\paragraph*{Step 2: stochastic approximation.} Let $\mathbf e_n \triangleq \mathbf m_n - \mathbf p^{NE}$. Since $\mathbf m_{n+1} = \mathbf m_n + (n+1)^{-1}(\mathbf p_{n+1} - \mathbf m_n)$, we have
\[
\mathbf e_{n+1} = \left(1 - \frac{1}{n+1}\right)\mathbf e_n + \frac{1}{n+1}\left(\mathbf p_{n+1} - \mathbf p^{NE}\right).
\]
By \eqref{eq:cor:price-decomp}, $\mathbf e_{n+1} = [I - \tfrac{I - B}{n+1}]\mathbf e_n + \tfrac{\boldsymbol\delta_n + \mathbf z_{n+1}}{n+1}$. Define
\begin{equation}\label{eq:cor:An-def}
A_n \;\triangleq\; I - \frac{I - B}{n + 1},
\end{equation}
so that $\mathbf e_{n+1} = A_n \mathbf e_n + (n+1)^{-1}(\boldsymbol\delta_n + \mathbf z_{n+1})$. Note that $\|B\|_\infty = \max_{i \in [N]} \sum_{j \neq i}|B_{ij}| < 1 / 2$ and so $I - B$ is positive stable, so the Lyapunov theorem produces a unique symmetric $P \succ 0$ such that $P(I-B) + (I-B)^\top P = I$, which gives $I \succeq \sigma_{\max}(P)^{-1} P$ and $P(I-B) + (I-B)^\top P \succeq \mu P$ with $\mu \triangleq \sigma_{\max}(P)^{-1} > 0$. Computing $A_n^\top P A_n$ from \eqref{eq:cor:An-def} gives
\[
A_n^\top P A_n=P- \frac{I}{n+1}+\frac{(I-B)^\top P (I-B)}{(n+1)^2}.
\]
Using $-I \preceq -\mu P$ and bounding the second-order term $(I-B)^\top P (I-B) \preceq \|I-B\|_2^2\,\sigma_{\max}(P)\,I \preceq C_3\,P$ with $C_3 \triangleq \|I-B\|_2^2\,\sigma_{\max}(P)/\sigma_{\min}(P)$,
\begin{equation}\label{eq:cor:AnPAn-tight}
A_n^\top P A_n \;\preceq\; \left(1 - \frac{\mu}{n+1} + \frac{C_3}{(n+1)^2}\right) P.
\end{equation}
For all $n \ge n_0 \triangleq 2C_3/\mu$, $C_3/(n+1)^2 \le \mu/(2(n+1))$, so $A_n^\top P A_n \preceq P$. Define the Lyapunov $V_n \triangleq \mathbf e_n^\top P\mathbf e_n$ and the root-mean-square $r_n \triangleq \sqrt{\mathbb E V_n} = \sqrt{\mathbb E\|\mathbf e_n\|_P^2}$. Since $\mathbf e_{n+1} = A_n\mathbf e_n + (n+1)^{-1}(\boldsymbol\delta_n + \mathbf z_{n+1})$, by triangle inequality
\[
\|\mathbf e_{n+1}\|_P \le \|A_n \mathbf e_n\|_P + \frac{\|\boldsymbol\delta_n\|_P}{n+1} + \frac{\|\mathbf z_{n+1}\|_P}{n+1}.
\]
Squaring, taking expectations, taking square root, and applying Minkowski's inequality yield
\begin{equation}\label{eq:cor:rn-Mink}
r_{n+1} \;\le\; \sqrt{\mathbb E\|A_n\mathbf e_n\|_P^2} \;+\; \frac{\sqrt{\mathbb E\|\boldsymbol\delta_n\|_P^2}}{n+1} \;+\; \frac{\sqrt{\mathbb E\|\mathbf z_{n+1}\|_P^2}}{n+1}.
\end{equation}
For the first term, \eqref{eq:cor:AnPAn-tight} gives $\mathbb E\|A_n\mathbf e_n\|_P^2 \le (1 - \mu/(n+1) + C_3/(n+1)^2)\,r_n^2$; the elementary inequality $\sqrt{1+x} \le 1 + x/2$ for $x \ge -1$ then yields, for $n \ge n_0$, $\sqrt{\mathbb E\|A_n\mathbf e_n\|_P^2} \le (1 - \tfrac{\mu}{2(n+1)} + \tfrac{C_3}{2(n+1)^2}) r_n$. Since $\sqrt{\mathbb E\|\boldsymbol\delta_n\|_P^2}/(n+1) \lesssim n^{(\eta_{\max}-3)/2}(\log n)^{1/2}$ and $\sqrt{\mathbb E\|\mathbf z_{n+1}\|_P^2}/(n+1) \lesssim n^{-(2+\eta_{\min})/2}$, \eqref{eq:cor:rn-Mink} gives, for $n \ge n_0$, $r_{n+1} \le \rho_n r_n + g_n$, where $\rho_n \triangleq 1 - \mu/(2(n+1)) + C_3/(2(n+1)^2)$ and $g_n \triangleq K_\delta n^{(\eta_{\max}-3)/2}(\log n)^{1/2} + K_z n^{-(2+\eta_{\min})/2}$. Iterating from $n_0$ to $n - 1$ gives $r_n \le (\prod_{m=n_0}^{n-1}\rho_m) r_{n_0} + \sum_{k=n_0}^{n-1} g_k\prod_{m=k+1}^{n-1}\rho_m$. By Lemma~\ref{lem:nonhomogeneouslinearrecursion}, we have
\[
r_n \lesssim n^{-\mu/2}(\log n)^{3/2} + n^{(\eta_{\max}-1)/2}(\log n)^{1/2} + n^{-\eta_{\min}/2}.
\]
Squaring and norm equivalence give 
\begin{equation}\label{eq:cor:mn-rate}
\mathbb E\|\mathbf m_n - \mathbf p^{NE}\|_2^2 \;\le\; \frac{r_n^2}{\sigma_{\min}(P)}
\;=\; O\!\left(n^{-\mu}(\log n)^3 + n^{\eta_{\max} - 1}\log n + n^{-\eta_{\min}}\right).
\end{equation}

\paragraph*{Step 3: convergence.} Squaring \eqref{eq:cor:price-decomp} gives $\|\mathbf p_{n+1} - \mathbf p^{NE}\|_2^2 = \|B\mathbf e_n + \boldsymbol\delta_n + \mathbf z_{n+1}\|_2^2$. Expanding and taking iterated expectations gives
\[
\mathbb E\|\mathbf p_{n+1} - \mathbf p^{NE}\|_2^2\le 2\|B\|_2^2 \mathbb E\|\mathbf e_n\|_2^2 + 2\mathbb E\|\boldsymbol\delta_n\|_2^2 + N\,\nu_{\max,n+1}^2.
\]
Since $\|B\|_2 < \infty$, by \eqref{eq:cor:delta-bound} and \eqref{eq:cor:mn-rate} we have $\mathbb E\|\mathbf p_n - \mathbf p^{NE}\|_2^2 = O(n^{-\mu}(\log n)^3 + n^{\eta_{\max} - 1}\log n + n^{-\eta_{\min}})$. Theorem~\ref{thm:informedolsrate} gives $\hat\theta_{n,k} \to \theta_k^*$ a.s. for every $k \in [N]$. Take the recursion $\mathbf e_{n+1} = A_n\mathbf e_n + (n+1)^{-1}(\boldsymbol\delta_n + \mathbf z_{n+1})$ from Step~2 and compute
\[
\mathbb{E}[V_{n+1}\mid\mathcal{F}_n] = \mathbf e_n^\top A_n^\top P A_n\mathbf e_n+ \frac{2}{n+1}\mathbf e_n^\top A_n^\top P\boldsymbol\delta_n + \frac{1}{(n+1)^2}\|\boldsymbol\delta_n\|_P^2+ \frac{1}{(n+1)^2}\mathbb{E}\|\mathbf z_{n+1}\|_P^2.
\]
Equation~\eqref{eq:cor:AnPAn-tight} bounds the first piece by $(1 - \mu/(n+1) + C_3/(n+1)^2)\,V_n$, and Young's inequality $2\langle u, v\rangle \le \epsilon\|u\|^2 + \epsilon^{-1}\|v\|^2$ applied to $u = P^{1/2}A_n\mathbf e_n$ and $v = (n+1)^{-1}P^{1/2}\boldsymbol\delta_n$ bounds the cross-term as $\tfrac{2}{n+1}\,\mathbf e_n^\top A_n^\top P\,\boldsymbol\delta_n \le\epsilon V_n + \tfrac{\|\boldsymbol\delta_n\|_P^2}{\epsilon\,(n+1)^2}$, where the inequality uses $\|A_n\mathbf e_n\|_P^2 \le V_n$ (holds for $n \ge n_0$ by~\eqref{eq:cor:AnPAn-tight}). Choosing $\epsilon = \mu/(2(n+1))$ gives, eventually a.s.,
\begin{equation}\label{eq:cor:RS-ineq}
\mathbb{E}[V_{n+1}\mid\mathcal{F}_n] \;\le\; \left(1 - \frac{\mu}{2(n+1)}\right)V_n \;+\; h_n
\end{equation}
with
\[
h_n \triangleq \frac{2\sigma_{\max}(P)\|\boldsymbol\delta_n\|_2^2}{\mu(n+1)} + \frac{\sigma_{\max}(P)\|\boldsymbol\delta_n\|_2^2}{(n+1)^2} + \frac{N\sigma_{\max}(P)\nu_{\max,n+1}^2}{(n+1)^2} + \frac{C_3}{(n+1)^2}V_n.
\]
The first two terms are bounded by $O(\|\boldsymbol\delta_n\|_2^2/(n+1))$, which is a.s.-summable by~\eqref{eq:thm:informedolsrate-sum} and \eqref{eq:cor:phi-lipschitz}. The last two terms are trivially summable. Hence, $\sum_n h_n < \infty$ a.s. Robbins--Siegmund applied to~\eqref{eq:cor:RS-ineq} gives that $V_n \to 0$ a.s. By norm equivalence, $\mathbf m_n \to \mathbf p^{NE}$ a.s. By \eqref{eq:cor:linear-BR}, $\tilde{\mathbf p}_{n+1} - \mathbf p^{NE} = B(\mathbf m_n - \mathbf p^{NE})+ \boldsymbol\delta_n$. Since both parts vanish a.s., $\tilde p_{n,k} \to p_k^{NE}$ a.s.\ for every $k\in[N]$.

\paragraph*{Regret.} By Step~2, $\tilde{\mathbf p}_{n+1} - \mathbf p^{NE} = B\mathbf e_n + \boldsymbol\delta_n$, so~\eqref{eq:cor:mn-rate} and~\eqref{eq:cor:delta-bound} give the greedy-price MSE $\mathbb E\|\tilde{\mathbf p}_n - \mathbf p^{NE}\|_2^2 = O(n^{-\mu}(\log n)^3 + n^{\eta_{\max}-1}\log n)$. Summing over $n\le T$ and applying Proposition~\ref{prop:dynamicbenchmark} yields the aggregate learning regret $\mathrm{\mathbf{RegLe}}(T) = O(T^{1-\mu}(\log T)^3 + T^{\eta_{\max}}\log T)$, while the tax is $\mathrm{\mathbf{RegEx}}(T) = -\sum_i\beta_i\sum_{n\le T}(\nu_{n,i}^2 + \sum_{j\ne i}(\gamma_{i,j}/2\beta_i)^2\nu_{n,j}^2) = O(T^{1-\eta_{\min}})$ since $\nu_{n,k}^2 = \Theta(n^{-\eta_k})$. For the surplus capture, every seller's exploration decays ($\nu_{n,k}^2\to0$), so $\bar\nu_{T,k}^2\to0$ and Lemma~\ref{lem:asymp-revenue} gives $\bar R_{T,k}\to\Pi_k^{NE}$; hence $S_k=0$ for every $k\in[N]$. 

\paragraph*{Special case.}
Write $A_s \triangleq I - \tfrac12 (B + B^\top)$ for the symmetric part of $I - B$ and let $v \in \mathbb R^N$ be a unit eigenvector of $P$ corresponding to its largest eigenvalue. Multiplying the Lyapunov equation $P(I-B) + (I-B)^\top P = I$ on the left by $v^\top$ and on the right by $v$ gives $2\sigma_{\max}(P) v^\top A_s v = 1$, where we use $v^\top A_a v = 0$ for the antisymmetric part $A_a \triangleq \tfrac12((I-B) - (I-B)^\top)= -A_a$. Therefore, $\sigma_{\max}(P) \le \tfrac{1}{2\,\lambda_{\min}(A_s)}$. Since $\lambda_{\min}(A_s) = 1 - \tfrac12\,\lambda_{\max}(B+B^\top)$ and by \eqref{eq:cor:symmpart-regularity}, $\mu \ge 2 - \lambda_{\max}(B+B^\top) > 1/2$. Substituting $\mu>1/2$ and $\eta_i=1/2$ into the bounds above gives learning regret $O(\sqrt T\,\log T)$ and tax $O(\sqrt T)$, whence $\Delta_i(T) = O(\sqrt T\,\log T)$ for every $i$.

\paragraph{Condition~\eqref{eq:cor:symmpart-regularity}.}
Write $M \triangleq B + B^\top$, and recall $\|B\|_\infty < 1/2$. If $B$ is symmetric, then $M = 2B$ and $\lambda_{\max}(M) < 1$. In general $\lambda_{\max}(M) \le 2\sqrt{\|B\|_\infty\,\|B\|_1}$, where $\|B\|_1 < (N-1)/2$. Hence, $\lambda_{\max}(M) < \sqrt{N-1}$, which $< 3/2$ when $N \le 3$. If $\beta_i \equiv \beta$, a direct computation gives $M = I - \tfrac{1}{2\beta}H$, so $M \preceq I$ whenever $H \prec 0$.

\subsection{Proof of Theorem~\ref{thm:mixedmarketconvergence}}\label{proof:thm:mixedmarketconvergence}

Theorem~\ref{thm:mixedmarketconvergence} is a corollary of Theorem~\ref{thm:mixedmarketconvergence-primitive}. Condition~(i) is common to both, and it therefore suffices to verify that conditions~(ii) and~(iii) of Theorem~\ref{thm:mixedmarketconvergence} imply conditions~(ii) and~(iii) of Theorem~\ref{thm:mixedmarketconvergence-primitive}. Let $\gamma^r_i \triangleq \sum_{j\ne i}\gamma_{i,j}$ and $\gamma^c_j \triangleq \sum_{i\ne j}\gamma_{i,j}$ denote the row and column sums of $\Gamma^{\mathrm{off}}$, and let $b^r_i \triangleq \sum_{j\ne i}\tfrac{\gamma_{i,j}}{2|\beta_i|}$ and $b^c_j \triangleq \sum_{i\ne j}\tfrac{\gamma_{i,j}}{2|\beta_i|}$ denote those of $B$. By definition, $\max_i\gamma^r_i + \max_j\gamma^c_j = 2\bar\gamma$ and $\max_i b^r_i + \max_j b^c_j = 2\bar\mu$, so $\bar\gamma^{ob} \le \tfrac12(\max_i\gamma^r_i + \max_j\gamma^c_j) = \bar\gamma$. Factoring the two summands in the definition of $\bar\Lambda$, bounding each factor by its aggregate counterpart, and using $L_\phi^{in,\theta}\le L$ and $\max_{j\in\mathcal{I}^{in}}\gamma_j/(2|\beta_j|)\le\max_j b^r_j$ give $\bar\Lambda \le 2\bar\gamma (2\bar\mu + L)$. Swapping the order of summation and bounding the inner oblivious sum by the full column sum of $\Gamma^{\mathrm{off}}$ give $\bar\Psi \le 4\bar\gamma\bar\mu$.  Bounding the diagonal-contraction minimum from below by $2-L$ and the informed cross-drag from above by $\max_k b^c_k$ give $\bar\kappa \ge 2 - L - 2\bar\mu$. Condition~(ii) of Theorem~\ref{thm:mixedmarketconvergence} states $L + 2\bar\mu < 2$, so $\bar\kappa \ge 2 - L - 2\bar\mu > 0$, which is condition~(ii) of Theorem~\ref{thm:mixedmarketconvergence-primitive}. Substituting the derived inequalities into the right-hand side of Theorem~\ref{thm:mixedmarketconvergence-primitive} (iii) and using $L_\phi^{ob}\le L$ and Condition~(iii) of Theorem~\ref{thm:mixedmarketconvergence} give
\[
C_x(2\bar\gamma^{ob}\,L_\phi^{ob} + \bar\Lambda + L_\phi^{ob}\bar\Psi/\bar\kappa) \le 4C_x\bar\gamma\left(L + \frac{2- 2\bar\mu}{2 - 2\bar\mu - L} \cdot \bar\mu\right) < C_M,
\]
so condition~(iii) of Theorem~\ref{thm:mixedmarketconvergence-primitive} holds. Parts~(i)--(iii) of Theorem~\ref{thm:mixedmarketconvergence} then follow directly from the corresponding parts of Theorem~\ref{thm:mixedmarketconvergence-primitive}. Note that $\eta^* \in (0,1-\eta_{\max})$ summarizes the three cases of the learning regret. 

\subsection{Proof of Theorem~\ref{thm:mixedmarketconvergence-primitive}}\label{proof:thm:mixedmarketconvergence-primitive}

We follow the structure of the proof of Theorem~\ref{thm:globalconvergencetocompetitiveoutcome}. Recall that $h_{n+1,i} \triangleq x_{n+1,i}^\top S_{n,i}^{-1} x_{n+1,i}$ satisfies $h_{n+1,i} \le \delta_h$ eventually a.s.\ for any $\delta_h > 0$ by Lemma~\ref{lem:spectralboundadaptivecovariancematrix}. Write: $\tilde e_{n,i}^{ob} \triangleq \tilde\theta_{n,i}^{ob} - \theta_i^{*,ob}$ for $i\in\mathcal{I}^{ob}$, $\tilde e_{n,j}^{in} \triangleq \hat\theta_{n,j}^{in} - \theta_j^{*,in}$ for $j\in\mathcal{I}^{in}$, $u_{n,k} \triangleq m_{n,k} - p_k^{NE}$ for $k\in[N]$, $U_{n,k} \triangleq u_{n,k}^2$, $U_n^{tot} \triangleq \sum_{k=1}^N U_{n,k}$, $W_{n,i} \triangleq \tfrac{1}{n} (\tilde e_{n,i}^{ob})^\top S_{n,i} \tilde e_{n,i}^{ob}$, and $W_n^{ob} \triangleq \sum_{i \in \mathcal{I}^{ob}} W_{n,i}$. Within this proof we further split $\bar\Lambda = \bar\Delta + \bar\Theta$ and $\bar\kappa = \bar c_{\mathrm{diag}} - \bar D$ into their four constituents:
\[
\begin{aligned}
\bar\Delta &\triangleq \max_{i \in \mathcal{I}^{ob}}\sum_{j \in \mathcal{I}^{in}}\frac{\gamma_{i,j}\,\gamma_j}{2|\beta_j|},\\
\bar\Theta &\triangleq L_\phi^{in,\theta}\max_{i \in \mathcal{I}^{ob}}\sum_{j \in \mathcal{I}^{in}}\gamma_{i,j},\\
\bar c_{\mathrm{diag}} &\triangleq \min\{2 - L_\phi^{ob},\, 2 - L_\phi^{in,\theta} - \max_{j \in \mathcal{I}^{in}}\frac{\gamma_j}{2|\beta_j|}\},\\
\bar D &\triangleq \max_{k \in [N]}\sum_{j \in \mathcal{I}^{in}\setminus\{k\}}\frac{\gamma_{j,k}}{2|\beta_j|}.
\end{aligned}
\]
We also abbreviate the oblivious-side small-gain slack from condition~(iii) as $K \triangleq C_M - C_x(2\bar\gamma^{ob} L_\phi^{ob} + \bar\Lambda) > 0$.

\paragraph{Step 1. Oblivious OLS recursion.}
For each $i \in \mathcal{I}^{ob}$, recall $\eta_{n,i} = \varepsilon_{n,i} + \sum_{j \ne i}\gamma_{i,j}(p_{n,j} - p_j^{NE})$. Steps identical to proof of Theorem~\ref{thm:globalconvergencetocompetitiveoutcome} give, eventually a.s.,
\begin{align}
\mathbb{E}[W_{n+1,i}\mid\mathcal{F}_n]
&\;\le\; \frac{n}{n+1}\,W_{n,i} \;-\; \frac{C_M\,\norm{\tilde e_{n,i}^{ob}}_2^2}{(n+1)(1+\delta_h)}\notag\\
&\quad +\;\frac{2\,\left|\mathbb{E}[\eta_{n+1,i}\, x_{n+1,i}^\top \tilde e_{n,i}^{ob}\mid\mathcal{F}_n]\right|}{n+1} \;+\; \frac{C_1}{n(n+1)},\label{eq:proof-mm-Wrec-init}
\end{align}
where $C_1 \triangleq C_x^2\,C_\eta^2/(C_M - \delta)$.

\paragraph{Step 2. Bounding the cross term.}
By the same steps as in the proof of Theorem~\ref{thm:globalconvergencetocompetitiveoutcome},
\[
\mathbb{E}[\eta_{n+1,i} x_{n+1,i}^\top \tilde e_{n,i}^{ob}\mid\mathcal{F}_n]= \sum_{j\ne i}\gamma_{i,j}\mathbb{E}[p_{n+1,j} - p_j^{NE}\mid\mathcal{F}_n]\cdot\mathbb{E}[x_{n+1,i}^\top \tilde e_{n,i}^{ob}\mid\mathcal{F}_n],
\]
where $|\mathbb{E}[x_{n+1,i}^\top \tilde e_{n,i}^{ob}\mid\mathcal{F}_n]| \le C_x\norm{\tilde e_{n,i}^{ob}}_2$. The first factor depends on whether $j$ is oblivious or informed. For $j\in\mathcal{I}^{ob}\setminus\{i\}$,
\begin{equation}\label{eq:proof-mm-ob-comp}
\left|\mathbb{E}[p_{n+1,j} - p_j^{NE}\mid\mathcal{F}_n]\right| \;=\; \left|\phi^{ob}(\hat\theta_{n,j}^{ob}) - \phi^{ob}(\theta_j^{*,ob})\right| \;\le\; L_\phi^{ob}\,\norm{\tilde e_{n,j}^{ob}}_2
\end{equation}
since projection is non-expansive. For $j\in\mathcal{I}^{in}$, 
\begin{align}
&\quad\; \left|\mathbb{E}[p_{n+1,j} - p_j^{NE}\mid\mathcal{F}_n]\right| = \left|\phi_j^{in}(\hat\theta_{n,j}^{in};\, \mathbf m_{n,-j}) - \phi_j^{in}(\theta_j^{*,in};\, \mathbf p_{-j}^{NE})\right| \notag\\
&\le \left|\phi_j^{in}(\hat\theta_{n,j}^{in};\, \mathbf m_{n,-j}) - \phi_j^{in}(\theta_j^{*,in};\, \mathbf m_{n,-j})\right| +\;\left|\phi_j^{in}(\theta_j^{*,in};\, \mathbf m_{n,-j}) - \phi_j^{in}(\theta_j^{*,in};\, \mathbf p_{-j}^{NE})\right| \notag \\
&\le L_\phi^{in,\theta}\,\norm{\tilde e_{n,j}^{in}}_2 \;+\; \sum_{k\ne j}\frac{\gamma_{j,k}}{2|\beta_j|}\,|u_{n,k}|. \label{eq:proof-mm-in-comp}
\end{align}
Combining, the cross term satisfies
\[
\begin{aligned}
&|\mathbb{E}[\eta_{n+1,i} x_{n+1,i}^\top \tilde e_{n,i}^{ob}\mid\mathcal{F}_n]|\\
&\le C_x\norm{\tilde e_{n,i}^{ob}}_2\left[\sum_{j\in\mathcal{I}^{ob}\setminus\{i\}} \gamma_{i,j} L_\phi^{ob} \norm{\tilde e_{n,j}^{ob}}_2\right.\\
&\quad \left.+ \sum_{j\in\mathcal{I}^{in}}\gamma_{i,j}\left(L_\phi^{in,\theta}\norm{\tilde e_{n,j}^{in}}_2 + \sum_{k\ne j}\frac{\gamma_{j,k}}{2|\beta_j|}|u_{n,k}|\right)\right].
\end{aligned}
\]
We next apply Young's inequality to each product on the right-hand side. For the first, Young's gives
\[
\sum_{j\in\mathcal{I}^{ob}\setminus\{i\}} \gamma_{i,j} L_\phi^{ob}\norm{\tilde e_{n,i}^{ob}}_2\norm{\tilde e_{n,j}^{ob}}_2\le \frac{L_\phi^{ob}}{2}\sum_{j\in\mathcal{I}^{ob}\setminus\{i\}}\gamma_{i,j}\left(\norm{\tilde e_{n,i}^{ob}}_2^2 + \norm{\tilde e_{n,j}^{ob}}_2^2\right).
\]
Summing over $i\in\mathcal{I}^{ob}$ gives
\[
\sum_{i\in\mathcal{I}^{ob}}\sum_{j\in\mathcal{I}^{ob}\setminus\{i\}}\gamma_{i,j}\left(\norm{\tilde e_{n,i}^{ob}}_2^2 + \norm{\tilde e_{n,j}^{ob}}_2^2\right)\le 2\bar\gamma^{ob}\sum_{i\in\mathcal{I}^{ob}}\norm{\tilde e_{n,i}^{ob}}_2^2.
\]
Following the same steps for the other two terms, \eqref{eq:proof-mm-Wrec-init} leads to
\begin{align}
\sum_{i\in\mathcal{I}^{ob}} \mathbb{E}[W_{n+1,i}\mid\mathcal{F}_n]
&\;\le\; \frac{n}{n+1}\,W_n^{ob}\notag\\
&\quad -\;\frac{1}{n+1}\left[\frac{C_M}{1+\delta_h} - C_x\,\left(2\bar\gamma^{ob}\,L_\phi^{ob} + \bar\Delta + \bar\Theta\right)\right]\sum_{i\in\mathcal{I}^{ob}}\norm{\tilde e_{n,i}^{ob}}_2^2\notag\\
&\quad +\;\frac{C_x\,\bar\Psi}{n+1}\,U_n^{tot} \;+\; \frac{C_x\,\bar\Theta^{\mathrm{col}}}{n+1}\sum_{j\in\mathcal{I}^{in}}\norm{\tilde e_{n,j}^{in}}_2^2 \;+\; \frac{C_1\,|\mathcal{I}^{ob}|}{n(n+1)}.\label{eq:proof-mm-Wrec-agg}
\end{align}

\paragraph{Step 3. Running-mean recursion.}
For every $k$,
\[
u_{n+1,k} = \frac{n}{n+1}u_{n,k} + \frac{1}{n+1}\left(p_{n+1,k} - p_k^{NE}\right).
\]
Squaring, taking conditional expectation, and using $C_p^2 \triangleq \sup_{n,k}\mathbb{E}[(p_{n+1,k} - p_k^{NE})^2\mid\mathcal{F}_n] < \infty$ give
\[
\begin{aligned}
&\mathbb{E}[U_{n+1,k}\mid\mathcal{F}_n]\le U_{n,k} - \frac{2}{n+1}U_{n,k}\\
&\quad +\frac{2}{n+1}|u_{n,k}||\mathbb{E}[p_{n+1,k} - p_k^{NE}\mid\mathcal{F}_n]|+\frac{U_{n,k} + C_p^2}{(n+1)^2}.
\end{aligned}
\]
The cross-term bound distinguishes $k\in\mathcal{I}^{ob}$ and $k\in\mathcal{I}^{in}$. For $k\in\mathcal{I}^{ob}$, \eqref{eq:proof-mm-ob-comp} gives $|\mathbb{E}[p_{n+1,k}-p_k^{NE}\mid\mathcal{F}_n]|\le L_\phi^{ob}\,\norm{\tilde e_{n,k}^{ob}}_2$, and Young's inequality $|u_{n,k}|\cdot L_\phi^{ob}\norm{\tilde e_{n,k}^{ob}}_2 \le \tfrac{L_\phi^{ob}}{2}(U_{n,k} + \norm{\tilde e_{n,k}^{ob}}_2^2)$ yields
\begin{equation}\label{eq:proof-mm-Urec-ob}
\mathbb{E}[U_{n+1,k}\mid\mathcal{F}_n] \;\le\; U_{n,k} \;-\; \frac{2 - L_\phi^{ob}}{n+1}\,U_{n,k} \;+\; \frac{L_\phi^{ob}}{n+1}\,\norm{\tilde e_{n,k}^{ob}}_2^2 \;+\; \frac{U_{n,k} + C_p^2}{(n+1)^2}.
\end{equation}
For $k\in\mathcal{I}^{in}$, \eqref{eq:proof-mm-in-comp} gives the two-piece bound. Young's on the parameter piece gives
\[
|u_{n,k}|\cdot L_\phi^{in,\theta}\norm{\tilde e_{n,k}^{in}}_2 \le \frac{L_\phi^{in,\theta}}{2}\left(U_{n,k} + \norm{\tilde e_{n,k}^{in}}_2^2\right)
\]
and on the running-mean piece gives
\[
|u_{n,k}|\cdot \frac{\gamma_{k,\ell}}{2|\beta_k|}\,|u_{n,\ell}| \le \frac{\gamma_{k,\ell}}{2|\beta_k|}\frac{U_{n,k} + U_{n,\ell}}{2}
\]
for all $\ell\ne k$. Combining, for $k\in\mathcal{I}^{in}$,
\begin{align}
\mathbb{E}[U_{n+1,k}\mid\mathcal{F}_n]
&\;\le\; U_{n,k} \;-\; \frac{2 - L_\phi^{in,\theta} - \gamma_k/(2|\beta_k|)}{n+1}\,U_{n,k}\notag\\
&\quad +\;\frac{1}{n+1}\sum_{\ell\ne k}\frac{\gamma_{k,\ell}}{2|\beta_k|}\,U_{n,\ell} \;+\; \frac{L_\phi^{in,\theta}}{n+1}\,\norm{\tilde e_{n,k}^{in}}_2^2 \;+\; \frac{U_{n,k}+C_p^2}{(n+1)^2}.\label{eq:proof-mm-Urec-in}
\end{align}
Combining \eqref{eq:proof-mm-Urec-ob}--\eqref{eq:proof-mm-Urec-in} and applying $\bar c_{\mathrm{diag}} = \min\{2 - L_\phi^{ob}, 2 - L_\phi^{in,\theta} - \max_{j\in\mathcal{I}^{in}}\gamma_j/(2|\beta_j|)\}$ and
\[
\sum_{k\in\mathcal{I}^{in}}\sum_{\ell\ne k}\frac{\gamma_{k,\ell}}{2|\beta_k|}U_{n,\ell} = \sum_{\ell\in[N]}U_{n,\ell}\sum_{k\in\mathcal{I}^{in}\setminus\{\ell\}}\frac{\gamma_{k,\ell}}{2|\beta_k|} \le \bar D U_n^{tot},
\]
we arrive at
\begin{align}
\sum_{k=1}^N \mathbb{E}[U_{n+1,k}\mid\mathcal{F}_n]
&\;\le\; U_n^{tot} \;-\; \frac{\bar c_{\mathrm{diag}} - \bar D}{n+1}\,U_n^{tot}
\;+\; \frac{L_\phi^{ob}}{n+1}\sum_{k\in\mathcal{I}^{ob}}\norm{\tilde e_{n,k}^{ob}}_2^2\notag\\
&\quad +\;\frac{L_\phi^{in,\theta}}{n+1}\sum_{k\in\mathcal{I}^{in}}\norm{\tilde e_{n,k}^{in}}_2^2 \;+\; \frac{N\,C_2}{(n+1)^2},\label{eq:proof-mm-Urec-agg}
\end{align}
where $C_2 \triangleq \sup_{n,k} U_{n,k} + C_p^2 < \infty$ and $\bar c_{\mathrm{diag}} - \bar D > 0$ by condition~(ii).

\paragraph{Step 4. Combined Lyapunov contraction.}
For $\lambda > 0$ to be chosen, define the joint Lyapunov function $\mathcal{L}_n \triangleq W_n^{ob} + \lambda U_n^{tot}$. Combining~\eqref{eq:proof-mm-Urec-agg} and \eqref{eq:proof-mm-Wrec-agg},
\begin{equation}\label{eq:proof-mm-Lstep4-bracket}
\mathbb{E}[\mathcal{L}_{n+1}\mid\mathcal{F}_n] \;\le\; \mathcal{L}_n \;-\; \frac{1}{n+1}\,\left[\,W_n^{ob} + A(\lambda)\sum_{i\in\mathcal{I}^{ob}}\norm{\tilde e_{n,i}^{ob}}_2^2 + B(\lambda)\,U_n^{tot}\,\right] \;+\; R_n,
\end{equation}
where $A(\lambda) \triangleq \tfrac{C_M}{1+\delta_h} - C_x(2\bar\gamma^{ob}L_\phi^{ob} + \bar\Delta + \bar\Theta) - \lambda L_\phi^{ob}$, $B(\lambda) \triangleq\lambda(\bar c_{\mathrm{diag}} - \bar D) - C_x\bar\Psi$, and the residual is
\begin{equation}\label{eq:proof-mm-R}
R_n \;\triangleq\; \frac{C_1\,|\mathcal{I}^{ob}| + \lambda\,N\,C_2}{n(n+1)} \;+\; \frac{C_x\,\bar\Theta^{\mathrm{col}} + \lambda\,L_\phi^{in,\theta}}{n+1}\sum_{j\in\mathcal{I}^{in}}\norm{\tilde e_{n,j}^{in}}_2^2.
\end{equation}
$\min \{A(\lambda), B(\lambda)\} > 0$ holds if and only if $\lambda$ lies in the open interval
\begin{equation}\label{eq:proof-mm-lambda-range}
\frac{C_x\,\bar\Psi}{\bar c_{\mathrm{diag}} - \bar D} \;<\; \lambda \;<\; \frac{C_M/(1+\delta_h) - C_x\,(2\bar\gamma^{ob}\,L_\phi^{ob} + \bar\Delta + \bar\Theta)}{L_\phi^{ob}}.
\end{equation}
A feasible $\lambda$ exists if and only if the upper bound exceeds the lower bound, which after letting $\delta_h \downarrow 0$ is exactly condition~(iii). Fix any such $\lambda$. By Lemma~\ref{lem:spectralboundadaptivecovariancematrix}, $\sum_i\norm{\tilde e_{n,i}^{ob}}_2^2 \ge W_n^{ob}/C_x^2$, and \eqref{eq:proof-mm-Lstep4-bracket} gives $W_n^{ob} + A(\lambda)\sum_i\norm{\tilde e_{n,i}^{ob}}_2^2 + B(\lambda)U_n^{tot} \ge c^*(\lambda)\mathcal{L}_n$, where
\begin{equation}\label{eq:proof-mm-cstar-lambda}
c^*(\lambda) \;\triangleq\; \min\,\left\{1 + \frac{A(\lambda)}{C_x^2},\;\, \frac{B(\lambda)}{\lambda}\right\} \;=\; \min\,\left\{1 + \frac{K - \lambda\,L_\phi^{ob}}{C_x^2},\;\, \bar\kappa - \frac{C_x\,\bar\Psi}{\lambda}\right\} \;>\; 0
\end{equation}
after letting $\delta_h \downarrow 0$ in $A(\lambda)$. Substituting back into~\eqref{eq:proof-mm-Lstep4-bracket},
\begin{equation}\label{eq:proof-mm-Lcontract}
\mathbb{E}[\mathcal{L}_{n+1}\mid\mathcal{F}_n] \;\le\; \mathcal{L}_n \;-\; \frac{c^*(\lambda)}{n+1}\,\mathcal{L}_n \;+\; R_n.
\end{equation}

\paragraph{Step 5. Convergence.}
By~\eqref{eq:thm:informedolsrate-sum} of Theorem~\ref{thm:informedolsrate}, $\sum_n \tfrac{1}{n+1}\sum_j\norm{\tilde e_{n,j}^{in}}_2^2 < \infty$ a.s., which together with the $O(1/n^2)$ term in~\eqref{eq:proof-mm-R} gives $\sum_n R_n < \infty$ a.s. Applying the Robbins--Siegmund theorem~\citep{robbinsConvergenceTheoremNonnegative1971} to~\eqref{eq:proof-mm-Lcontract} yields $\mathcal{L}_n \to 0$ a.s., so $W_n^{ob}\to 0$ and $U_n^{tot}\to 0$ a.s. Lemma~\ref{lem:spectralboundadaptivecovariancematrix} then gives $\norm{\tilde e_{n,i}^{ob}}_2\to 0$ a.s.\ for every $i\in\mathcal{I}^{ob}$, hence $\hat\theta_{n,i}^{ob}\to\theta_i^{*,ob}$ a.s., and continuity of $\phi^{ob}$ gives $\tilde p_{n,i}\to p_i^{NE}$ a.s. Theorem~\ref{thm:informedolsrate} gives $\hat\theta_{n,j}^{in}\to\theta_j^{*,in}$ a.s.\ for $j\in\mathcal{I}^{in}$. Finally, $U_n^{tot}\to 0$ gives $m_{n,k}\to p_k^{NE}$ a.s.\ for every $k$, so by continuity of $\phi_j^{in}(\cdot;\cdot)$,
\[
\tilde p_{n,j} = \phi_j^{in}(\hat\theta_{n,j}^{in}; \mathbf m_{n,-j}) \to\phi_j^{in}(\theta_j^{*,in}; \mathbf p_{-j}^{NE}) = p_j^{NE} \quad\text{a.s. for every } j\in\mathcal{I}^{in}.
\]
Taking expectation in~\eqref{eq:proof-mm-Lcontract} gives
\[
\mathbb{E}[\mathcal{L}_{n+1}] \le \left(1 - \frac{c^*(\lambda)}{n+1}\right)\mathbb{E}[\mathcal{L}_n] + \mathbb{E}[R_n].
\]
By Theorem~\ref{thm:informedolsrate}, $\sum_j\mathbb{E}\norm{\tilde e_{n,j}^{in}}_2^2 = O(n^{\eta_{\max}-1}\log n)$, so $\mathbb{E}[R_n] = O(n^{\eta_{\max}-2}\log n)$. Applying Lemma~\ref{lem:nonhomogeneouslinearrecursion}(b) with $a = c^*(\lambda)$ and $r = 2 - \eta_{\max}$ yields the desired rates on $\mathbb{E}[\mathcal{L}_n]$, and optimizing over $\lambda$ in the feasible interval~\eqref{eq:proof-mm-lambda-range} replaces $c^*(\lambda)$ by $c^* \triangleq \sup_\lambda c^*(\lambda)$. The greedy price MSE inherits both bottlenecks since $\tilde p_{n,j} = \phi_j^{in}(\hat\theta_{n,j}^{in}; \mathbf m_{n,-j})$ inherits the $\mathcal{L}_n$ rate through $\mathbf m$ and the informed-OLS rate through $\hat\theta$, with overall rate being $\min\{c^*, 1-\eta_{\max}\}$.

\smallskip
\noindent\emph{Step 6. Regret.} By Proposition~\ref{prop:dynamicbenchmark}, $\mathrm{\mathbf{RegLe}}(T) \asymp \sum_{n\le T}\mathbb{E}\|\tilde{\mathbf p}_n - \mathbf p^{NE}\|_2^2$. Summing the aggregate greedy-price MSE from Step~5 over $n\le T$ gives the stated learning-regret rates $O(T^{\eta_{\max}}\log T)$, $O(T^{\eta_{\max}}(\log T)^2)$, and $O(T^{1-c^*})$ for $c^*>1-\eta_{\max}$, $c^*=1-\eta_{\max}$, and $c^*<1-\eta_{\max}$, respectively, each sublinear. The exploration tax is deterministic: oblivious sellers use constant variance $\nu_i^2>0$, so $\Delta_i^{\mathrm{ex}}(T)=\Theta(T)$, and an informed seller $j$ with a persistently exploring neighbor $i$ ($\gamma_{j,i}\ne0$) inherits that neighbor's constant variance through $\Delta_j^{\mathrm{ex}}(T)=-\beta_j\sum_{n\le T}(\nu_{n,j}^2+\sum_{k\ne j}(\gamma_{j,k}/2\beta_j)^2\nu_{n,k}^2)=\Theta(T)$. Hence the total regret $\Delta_k(T)=\Theta(T)$ for every seller with a persistently exploring neighbor, even though its greedy price converges to Nash.

\smallskip
\noindent\emph{Step 7. Surplus capture.} The convergence $\tilde p_{n,k}\to p_k^{NE}$ a.s.\ established above lets Lemma~\ref{lem:asymp-revenue} apply. Oblivious sellers use constant variance, so $\bar\nu_{T,i}^2=\nu_i^2$ and $\bar R_{T,i}\to\Pi_i^{NE}-|\beta_i|\nu_i^2$ a.s.; dividing by $\Pi_i^{C}-\Pi_i^{NE}>0$ gives $S_i=-|\beta_i|\nu_i^2/(\Pi_i^{C}-\Pi_i^{NE})<0$ for $i\in\mathcal{I}^{ob}$. Informed sellers' own exploration decays ($\nu_{n,j}^2\to0$), so $\bar\nu_{T,j}^2\to0$ and $\bar R_{T,j}\to\Pi_j^{NE}$ a.s., giving $S_j=0$ for $j\in\mathcal{I}^{in}$.

\subsection{Proof of Proposition~\ref{prop:stackelberg}}\label{proof:prop:stackelberg}

Since $\beta_1\beta_2>\gamma_{1,2}\gamma_{2,1}$, $\beta_1+\tfrac{\gamma_{1,2}\gamma_{2,1}}{-2\beta_2}<0$ and the denominator in~\eqref{eq:stackelberg-prices-app} is positive. The interior assumption renders the parameter-set projections asymptotically inactive. Under perfect prediction the informed seller observes $\hat p_{n+1,1}=p_{n+1,1}$, so by \eqref{eq:greedy-prices} its posted price is
\[
p_{n+1,2}=\frac{\hat\alpha_{n,2}+\hat\gamma_{n,2,1}p_{n+1,1}}{-2\hat\beta_{n,2}}+z_{n+1,2}.
\]
By Theorem~\ref{thm:informedolsrate}, $\hat\theta_{n,2}^{in}\to\theta_2^{*,in}$ almost surely. Since $\phi_i^{in}(\theta^{in}; \hat{\mathbf{p}}_{n+1,-i})$ is Lipschitz in $\theta^{in}$, this implies that $p_{n+1,2}=\mathrm{BR}_2(p_{n+1,1})+z_{n+1,2}+o(1)$, so that
\begin{equation}\label{eq:stack-mean-track}
m_{n,2}=\mathrm{BR}_2(m_{n,1})+o(1)
\end{equation}
and
\[
p_{m,2}-m_{n,2}=\frac{\gamma_{2,1}}{-2\beta_2}(p_{m,1}-m_{n,1})+(z_{m,2}-\bar z_{n,2})+o(1).
\]
Substituting into $r_{n,1\leftarrow 2}$ gives
\[
r_{n,1\leftarrow 2}=\frac{\gamma_{2,1}}{-2\beta_2}+\frac{\sum_{m\le n}(p_{m,1}-m_{n,1})z_{m,2}}{J_{n,1}}+o(1).
\]
Because $z_{m,2}$ is independent of $\mathcal{F}_{m-1}$ and of $z_{m,1}$, $\sum_{m\le n}(p_{m,1}-m_{n,1})z_{m,2}$ is a martingale with bounded martingale differences. Since $J_{n,1}=\Theta(n)$, the martingale strong law forces the ratio to vanish, so $r_{n,1\leftarrow 2}\to \tfrac{\gamma_{2,1}}{-2\beta_2}$. Consequently, by \eqref{eq:oblivious-ls-coeffs} and Theorem~\ref{thm:impactdemandnoise}, $\hat b_{n,1}\to\beta_1+\tfrac{\gamma_{1,2}\gamma_{2,1}}{-2\beta_2}<0$. Substituting \eqref{eq:stack-mean-track} gives $m_{n,2}-r_{n,1\leftarrow 2}\,m_{n,1}=\tfrac{\alpha_2}{-2\beta_2}+o(1)$ so that
\[
\hat a_{n,1}\to\alpha_1+\frac{\gamma_{1,2}\alpha_2}{-2\beta_2}.
\]
The leader's greedy price is thus $\tilde p_{n+1,1}=F_1(m_n)+o(1)$, where
\[
F_1(m)=\frac{\alpha_1+\gamma_{1,2}(m_2-\frac{\gamma_{2,1}}{-2\beta_2}m_1)}{D}, \qquad D\triangleq -2\left(\beta_1+\frac{\gamma_{1,2}\gamma_{2,1}}{-2\beta_2}\right)>0.
\]
Under perfect prediction,
\[
\mathbb{E}[p_{n+1,2}\mid\mathcal{F}_n]=F_2(m_n)+o(1) \quad\text{with } F_2=\mathrm{BR}_2\circ F_1.
\]
Collecting the mean-zero exploration noise into $\zeta_{n+1}= (z_{n+1,1}, \tfrac{\gamma_{2,1}}{-2\beta_2}z_{n+1,1}+z_{n+1,2})^\top$ and the $o(1)$ terms into $\rho_{n+1}$ gives
\[
m_{n+1}=m_n+\frac1{n+1}[F(m_n)-m_n+\zeta_{n+1}+\rho_{n+1}].
\]
Here $F$ is affine with a unique fixed point $m^{*}=(p_1^{S},p_2^{S})$ and constant Jacobian 
$$
A=\frac{\gamma_{1,2}}{D}
\left(\begin{matrix}
\gamma_{2,1} / 2\beta_2 & 1\\
-(\gamma_{2,1} / 2\beta_2)^{2} & -\gamma_{2,1} / 2\beta_2
\end{matrix}\right).
$$
So, $A-I$ has a double eigenvalue $-1$ and thus is Hurwitz. Let $u_n\triangleq m_n-m^{*}$. Since $F(m)-m^{*}=A\,(m-m^{*})$, we have
\[
u_{n+1}=u_n+\frac1{n+1}[(A-I)u_n+\zeta_{n+1}+\rho_{n+1}].
\]
Since $A-I$ is Hurwitz, a standard stochastic approximation argument (e.g., \S 2 of \citet{borkarStochasticApproximationDynamical2023}) yields $m_n\to m^{*}=(p_1^{S},p_2^{S})$. As a result, $\tilde p_{n,i}=F_i(m_{n-1})+o(1)\to p_i^{S}$. 

\paragraph*{Surplus capture of the follower.}
Since $\tilde p_{n,i}\to p_i^{S}$ and $p_{n,1}=p_1^{S}+z_{n,1}+o(1)$ and by perfect prediction, $p_{n,2}=\mathrm{BR}_2(p_{n,1})+z_{n,2}+o(1)$. Let $W_2(p_1)\triangleq\max_{p_2}R_2(p_2,p_1)$ be the follower's value function. Because $R_2(p_2,p_1)=p_2(\alpha_2+\beta_2 p_2+\gamma_{2,1}p_1)$ is a concave quadratic in $p_2$, for any deviation $\delta$, we have
\begin{equation}\label{eq:stack-square}
R_2\left(\mathrm{BR}_2(p_1)+\delta,\,p_1\right)=W_2(p_1)+\beta_2\,\delta^{2},
\qquad
W_2(p_1)=\frac{(\alpha_2+\gamma_{2,1}p_1)^{2}}{-4\beta_2}.
\end{equation}
The realized per-period revenue is $R_{n,2}=p_{n,2}d_{n,2}=R_2(\mathbf p_n)+p_{n,2}\varepsilon_{n,2}$. Applying~\eqref{eq:stack-square} with $\delta=z_{n,2}+o(1)$ gives
\[
R_{n,2}=W_2(p_{n,1})+\beta_2\,z_{n,2}^{2}+p_{n,2}\varepsilon_{n,2}+o(1).
\]
Averaging over $n\le T$,
\[
\bar R_{T,2}\triangleq\frac1T\sum_{n\le T}R_{n,2}=\frac1T\sum_{n\le T}W_2(p_{n,1})+\beta_2\frac1T\sum_{n\le T}z_{n,2}^{2} +\frac1T\sum_{n\le T}p_{n,2}\varepsilon_{n,2}+o(1).
\]
On the right-hand side, the second average vanishes because $\nu_{n,2}^{2}\to0$, and the last average is a martingale with bounded increments, hence $\to0$ by the strong law. For the first average, expanding the quadratic $W_2$ about $p_1^{S}$ using its constant curvature $W_2''=\tfrac{\gamma_{2,1}^{2}}{-2\beta_2}$ gives
\[
W_2(p_{n,1})=W_2(p_1^{S})+W_2'(p_1^{S})(p_{n,1}-p_1^{S})+\frac{\gamma_{2,1}^{2}}{-4\beta_2}(p_{n,1}-p_1^{S})^{2}.
\]
Since $p_{n,1}-p_1^{S}=z_{n,1}+o(1)$, $\frac1T\sum_{n\le T} z_{n,1}\to0$, and $\frac1T\sum_{n\le T} z_{n,1}^{2}\to\nu_1^{2}$, the first average converges to $W_2(p_1^{S})+\tfrac{\gamma_{2,1}^{2}}{-4\beta_2}\nu_1^{2}$. Combining,
\[
\bar R_{T,2}\to W_2(p_1^{S})+\frac{\gamma_{2,1}^{2}}{-4\beta_2}\nu_1^{2}\ge W_2(p_1^{S})=R_2(p_2^{S},p_1^{S}) \quad\text{a.s.}
\]
It remains to show $R_2(p_2^{S},p_1^{S})>\Pi_2^{NE}$. Since $\Pi_2^{NE}=R_2(\mathrm{BR}_2(p_1^{NE}),p_1^{NE})=W_2(p_1^{NE})$ and $W_2$ is strictly increasing on positive prices, it suffices that $p_1^{S}>p_1^{NE}$. The leader's Nash price solves out to be
\[
p_1^{NE}=\frac{-2\beta_2\alpha_1+\gamma_{1,2}\alpha_2}{4\beta_1\beta_2-\gamma_{1,2}\gamma_{2,1}},
\]
which shares the numerator of $p_1^{S}$ in~\eqref{eq:stackelberg-prices-app} but has a larger denominator than that of $p_1^{S}$. Therefore, $p_1^{S} > p_1^{NE}$ and $\bar R_{T,2}\to W_2(p_1^{S})+\frac{\gamma_{2,1}^{2}}{-4\beta_2}\nu_1^{2}>\Pi_2^{NE}$ a.s., and the result follows.

\subsection{Proof of Theorem~\ref{thm:informedolsrate}}\label{proof:thm:informedolsrate}

\begin{lemma}[Growth of the design]\label{lem:designmatrixgrowthcondition}
Fix an integer $d \ge 1$ and let $(\mathcal{F}_n)_{n \ge 0}$ be a filtration. For each $n \ge 1$, define the random vector $x_n = (1, p_{n,1}, ..., p_{n,d})^\top \in \mathbb{R}^{d+1}$ and the design matrix $\mathcal{J}_n = \sum_{m=1}^n x_m x_m^\top$. Assume there exist constants $0 < l < u < \infty$ such that $p_{n,i} \in [l,u]$ a.s. for all $n$ and $i$, and that $p_{n,i} = \chi_{n,i} + z_{n,i}$ for all $i \in [d]$, where $\chi_{n,i} \in \mathcal{F}_{n-1}$ are uniformly bounded, and $z_{n,i}$ are independent across all $n$ and $i$, mean-zero, and uniformly bounded. Let $\nu_{n,i}^2 \triangleq \mathrm{Var}(z_{n,i}) = c_i n^{-\eta_i}$ with $c_i > 0$ and $\eta_i \in [0,1)$. Write $\eta_{\min} = \min_{1 \le i \le d} \eta_i$ and $\eta_{\max} = \max_{1 \le i \le d} \eta_i$. Then, there exist constants $C_1, C_2 > 0$ and an $n_0$ such that for all $n \ge n_0$, 
$$
\mathbb{P} \left(\lambda_{\min} (\mathcal{J}_n) \le C_1 n^{1 - \eta_{\max}}\right) \le (d+1) \exp \left(-C_2 n^{1 + \eta_{\min} - 2 \eta_{\max}}\right).
$$
In particular, the right-hand side decreases to zero exponentially fast if $1 + \eta_{\min} > 2 \eta_{\max}$.
\end{lemma}
\begin{proof}[Proof of Lemma~\ref{lem:designmatrixgrowthcondition}]
To set up, let $\mathbb{E}_{n-1}$ denote the conditional expectation given $\mathcal{F}_{n-1}$. Denote $M_n = \sum_{m=1}^n \mathbb{E}_{m-1}[x_m x_m^\top]$ and $Y_m = x_m x_m^\top - \mathbb{E}_{m-1}[x_m x_m^\top]$. Then, $(Y_m, \mathcal{F}_n)$ is a self-adjoint matrix martingale difference sequence and $\mathcal{J}_n = M_n + \sum_{m=1}^n Y_m$. For any $t > 0$,
\begin{equation}\label{eq:designmatrixgrowthconditiondecomposition}
\left\{\lambda_{\min}(\mathcal{J}_n) \le \lambda_{\min}(M_n) - t \right\} \subseteq \left\{\lambda_{\max}(M_n - \mathcal{J}_n) \ge t\right\} = \left\{\lambda_{\max}\left(-\sum_{m=1}^n Y_m\right) \ge t\right\}.
\end{equation}
So, it suffices to upper bound $\mathbb{P}\left(\lambda_{\max}\left(-\sum_{m=1}^n Y_m\right) \ge t\right)$. To set the stage, we need a uniform bound $R$ on $\norm{Y}_{op}$, a lower bound on $\lambda_{\min}(M_n)$, and an upper bound $\sigma_n^2$ on $\norm{\sum_{m=1}^n \mathbb{E}_{m-1}[Y_m^2]}_{op}$. Since $p_{n,i}$ are uniformly bounded, write $C_x^2$ as an upper bound on $\norm{x_m}_2^2$, and by Jensen's inequality we can take $R = 2 C_x^2$ for the first item. 

Write $\mu_m \triangleq \mathbb{E}_{m-1}[x_m] = (1, \chi_{m,1}, ..., \chi_{m,d})^\top$ and $\xi_m \triangleq x_m - \mu_m = (0, z_{m,1}, ..., z_{m,d})^\top$. Then, by independence across coordinates, $S_i(n) = \sum_{m=1}^n \nu_{m,i}^2$ and
\[
M_n = \sum_{m=1}^n \mu_m \mu_m^\top + \mathrm{diag}(1, S_1(n), ..., S_d(n)).
\]
Let $\bar{\mu}_n = \tfrac{1}{n} \sum_{m=1}^n \mu_m$. Since $\sum_{m=1}^n \mu_m \mu_m^\top = \sum_{m=1}^{n} (\mu_m - \bar{\mu}_n)(\mu_m - \bar{\mu}_n)^\top + n \bar{\mu}_n \bar{\mu}_n^\top \succeq n \bar{\mu}_n \bar{\mu}_n^\top$, we have
\[
M_n \succeq K_n \triangleq n \bar{\mu}_n \bar{\mu}_n^\top + \mathrm{diag}(0, S_1(n), ..., S_d(n))
\]
and $\lambda_{\min}(M_n) \ge \lambda_{\min}(K_n)$. Write $\bar{\mu}_n = (1, b_n)^\top$ with $b_n \in \mathbb{R}^d$ and $S(n) = \mathrm{diag}(S_1(n), ..., S_d(n))$. Then, by Lemma~\ref{lem:rankonediagonalmatrixmineigenvalue} we have
\[
\lambda_{\min}(M_n) \ge \frac{1}{2} \min \left\{\frac{n}{1 + n b_n^\top S(n)^{-1} b_n},\, \min_{1 \le i \le d} S_i(n)\right\}.
\]
Since $\nu_{m,i}^2 = c_i m^{-\eta_i}$, we have $\min_{1 \le i \le d} S_i(n) \gtrsim n^{1 - \eta_{\max}}$. Since $b_n$ is uniformly bounded, we have $b_n^\top S(n)^{-1} b_n \lesssim n^{-(1 - \eta_{\max})}$. Thus, $\frac{n}{1 + n b_n^\top S(n)^{-1} b_n} \gtrsim n^{1 - \eta_{\max}}$, and together we have $\lambda_{\min}(M_n) \gtrsim n^{1 - \eta_{\max}}$ for the second item.

We claim that there exists some constant $C_V > 0$ such that $\mathbb{E}_{m-1}\norm{Y_m^2}_{op} \le C_V \sum_{i=1}^{d} \nu_{m,i}^2$ for all $m$. Given this, we have $\sigma_n^2 \le \sum_{m=1}^{n} \mathbb{E}_{m-1}\norm{Y_m^2}_{op} \le C_V \sum_{m=1}^{n} \sum_{i=1}^{d} \nu_{m,i}^2 \lesssim n^{1 - \eta_{\min}}$. To show the claim, write $Y_m = \mu_m \xi_m^\top + \xi_m \mu_m^\top + (\xi_m \xi_m^\top - \mathbb{E}_{m-1}[\xi_m \xi_m^\top])$. Since $\mu_m$ is bounded, $\mu_m \xi_m^\top + \xi_m \mu_m^\top \lesssim \norm{\xi_m}_2$. The last two terms can be bounded by $\norm{\xi_m}_2^2 + \mathbb{E}_{m-1} [\norm{\xi_m}_2^2]$. Squaring and taking conditional expectation give
\[
\mathbb{E}_{m-1}\norm{Y_m}_{op}^2 \lesssim \mathbb{E}_{m-1}[\norm{\xi_m}_2^2]+ \mathbb{E}_{m-1}[\norm{\xi_m}_2^3] + \mathbb{E}_{m-1}[\norm{\xi_m}_2^4].
\]
Since $\xi_m$ are bounded, the higher moment terms can be controlled by the second moment term, and so $\mathbb{E}_{m-1}\norm{Y_m}_{op}^2 \lesssim \sum_{i=1}^{d} \nu_{m,i}^2$, proving the claim.

By the Matrix Freedman inequality (Theorem~1.2 of \citet{troppFreedmansInequalityMatrix2011}), for any $t \ge 0$,
\[
\mathbb{P}\left(\lambda_{\max}\left(-\sum_{m=1}^n Y_m\right) \ge t\right) \le (d+1) \exp\left(-\frac{t^2/2}{\sigma_n^2 + R t / 3}\right).
\]
By \eqref{eq:designmatrixgrowthconditiondecomposition}, $\mathbb{P}\left(\lambda_{\min}(\mathcal{J}_n) \le \lambda_{\min}(M_n) - t\right) \le (d+1) \exp\left(-\tfrac{t^2/2}{\sigma_n^2 + R t / 3}\right)$. Since $\lambda_{\min}(M_n) \gtrsim n^{1 - \eta_{\max}}$ and $\sigma_n^2 \lesssim n^{1 - \eta_{\min}}$, choose $t = \tfrac{1}{2} n^{1 - \eta_{\max}}$ to get
\[
\mathbb{P}\left(\lambda_{\min}(\mathcal{J}_n) \le \frac{1}{2} n^{1 - \eta_{\max}}\right) \le (d+1) \exp\left(-\frac{n^{2(1 - \eta_{\max})}/8}{C_1 n^{1 - \eta_{\min}} + R n^{1 - \eta_{\max}} / 6}\right)
\]
for some constant $C_1 > 0$ for all large $n$. Since $\eta_{\min} \le \eta_{\max}$, we have $\tfrac{n^{2(1 - \eta_{\max})}/8}{C_1 n^{1 - \eta_{\min}} + R n^{1 - \eta_{\max}} / 6} \gtrsim n^{1 + \eta_{\min} - 2 \eta_{\max}}$. The result follows by redefining the constants.
\end{proof}

We drop the subscript $i$ and the superscript $*$ for notational simplicity. Let $\mathcal{J}_n = \sum_{m=1}^n x_{m,i}^{in} (x_{m,i}^{in})^\top$ be the empirical Fisher information matrix up to time $n$. The compactness of the projection set $\Theta^{in}_i$ from Section~\ref{sec:learning-dynamics} provides a finite a.s.\ bound $D_\Theta \triangleq \operatorname{diam}(\Theta^{in}_i)^2$ on $\norm{\hat{\mathbf{\theta}}_n - \mathbf{\theta}}_2^2$. By Lemma~\ref{lem:designmatrixgrowthcondition}, there exist constants $c_1, c_2 > 0$ and $N_0 \in \mathbb{N}$ such that, writing $E_n \triangleq \{\lambda_{\min}(\mathcal{J}_n) > c_1 n^{1 - \eta_{\max}}\}$, for all $n \ge N_0$,
\[
\mathbb{E}\norm{\hat{\mathbf{\theta}}_n - \mathbf{\theta}}_2^2 \le \int_{0}^{\infty}  \mathbb{P}\left(\norm{\hat{\mathbf{\theta}}_n - \mathbf{\theta}}_2^2 > x,E_n\right) dx + D_\Theta(N+1) \exp\left(-c_2 n^{1 + \eta_{\min} - 2 \eta_{\max}}\right).
\]
By Lemma 3 of \citet{keskinDynamicPricingUnknown2014}, we have, for some $k, \rho > 0$ and all $x > 0$ and $n$ large enough,
\[
\mathbb{P}\left(\norm{\hat{\mathbf{\theta}}_n - \mathbf{\theta}}_2^2 > x, E_n\right)\le kn \exp \left(- c_1 \rho \min\{x, \sqrt{x}\} n ^{1 - \eta_{\max}}\right).
\]
Set $b_n \triangleq c_1 \rho n^{1 - \eta_{\max}}$ and $x_n^\circ \triangleq 2 \log n / b_n$. Since $\eta_{\max} < 1$, $x_n^\circ < 1$ for all large $n$, and $x_n^\circ = (2/(c_1\rho))n^{\eta_{\max} - 1} \log n = O(n^{\eta_{\max} - 1}\log n)$. Splitting the integral at $x_n^\circ$ and at $1$, $\int_{0}^{\infty} \mathbb{P}(\norm{\hat{\mathbf{\theta}}_n - \mathbf{\theta}}_2^2 > x, E_n) dx \le x_n^\circ + \int_{x_n^\circ}^{1}  kn \exp (- b_n \min\{x, \sqrt{x}\}) dx + \int_{1}^{\infty} kn \exp (- b_n \min\{x, \sqrt{x}\}) dx$. The leading term $x_n^\circ$ is already of order $n^{\eta_{\max}-1}\log n$. On $[x_n^\circ, 1]$ we have $\min\{x, \sqrt{x}\} = x$, and using $\int_a^1 e^{-bx}dx = \tfrac{1}{b}(e^{-b a} - e^{-b})$ with $a = x_n^\circ$ and $b = b_n$ gives $\int_{x_n^\circ}^{1} kn e^{-b_n x}dx \le \tfrac{kn}{b_n} e^{-b_n x_n^\circ} = o(n^{\eta_{\max}-1}\log n)$. On $[1, \infty)$ we have $\min\{x, \sqrt{x}\} = \sqrt{x}$, and using $\int_1^{\infty} e^{-a \sqrt{x}}dx = \tfrac{2(a+1)}{a^2}e^{-a}$ with $a = b_n$ gives $\int_{1}^{\infty} kn e^{-b_n \sqrt{x}}dx = \tfrac{2kn(b_n + 1)}{b_n^2} e^{-b_n}$, which is exponentially small in $n$. The term $D_\Theta\,(N+1) \exp\left(-c_2 n^{1 + \eta_{\min} - 2 \eta_{\max}}\right)$ is also exponentially small under condition~\eqref{eq:conditioninformedsellersmixed}. Therefore, 
$$
\mathbb{E}\norm{\hat{\mathbf{\theta}}_n - \mathbf{\theta}}_2^2 = O\left(n^{\eta_{\max} - 1}\log n\right),
$$
and the $L^2$ rate follows. Fix $\varepsilon \in (0, 1 - \eta_{\max})$ and set $x_n \triangleq n^{\eta_{\max} - 1 + \varepsilon}$; since the exponent is negative, $x_n \to 0$ and $x_n \le 1$ eventually. Combining the design-matrix tail $\mathbb P(E_n^c) \le (N+1)\exp(-c_2\, n^{1+\eta_{\min}-2\eta_{\max}})$ with the Keskin--Zeevi bound at $x = x_n$ (so that $\min\{x_n, \sqrt{x_n}\} = x_n$ and $x_n \cdot n^{1-\eta_{\max}} = n^\varepsilon$),
\[
\mathbb P(\norm{\hat{\mathbf\theta}_n - \mathbf\theta}_2^2 > x_n)\le kn \exp \left(-c_1\rho\, n^\varepsilon\right) + (N+1)\exp \left(-c_2 n^{1+\eta_{\min}-2\eta_{\max}}\right).
\]
Both terms are summable over $n$. Applying Borel--Cantelli, $\norm{\hat{\mathbf\theta}_n - \mathbf\theta}_2^2 \le n^{\eta_{\max} - 1 + \varepsilon}$ eventually a.s. Summing over the finite set $i \in \mathcal{I}^{in}$ and taking a countable intersection over a sequence $\varepsilon_k \downarrow 0$ in $(0, 1 - \eta_{\max})$ upgrades this to~\eqref{eq:thm:informedolsrate-as}. The a.s.\ convergence $\hat\theta_{n,i}^{in} \to \theta_i^{*,\,in}$ is immediate since the exponent $\eta_{\max} - 1 + \varepsilon < 0$ for every admissible $\varepsilon$. For~\eqref{eq:thm:informedolsrate-sum}, pick $\varepsilon = (1 - \eta_{\max})/2 \in (0, 1 - \eta_{\max})$. By~\eqref{eq:thm:informedolsrate-as}, eventually a.s.\ $\sum_i \norm{\hat\theta_{n,i}^{in} - \theta_i^{*,\,in}}_2^2 \lesssim n^{\eta_{\max} - 1 + \varepsilon}$, so
\[
\frac{1}{n+1}\sum_{i \in \mathcal{I}^{in}} \norm{\hat\theta_{n,i}^{in} - \theta_i^{*,in}}_2^2 \lesssim n^{-1 - \varepsilon},
\]
which is summable.

\subsection{Technical Lemmas}
\begin{lemma}[Empirical dispersion]\label{lem:Jn-lower}
Let $(\mathcal F_n)_{n\ge 0}$ be a filtration. Suppose $(x_n)_{n\ge 1}$ satisfies $x_{n} = \mu_n + z_n$, where $\mu_n \in \mathcal F_{n-1}$ and $(z_n)_{n\ge 1}$ are independent across $n$, satisfy $z_n \perp \mathcal F_{n-1}$, $\mathbb E[z_n]=0$, and have finite second moments $\sigma_n^2 \triangleq \mathrm{Var}(z_n)\in (0,\infty)$. Define $\bar x_n \triangleq \tfrac1n\sum_{m=1}^n x_m$ and the empirical dispersion $J_n \triangleq \sum_{m=1}^n (x_m-\bar x_n)^2$. Then, for every $n\ge 1$, 
$$
\mathbb E[J_n] \ge (1-\tfrac1n)\sum_{m=1}^n \sigma_m^2.
$$
In addition, suppose $\mu_n$ and $z_n$ are a.s. uniformly bounded. Then:
\begin{enumerate}[label=(\roman*)]
\item If $(z_n)$ are i.i.d.\ with $\mathrm{Var}(z_n)=\nu^2>0$, then $\liminf_{n\to\infty}\tfrac{J_n}{n} \ge \nu^2$ and $J_n=\Theta(n)$ a.s.
\item If $(z_n)$ are independent with $\mathbb E[z_n]=0$ and there exist constants $\underline{\nu},\bar{\nu}>0$ such that $\underline{\nu}^2 n^{-c} \le \sigma_n^2 \le \bar{\nu}^2 n^{-c}$ for all $n$ for some $c\in(0,1)$, then $\liminf_{n\to\infty}\tfrac{J_n}{n^{1-c}} \ge \tfrac{\underline{\nu}^2}{1-c}$ a.s.
\end{enumerate}
\end{lemma}

\begin{proof}[Proof of Lemma~\ref{lem:Jn-lower}]
Expanding $x_m^2$ and taking expectations,
\[
\mathbb E\left[\sum_{m=1}^n x_m^2\right] =\sum_{m=1}^n \mathbb E[\mu_m^2] + 2\sum_{m=1}^n \mathbb E[\mu_m z_m] + \sum_{m=1}^n \mathbb E[z_m^2].
\]
Since $\mathbb E[\mu_m z_m] =\mathbb E[\mu_m \mathbb E[z_m\mid \mathcal F_{m-1}]]=0$,
\[
\mathbb E\left[\sum_{m=1}^n x_m^2\right] =\sum_{m=1}^n \mathbb E[\mu_m^2] + \sum_{m=1}^n \sigma_m^2, \qquad \mathbb E\left[\left(\sum_{m=1}^n x_m\right)^2\right] =\mathbb E\left[\left(\sum_{m=1}^n \mu_m\right)^2\right] + \sum_{m=1}^n \sigma_m^2.
\]
Noting $J_n = \sum_{m=1}^n x_m^2 - n\bar x_n^2 = \sum_{m=1}^n x_m^2 - \tfrac1n(\sum_{m=1}^n x_m)^2$, we have
\[
\mathbb E[J_n] = \sum_{m=1}^n \mathbb E[\mu_m^2] - \frac1n \mathbb E\left[\left(\sum_{m=1}^n \mu_m\right)^2\right]+\left(1-\frac1n\right)\sum_{m=1}^n \sigma_m^2,
\]
where the first two terms together are nonnegative by Cauchy--Schwarz. 

We first tackle (i). Expand $J_n = J_n^\mu + 2C_n + J_n^z$, where
\[
\begin{aligned}
J_n^\mu &= \sum_{m=1}^n \mu_m^2-\frac1n\left(\sum_{m=1}^n \mu_m\right)^2 \ge 0,\\
C_n &= \sum_{m=1}^n \mu_m z_m-\frac1n\left(\sum_{m=1}^n \mu_m\right)\left(\sum_{m=1}^n z_m\right),\\
J_n^z &= \sum_{m=1}^n z_m^2-\frac1n\left(\sum_{m=1}^n z_m\right)^2.
\end{aligned}
\]
Since $\sum_{m=1}^n \mu_m z_m$ is a martingale with bounded increments, by the strong law $\tfrac1n\sum_{m=1}^n \mu_m z_m \to 0$ a.s. The second term of $C_n$ also vanishes trivially when divided by $n$. Thus,
\[
\liminf_{n\to\infty}\frac{J_n}{n} \ge \lim_{n\to\infty}\frac{J_n^z}{n} + 2\lim_{n\to\infty}\frac{C_n}{n}= \nu^2 \quad\text{a.s.}
\]
$J_n=O(n)$ follows from boundedness.

Now, consider (ii). Following the same decomposition of $J_n$, we again have $J_n\ge J_n^z + 2C_n$. Let $S_n \triangleq \sum_{m=1}^n z_m$, $Q_n \triangleq \sum_{m=1}^n z_m^2$, and $V_n \triangleq \sum_{m=1}^n \sigma_m^2$. Then, $J_n^z = Q_n - \tfrac{1}{n}S_n^2$ and $\tfrac{J_n^z}{V_n} = \tfrac{Q_n}{V_n} - \tfrac{S_n^2/n}{V_n}$. Define $w_m \triangleq z_m^2 - \sigma_m^2$, which is bounded and independent, and satisfies $\mathbb E[w_m]=0$ and $\mathrm{Var}(w_m) \le K\sigma_m^2$ for some constant $K$. Since the growth condition implies $V_n\asymp n^{1-c}$, we have
\[
\sum_{m=1}^\infty \frac{\mathrm{Var}(w_m)}{V_m^2}\lesssim \sum_{m=1}^\infty m^{-(2-c)} <\infty,
\]
and $\sum_{m=1}^\infty w_m/V_m$ converges almost surely by Kolmogorov's two-series theorem. Since $V_m \rightarrow \infty$, Kronecker's lemma yields $\tfrac{Q_n}{V_n} \to 1$ a.s. Similar arguments of Kolmogorov's two-series theorem and Kronecker's lemma show that $\tfrac{S_n}{V_n} \to 0$, $\tfrac{J_n^z}{V_n} \to 1$, and $\tfrac{C_n}{V_n} \to 0$, all a.s. Therefore, $\liminf_{n\to\infty}\tfrac{J_n}{V_n} \ge 1$ a.s. The result follows since $V_n \asymp n^{1-c}$.
\end{proof}

\begin{lemma}[Conditional second-moment]\label{lem:yy-lower}
Let $(\mathcal F_n)_{n\ge 0}$ be a filtration and define the feature vector $y_{n, i} \triangleq (1,x_{n,i})^\top \in \mathbb R^{2}$. Assume that for each $i\in\{1,\dots,N\}$, $x_{n,i} = \mu_{n,i} + z_{n,i}$ with $\mu_{n,i}\in\mathcal F_{n-1}$, and, conditional on $\mathcal F_{n-1}$, $\{z_{n,i}\}_{i=1}^N$ are independent with $\mathbb E[z_{n,i}\mid \mathcal F_{n-1}] = 0$ and $\mathrm{Var}(z_{n,i}\mid \mathcal F_{n-1}) = \nu_i^2>0$. Assume further that there exists $B>0$ such that $|\mu_{n,i}|\le B$ a.s. for all $n$ and $i$. Let $\nu_* \triangleq \min_{1\le i\le N}\nu_i>0$. Then, for all $n\ge 1$, 
$$
\mathbb E[y_{n,i} y_{n,i}^\top\mid \mathcal F_{n-1}] \succeq C_M I_{2}, \qquad i \in [N],
$$ 
where $C_M \triangleq \tfrac12 \min\{\tfrac{\nu_*^2}{\nu_*^2 + B^2}, \nu_*^2\} > 0$.
\end{lemma}

\begin{proof}[Proof of Lemma~\ref{lem:yy-lower}]
Fix $n\ge 1$ and $i$. Since $\mathbb E[x_{n,i}\mid\mathcal F_{n-1}] = \mu_{n,i}$ and $\mathbb E[x_{n,i}^2\mid\mathcal F_{n-1}] = \mu_{n,i}^2 + \nu_i^2$,
\[
\mathbb E[y_{n,i} y_{n,i}^\top\mid\mathcal F_{n-1}] = (1,\mu_{n,i})^\top(1,\mu_{n,i}) + \mathrm{diag}(0, \nu_i^2).
\]
Applying Lemma~\ref{lem:rankonediagonalmatrixmineigenvalue} with $n=1$, $b=\mu_{n,i}$, and $S=\nu_i^2$ yields
\[
\lambda_{\min}(\mathbb E[y_{n,i} y_{n,i}^\top\mid\mathcal F_{n-1}])\ge \frac12 \min\left\{\frac{\nu_i^2}{\nu_i^2 + \mu_{n,i}^2}, \nu_i^2\right\} \ge \frac12 \min\left\{\frac{\nu_*^2}{\nu_*^2 + B^2}, \nu_*^2\right\},
\]
where we use the uniform boundedness of $\mu_{n,i}$.
\end{proof}

\begin{lemma}[Adaptive covariance matrices]\label{lem:spectralboundadaptivecovariancematrix}
Let $(\mathcal F_n)_{n\ge 0}$ be a filtration. Let $(x_n)_{n\ge 1}$ be an $\mathbb R^d$-valued adapted process such that there exist $C_x<\infty$ and $C_M>0$ such that $\norm{x_n}_2\le C_x$ a.s. for all $n\ge 1$ and $\mathbb E[x_n x_n^\top|\mathcal F_{n-1}]\succeq C_M I_d$ a.s. for all $n\ge 1$. Define $S_n = \sum_{m=1}^n x_m x_m^\top$. Then, for any $\delta>0$, $S_n \succeq (C_M - \delta) n I_d$ eventually a.s. As a consequence, for any $x$ with $\norm{x}_2 \le C_x$, $x^\top S_n^{-1} x \le \tfrac{C_x^2}{(C_M - \delta) n}$ eventually a.s.
\end{lemma}
\begin{proof}[Proof of Lemma~\ref{lem:spectralboundadaptivecovariancematrix}]
Fix any deterministic unit vector $v\in\mathbb R^d$. Define $Z_m(v) \triangleq (v^\top x_m)^2 \in \mathcal F_m$ with $Z_m(v) \le C_x^2$. Define martingale difference $\Delta_m(v) \triangleq Z_m(v) - \mathbb E[Z_m(v)|\mathcal F_{m-1}]$. By the strong law we have $\tfrac{1}{n}\sum_{m=1}^n \Delta_m(v) \to 0$ a.s. Since $\mathbb E[Z_m(v)\mid \mathcal F_{m-1}] \ge C_M$ a.s., we have
\[
\liminf_{n\to\infty} \frac{1}{n} v^\top S_n v \ge C_M \quad\text{a.s. for every fixed } v.
\]
We next upgrade the bound to hold uniformly over an $\varepsilon$-net. Fix $\varepsilon\in(0,1)$. There exists a finite set $\mathcal N_\varepsilon\subset \mathbb S^{d-1}$ (the unit sphere) such that $|\mathcal N_\varepsilon|\le (1+2/\varepsilon)^d$ and there exists $v\in\mathcal N_\varepsilon$ with $\norm{u-v}_2\le \varepsilon$ for every $u\in\mathbb S^{d-1}$. For each $v\in\mathcal N_\varepsilon$, there exists an a.s. finite random $N_v$ such that for all $n\ge N_v$, $\frac{1}{n}v^\top S_n v \ge C_M-\frac{\delta}{2}$ a.s. Since $\mathcal N_\varepsilon$ is finite, define $N_0\triangleq \max_{v\in\mathcal N_\varepsilon} N_v < \infty$ a.s., and for all $n\ge N_0$ and all $v\in\mathcal N_\varepsilon$, $\frac{1}{n}v^\top S_n v \ge C_M-\frac{\delta}{2}$ a.s. Now, for any unit vector $u\in\mathbb S^{d-1}$, there exists $v\in\mathcal N_\varepsilon$ such that $\norm{u-v}_2\le \varepsilon$ and
\[
|u^\top S_n u - v^\top S_n v| \le \norm{S_n}_{\text{op}} \norm{uu^\top - vv^\top}_{\text{op}}.
\]
Since
\[
\norm{uu^\top - vv^\top}_{\text{op}} \le \norm{u-v}_2 \norm{u}_2 + \norm{v}_2 \norm{u-v}_2 \le 2\varepsilon
\]
and $\norm{S_n}_{\text{op}} \le n C_x^2$, we have
\[
\frac{1}{n}|u^\top S_n u - v^\top S_n v| \le 2 C_x^2 \varepsilon.
\]
Choose $\varepsilon$ small such that $2\varepsilon C_x^2 \le \delta/2$. Then, for all $n\ge N_0$ and all unit vectors $u\in\mathbb S^{d-1}$, $S_n \succeq (C_M - \delta) n I_d$ and the result follows.
\end{proof}

\begin{lemma}[Non-homogeneous linear recursions]\label{lem:nonhomogeneouslinearrecursion}
Let $(x_n)_{n \ge 1}$ be nonnegative.
\begin{enumerate}[label=(\alph*), leftmargin=2em]
\item If, for some $a, b > 0$, $x_{n+1} \le (1 - \tfrac{a}{n+1}) x_n + \tfrac{b}{n(n+1)}$, then $x_n$ is bounded by: $O(n^{-a})$ if $0 < a < 1$, $O(n^{-1}\log n)$ if $a = 1$, and $O(n^{-1})$ if $a > 1$.
\item If, for some $a, b > 0$ and $r > 1$, $x_{n+1} \le (1 - \tfrac{a}{n+1}) x_n + b\, n^{-r} \log n$, then $x_n$ is bounded by: $O(n^{-(r-1)} \log n)$ if $a > r - 1$, $O(n^{-a} (\log n)^2)$ if $a = r - 1$, and $O(n^{-a})$ if $a < r - 1$.
\end{enumerate}
\end{lemma}
\begin{proof}[Proof of Lemma~\ref{lem:nonhomogeneouslinearrecursion}]
Both parts iterate a recursion of the common form $x_{n+1} \le (1 - \tfrac{a}{n+1}) x_n + f_n$ with nonnegative $f_n$. Iterating from $n_0 = 1$ to $n - 1$ gives
\[
x_n \le x_1 \prod_{m=1}^{n-1} \left(1 - \frac{a}{m+1}\right) + \sum_{k=1}^{n-1} f_k \prod_{m=k+1}^{n-1} \left(1 - \frac{a}{m+1}\right).
\]
The two product factors admit closed-form Gamma-ratio expressions:
\[
\prod_{m=1}^{n-1} \left(1 - \frac{a}{m+1}\right) = \prod_{m=2}^{n} \frac{m - a}{m} = \frac{\Gamma(n - a + 1)}{\Gamma(2 - a)\, \Gamma(n + 1)} \asymp n^{-a},
\]
\[
\prod_{m=k+1}^{n-1} \left(1 - \frac{a}{m+1}\right) = \prod_{m=k+2}^{n} \frac{m - a}{m} = \frac{\Gamma(n - a + 1)\, \Gamma(k + 2)}{\Gamma(k + 2 - a)\, \Gamma(n + 1)} \asymp \left(\frac{k+1}{n}\right)^{a}.
\]
Substituting,
\begin{equation}\label{eq:nonhomogeneous-collapsed}
x_n \le C_1\, n^{-a} + C_2\, n^{-a} \sum_{k=1}^{n-1} (k+1)^{a}\, f_k,
\end{equation}
for constants $C_1, C_2 > 0$. For (a), substituting $f_k = b/(k(k+1))$ into \eqref{eq:nonhomogeneous-collapsed} and using $(k+1)^a f_k \asymp (k+1)^{a-2}$, we have $x_n \le C_1 n^{-a} + C_2' n^{-a} \sum_{k=1}^{n-1} (k+1)^{a-2}$ and the result follows. For (b), substituting $f_k = b\, k^{-r} \log k$ into \eqref{eq:nonhomogeneous-collapsed} and using $(k+1)^a f_k \asymp k^{a - r} \log k$, we have $x_n \le C_1 n^{-a} + C_2'' n^{-a} \sum_{k=1}^{n-1} k^{a-r} \log k$. The result follows via $\int k^{a-r} \log k\, dk$ by parts for $a \ne r - 1$.
\end{proof}

\begin{lemma}[Diagonal-plus-rank-one matrix]\label{lem:rankonediagonalmatrixmineigenvalue}
Let $b \in \mathbb{R}^d$, $S = \mathrm{diag}(s_1, ..., s_d)$ with $s_i > 0$ for all $i$, and $K = n(1, b)^\top (1, b) + \mathrm{diag}(0, S)$. Then, 
$$
\lambda_{\min}(K) \ge \frac{1}{2} \min \left\{\frac{n}{1 + n b^\top S^{-1} b}, \min_{1 \le i \le d} s_i\right\}.
$$
\end{lemma}
\begin{proof}[Proof of Lemma~\ref{lem:rankonediagonalmatrixmineigenvalue}]
Let $x = (t, u) \in \mathbb{R}^{1+d}$. Then, $x^\top K x = n(t + b^\top u)^2 + u^\top S u$ and $\norm{x}_2^2 = t^2 + \norm{u}_2^2$. So, the Rayleigh quotient is
\[
\frac{x^\top K x}{\norm{x}_2^2} = \frac{n(t + b^\top u)^2 + u^\top S u}{t^2 + \norm{u}_2^2}, \qquad \lambda_{\min}(K) = \inf_{x \ne 0} \frac{x^\top K x}{\norm{x}_2^2}.
\] Since $S \succeq s_{\min} I_d$ where $s_{\min} = \min_{1 \le i \le d} s_i$, we have $u^\top S u \ge s_{\min} \norm{u}_2^2$. Write $K = \left(\begin{smallmatrix} A & B^\top \\
B & D\end{smallmatrix}\right)$, where $A = n$, $B = n b$, and $D = S + n b b^\top$. Then,
\[
(t,u) K (t,u)^\top = (u + t D^{-1} B)^\top D (u + t D^{-1} B) + t^2 (A - B^\top D^{-1} B).
\]
Since $S \succ 0$, $D \succ 0$, the first term is non-negative, and $x^\top K x \ge t^2 (A - B^\top D^{-1} B) = t^2 (n - n^2 b^\top (S + n b b^\top)^{-1} b)$. By the Woodbury matrix identity,
\[
(S + n b b^\top)^{-1} = S^{-1} - S^{-1} b \cdot \frac{n}{1 + n b^\top S^{-1} b} \cdot b^\top S^{-1},
\]
\[
b^\top (S + n b b^\top)^{-1} b = b^\top S^{-1} b - b^\top S^{-1} b \cdot \frac{n}{1 + n b^\top S^{-1} b} \cdot b^\top S^{-1} b,
\]
and so $x^\top K x \ge \tfrac{n t^2 }{1 + n b^\top S^{-1} b}$. Combining the two lower bounds, we have $2 x^\top K x \ge \tfrac{n}{1 + n b^\top S^{-1} b} \cdot t^2  + s_{\min} \norm{u}_2^2$. Since $\tfrac{n}{1 + n b^\top S^{-1} b} \ge 0$ and $s_{\min} > 0$, we have $2 x^\top K x \ge \min \{\tfrac{n}{1 + n b^\top S^{-1} b}, s_{\min}\} \norm{x}_2^2$. So, for any $x \ne 0$, $\tfrac{x^\top K x}{\norm{x}_2^2} \ge \tfrac{1}{2} \min \{\tfrac{n}{1 + n b^\top S^{-1} b}, s_{\min}\}$, and taking the infimum gives the desired result.
\end{proof}

\end{document}